\documentclass[11pt]{article}
\usepackage{epsf}
\usepackage{amsmath}
\usepackage{epsfig}
\usepackage{times}
\usepackage{amssymb}
\usepackage{amsthm}
\usepackage{setspace}
\usepackage{cite}

\usepackage{algorithmic}  
\usepackage{algorithm}

\usepackage{shadow}
\usepackage{fancybox}
\usepackage{fancyhdr}

\def\w{{\bf w}}

\def\y{{\bf y}}
\def\v{{\bf v}}
\def\x{{\bf x}}

\def\x{{\mathbf x}}

\def\w{{\bf w}}

\def\v{{\bf v}}
\def\x{{\bf x}}
\def\y{{\bf y}}
\def\z{{\bf z}}

\def\h{{\bf h}}

\def\be{\begin{equation}}
\def\ee{\end{equation}}
\def\ba{\left[\begin{array}}
\def\ea{\end{array}\right]}

\def\w{{\bf w}}

\def\v{{\bf v}}
\def\x{{\bf x}}
\def\y{{\bf y}}
\def\z{{\bf z}}

\def\xtilde{\tilde{\x}}

\def\1{{\bf 1}}

\def\g{{\bf g}}
\def\0{{\bf 0}}

\def\erfinv{\mbox{erfinv}}

\newtheorem{theorem}{Theorem}

\begin{document}

\begin{singlespace}

\title {A problem dependent analysis of SOCP algorithms in noisy compressed sensing
}
\author{
\textsc{Mihailo Stojnic}
\\
\\
{School of Industrial Engineering}\\
{Purdue University, West Lafayette, IN 47907} \\
{e-mail: {\tt mstojnic@purdue.edu}} }
\date{}
\maketitle

\centerline{{\bf Abstract}} \vspace*{0.1in}

Under-determined systems of linear equations with sparse solutions have been the subject of an extensive research in last several years above all due to results of \cite{CRT,CanRomTao06,DonohoPol}. In this paper we will consider \emph{noisy} under-determined linear systems. In a breakthrough \cite{CanRomTao06} it was established that in \emph{noisy} systems for any linear level of under-determinedness there is a linear sparsity that can be \emph{approximately} recovered through an SOCP (second order cone programming) optimization algorithm so that the approximate solution vector is (in an $\ell_2$-norm sense) guaranteed to be no further from the sparse unknown vector than a constant times the noise. In our recent work \cite{StojnicGenSocp10} we established an alternative framework that can be used for statistical performance analysis of the SOCP algorithms. To demonstrate how the framework works we then showed in \cite{StojnicGenSocp10} how one can use it to precisely characterize the \emph{generic} (worst-case) performance of the SOCP. In this paper we present a different set of results that can be obtained through the framework of \cite{StojnicGenSocp10}. The results will relate to \emph{problem dependent} performance analysis of SOCP's. We will consider specific types of unknown sparse vectors and characterize the SOCP performance when used for recovery of such vectors. We will also show that our theoretical predictions are in a solid agreement with the results one can get through numerical simulations.

\vspace*{0.25in} \noindent {\bf Index Terms: Noisy systems of linear equations; SOCP;
$\ell_1$-optimization; compressed sensing}.

\end{singlespace}

\section{Introduction}
\label{sec:back}

In recent years there has been an enormous interest in studying under-determined systems of linear equations with sparse solutions. With potential applications ranging from high-dimensional geometry, image
reconstruction, single-pixel camera design, decoding of linear
codes, channel estimation in wireless communications, to machine
learning, data-streaming algorithms, DNA micro-arrays,
magneto-encephalography etc. (see, e.g. \cite{ECicm,DDTLSKB,CT,JRimaging,BCDH08,CRchannel,VPH,PVMHjournal,WM08,Olgica,RFPrank,MBPSZ08,RS08,StojnicCSetamBlock09} and references therein) studying these systems seems to be of substantial theoretical/practical importance in a variety of different area. In this paper we study mathematical aspects of under-determined systems and put an emphasis on theoretical analysis of particular algorithms used for solving them.

In its simplest form solving an under-determined system of linear equations amounts to finding a, say, $k$-sparse $\x$ such
that
\begin{equation}
A\x=\y \label{eq:system}
\end{equation}
where $A$ is an $m\times n$ ($m<n$) matrix and $\y$ is
an $m\times 1$ vector (see Figure
\ref{fig:model}; here and in the rest of the paper, under $k$-sparse vector we assume a vector that has at most $k$ nonzero
components). Of course, the assumption will be that such an $\x$ exists. To make writing in the rest of the paper easier, we will assume the
so-called \emph{linear} regime, i.e. we will assume that $k=\beta n$
and that the number of equations is $m=\alpha n$ where
$\alpha$ and $\beta$ are constants independent of $n$ (more
on the non-linear regime, i.e. on the regime when $m$ is larger than
linearly proportional to $k$ can be found in e.g.
\cite{CoMu05,GiStTrVe06,GiStTrVe07}).
\begin{figure}[htb]
\centering
\centerline{\epsfig{figure=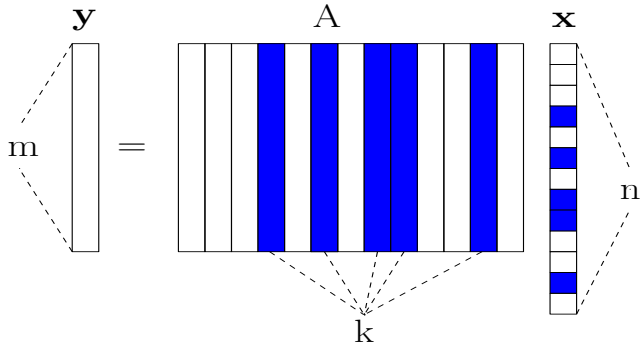,width=9cm,height=4.5cm}}
\caption{Model of a linear system; vector $\x$ is $k$-sparse}
\label{fig:model}
\end{figure}

Clearly, if $k\leq \frac{m}{2}$ the solution is unique and can be found through an exhaustive search. However, in the linear regime that we have just assumed above (and will consider throughout the paper) the exhaustive search is clearly of exponential complexity. Instead one can of course design algorithms of much lower complexity while sacrificing on the recovery abilities, i.e. on the recoverable size of the nonzero portion of vector $\x$. Various algorithms have been introduced and analyzed in recent years throughout the literature from those that relate to the parallel design of matrix $A$ and the recovery algorithms (see, e.g. \cite{FHicassp,Tarokh,MaVe05,XHexpander,JXHC08,InRu08}) to those that assume only the design of the recovery algorithms (see, e.g. \cite{JATGomp,JAT,NeVe07,DTDSomp,NT08,DaiMil08,CanRomTao06,DonohoPol}). If one restricts to the algorithms of polynomial complexity and allows for the design of $A$ then the results from \cite{FHicassp,Tarokh,MaVe05} that guarantee recovery of \emph{any} $k$-sparse $\x$ in
(\ref{eq:system}) for any $0<\alpha\leq 1$ and any
$\beta\leq\frac{\alpha}{2}$ are essentially optimal. On the other hand, if one restricts to the algorithms of polynomial complexity and does not allow for the parallel design of $A$ then the results of \cite{CanRomTao06,DonohoPol} established that for any $\alpha>0$ there still exists a $\beta>0$ such that any $k=\beta n$-sparse $\x$ in (\ref{eq:system}) can be recovered via a polynomial time \emph{Basis pursuit} (BP) algorithm.

Since the BP algorithm is fairly closely related to what we will be presenting later in the paper
we will now briefly introduce it (we will often refer to it as the $\ell_1$-optimization concept; a slight modification/adaptation of it will actually be the main topic of this paper). Variations of the standard $\ell_1$-optimization from e.g.
\cite{CWBreweighted,SChretien08,SaZh08} as well as those from \cite{SCY08,FL08,GN03} related to $\ell_q$-optimization, $0<q<1$,
are possible as well; moreover they can all be incorporated in what we will present below. The $\ell_1$-optimization concept suggests that one can maybe find the $k$-sparse $\x$ in
(\ref{eq:system}) by solving the following $\ell_1$-norm minimization problem
\begin{eqnarray}
\mbox{min} & & \|\x\|_{1}\nonumber \\
\mbox{subject to} & & A\x=\y. \label{eq:l1}
\end{eqnarray}
As mentioned above the results from \cite{CanRomTao06,DonohoPol} were instrumental in theoretical characterization of (\ref{eq:l1}) and its popularization in sparse recovery and even more so in generating an unprecedented interest in sparse recovery algorithms. The main reason is of course the quality of the results achieved in \cite{CanRomTao06,DonohoPol}. Namely, \cite{CanRomTao06} established that for any $\alpha>0$ there is a $\beta>0$
such that the solution of (\ref{eq:l1}) is the $k=\beta n$-sparse $\x$ in (\ref{eq:system}). In a statistical and large dimensional context in \cite{DonohoPol} and later in \cite{StojnicCSetam09,StojnicUpper10} for any given value of $\alpha$ the exact value of the maximum possible $\beta$ was determined.

The above sparse recovery scenario is in a sense idealistic. Namely, it assumes that $\y$ in (\ref{eq:l1}) was obtained through (\ref{eq:system}). On other hand in many applications only a \emph{noisy} version of $A\x$ may be available for $\y$ (this is especially so in measuring type of applications) see, e.g. \cite{CanRomTao06,HN,W} (another somewhat related version of ``imperfect" linear systems are the under-determined systems with the so-called approximately sparse solutions; more in this direction can be found in e.g. \cite{CanRomTao06,SXHapp}). When that happens one has the following equivalent to (\ref{eq:system}) (see, Figure \ref{fig:modelnoise})
\begin{equation}
\y=A\x+\v, \label{eq:systemnoise}
\end{equation}
where $\v$ is an $m\times 1$ so-called noise vector (the so-called ideal case presented above is of course a special case of the noisy one given in (\ref{eq:systemnoise})).
\begin{figure}[htb]
\centering
\centerline{\epsfig{figure=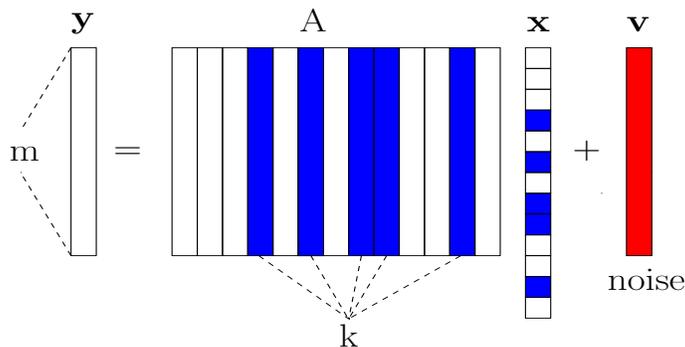,width=9cm,height=4.5cm}}
\caption{Model of a linear system; vector $\x$ is $k$-sparse}
\label{fig:modelnoise}
\end{figure}
Finding the $k$-sparse $\x$ in (\ref{eq:systemnoise}) is now incredibly hard. In fact it is pretty much impossible. Basically, one is looking for a $k$-sparse $\x$ such that (\ref{eq:systemnoise}) holds and on top of that $\v$ is unknown. Although the problem is hard there are various heuristics throughout the literature that one can use to solve it approximately. Majority of these heuristics are based on appropriate generalizations of the corresponding algorithms one would use in the noiseless case. Thinking along the same lines as in the noiseless case one can distinguish two scenarios depending on the availability of the freedom to choose/design $A$. If one has the freedom to design $A$ then one can adapt the corresponding noiseless algorithms to the noisy scenario as well (more on this can be found in e.g. \cite{BGIKS}). However, in this paper we mostly focus on the scenario where one has no control over $A$. In such a scenario one can again make a parallel to the noiseless case and look at e.g. CoSAMP algorithm from \cite{NT08} or Subspace pursuit from \cite{DaiMil09}. Essentially, in a statistical context, these algorithms can provably recover a linear sparsity while maintaining the approximation error proportional to the norm-2 of the noise vector which is pretty much a benchmark of what is currently known. These algorithms are in a way perfected noisy generalizations of the so-called \emph{Orthogonal matching pursuit} (OMP) algorithms.

On the other hand in this paper we will focus on generalizations of BP that can handle the noisy case. To introduce a bit or tractability in finding the $k$-sparse $\x$ in (\ref{eq:systemnoise}) one usually assumes certain amount of knowledge about either $\x$ or $\v$. As far as the tractability assumptions on $\v$ are concerned one typically (and possibly fairly reasonably in applications of interest) assumes that $\|\v\|_2$ is bounded (or highly likely to be bounded) from above by a certain known quantity. The following second-order cone programming (SOCP) analogue to (or say noisy generalization of) (\ref{eq:l1}) is one of the approaches that utilizes such an assumption (more on this approach and its variations can be found in e.g. \cite{CanRomTao06})
\begin{eqnarray}
\min_{\x} & & \|\x\|_1\nonumber \\
\mbox{subject to} & & \|\y-A\x\|_2\leq r_{socp} \label{eq:socp}
\end{eqnarray}
where $r_{socp}$ is a quantity such that $\|\v\|_2\leq r_{socp}$ (or $r_{socp}$ is a quantity such that $\|\v\|_2\leq r_{socp}$ is say highly likely). For example, in \cite{CanRomTao06} a statistical context is assumed and based on the statistics of $\v$, $r_{socp}$ was chosen such that $\|\v\|_2\leq r_{socp}$ happens with overwhelming probability (as usual, under overwhelming probability we in this paper assume
a probability that is no more than a number exponentially decaying in $n$ away from $1$). Given that (\ref{eq:socp}) is now among few almost standard choices when it comes to finding an approximation to the $k$-sparse $\x$ in (\ref{eq:systemnoise}), the literature on its properties is vast (see, e.g. \cite{CanRomTao06,DonElaTem06,Tropp06,HN} and references therein). We here briefly mention only what we consider to be the most influential work on this topic in recent years. Namely, in \cite{CanRomTao06} the authors analyzed performance of (\ref{eq:socp}) and showed a result similar in flavor to the one that holds in the ideal - noiseless - case. In a nutshell the following was shown in \cite{CanRomTao06}: let $\x$ be a $\beta n$-sparse vector such that (\ref{eq:systemnoise}) holds and let $\x_{socp}$ be the solution of (\ref{eq:socp}). Then
\begin{equation}
\|\x_{socp}-\x\|_2\leq C r_{socp}\label{eq:CRTsocp}
\end{equation}
where $\beta$ is a constant independent of $n$ and $C$ is a constant independent of $n$ and of course dependent on $\alpha$ and $\beta$. This result in a sense establishes a noisy equivalent to the fact that a linear sparsity can be recovered from an under-determined system of linear equations. In an informal language, it states that a linear sparsity can be \emph{approximately} recovered in polynomial time from a noisy under-determined system with the norm of the recovery error guaranteed to be within a constant multiple of the noise norm (as mentioned above, the same was also established later in \cite{NT08} for CoSAMP and in \cite{DaiMil09} for Subspace pursuit). Establishing such a result is, of course, a feat in its own class, not only because of its technical contribution but even more so because of the amount of interest that it generated in the field.

In our recent work \cite{StojnicGenSocp10} we designed a framework for performance analysis of the SOCP algorithm from (\ref{eq:socp}). We then went further in \cite{StojnicGenSocp10} and showed how the framework practically works through a precise characterization of the \emph{generic} (worst-case) performance of the SOCP from (\ref{eq:socp}). In this paper we will again focus on the general framework developed in \cite{StojnicGenSocp10}. This time though we will focus on the \emph{problem dependent} performance analysis of the SOCP. In other words, we will consider specific types of unknown sparse vectors $\x$ in (\ref{eq:systemnoise}) and provide a performance analysis of the SOCP when applied for recovery of such vectors.

Before going into the details of the SOCP approach we should also mention that the SOCP algorithms are of course not the only possible generalizations (adaptations) of $\ell_1$ optimization to the noisy case. For example, LASSO algorithms (more on these algorithms can be found in e.g. \cite{Tibsh96,CheDon95,CheDonSau98,BunTsyWeg07,vandeGeer08,MeinYu09} as well as in recent developments \cite{DonMalMon10,BayMon10lasso,StojnicGenLasso10}) are a very successful alternative. In our recent work \cite{StojnicGenLasso10} we established a nice connection between some of the algorithms from the LASSO group and certain SOCP algorithms and showed that with respect to certain performance measure they could be equivalent. Besides the LASSO algorithms the so-called Dantzig selector introduced in \cite{CanTao07} is another alternative to the SOCP algorithms that is often encountered in the literature (more on the Dantzig selector as well as on its relation to the LASSO or SOCP algorithms can be found in e.g. \cite{MeiRocYu07,BicRitTsy09,FriSau07,EfrHatTib07,AsiRom10,JamRadLv09,Koltch09}). Depending on the scenarios they are to be applied in each of these algorithms can have certain advantages/disadvantages over the other ones. A simple (general) characterization of these advantages/disadvanatges does not seem easy to us. In a rather informal language one could say that LASSO and SOCP are expected to perform better (i.e. to provide a solution vector that is under various metrics closer to the unknown one) in a larger set of different scenarios but as quadratic programs could be slower than the Dantzig selector which happens to be a linear program. Of course, whether LASSO or SOCP algorithms are indeed going to be slower or not or how much larger would be a set of different scenarios where they are expected to perform better are interesting/important questions. While clearly of interest answering these questions certainly goes way beyond the scope of the present paper and we will not pursue it here any further.

Before we proceed with the exposition we briefly summarize the organization of the rest of the paper. In Section
\ref{sec:unsigned}, we recall on a set of powerful results presented in \cite{StojnicGenSocp10} and discuss further how they can be utilized to analyze the problem dependent performance of the SOCP from (\ref{eq:socp}). The results that we will present in Section \ref{sec:unsigned} will relate to the so-called general sparse signals $\xtilde$. In Section \ref{sec:signed} we will show how these results from Section \ref{sec:unsigned} that relate to the general sparse vectors $\xtilde$ can be specialized further so they cover the so-called signed vectors $\x$. Finally, in Section \ref{sec:discuss} we will discuss obtained results.

\section{SOCP's problem dependent performance -- general $\x$} \label{sec:unsigned}

In this section we first recall on the basic properties of the statistical SOCP's performance analysis framework developed in \cite{StojnicGenSocp10}. We will then show how the framework can be used to characterize the SOCP from (\ref{eq:socp}) in certain situations when the SOCP's performance  substantially depends on $\xtilde$.

\subsection{Basic properties of the SOCP's framework} \label{sec:unsignedbasic}

Before proceeding further we will now first state major assumptions that will be in place throughout the paper (clearly, since we will be utilizing the framework of  \cite{StojnicGenSocp10} a majority of these assumptions was already present in \cite{StojnicGenSocp10}). Namely, as mentioned above, we will consider noisy under-determined systems of linear equations. The systems will be defined by a random matrix $A$ where the elements of $A$ are i.i.d. standard normal random variables. Also, we will assume that the elements of $\v$ are i.i.d. Gaussian random variables with zero mean and variance $\sigma$. $\xtilde$ will be assumed to be the original $\x$ in (\ref{eq:systemnoise}) that we are trying to recover (or a bit more precisely approximately recover). We will assume that $\xtilde$ is \emph{any} $k$-sparse vector with a given fixed location of its nonzero elements and a given fixed combination of their signs. Due to assumed statistics the analysis (and the performance of (\ref{eq:socp})) will clearly be irrelevant with respect to what particular location and what particular combination of signs of nonzero elements are chosen. We will therefore for the simplicity of the exposition and without loss of generality assume that the components $\x_{1},\x_{2},\dots,\x_{n-k}$ of $\x$ are equal to zero and the components $\x_{n-k+1},\x_{n-k+2},\dots,\x_n$ of $\x$ are greater than or equal to zero. Moreover, throughout the paper we will call such an $\x$ $k$-sparse and positive. In a more formal way we will set
\begin{eqnarray}
& & \xtilde_1=\xtilde_2 =  \dots=\xtilde_{n-k}=0\nonumber \\
& & \xtilde_{n-k+1}\geq 0,  \xtilde_{n-k+1}\geq 0, \dots, \xtilde_{n}\geq 0.\label{eq:xtildedef}
\end{eqnarray}
We also now take the opportunity to point out a rather obvious detail. Namely, the fact that $\xtilde$ is positive is assumed for the purpose of the analysis. However, this fact is not known \emph{a priori} and is not available to the solving algorithm (this will of course change in Section \ref{sec:signed}).

Before proceeding further we will introduce a few definitions that will be useful in formalizing/presenting our results as well as in conducting the entire analysis. Following what was done in \cite{StojnicGenSocp10} let us define the optimal value of a slightly changed objective of (\ref{eq:socp}) in the following way
\begin{eqnarray}
f_{obj}(\sigma,\xtilde,A,\v,r_{socp})=\min_{\x} & & \|\x\|_1-\|\xtilde\|_1\nonumber \\
\mbox{subject to} & & \|\y-A\x\|_2\leq r_{socp}. \label{eq:objlassol1}
\end{eqnarray}
To make writing easier we will instead of $f_{obj}(\sigma,\xtilde,A,\v,r_{socp})$ write just $f_{obj}$. A similar convention will be applied to few other functions  throughout the paper. On many occasions, though, (especially where we deem it as substantial to the presentation) we will also keep all (or a majority of) arguments of the corresponding functions.

Also let $\x_{socp}$ be the solution of (\ref{eq:socp}) (or the solution of (\ref{eq:objlassol1})) and let $\w_{socp}\in R^n$  be the so-called error vector defined in the following way
\begin{equation}
\w_{socp}=\x_{socp}-\xtilde.\label{eq:xhatdef}
\end{equation}
Our main goal in this paper will then boil down to various characterizations of $\w_{socp}$ and $f_{obj}$. Throughout the paper we will heavily rely on the following theorem from \cite{StojnicGenSocp10} that provides a general characterization of $\w_{socp}$ and $f_{obj}$.

\begin{theorem}(\cite{StojnicGenSocp10} --- SOCP's performance characterization)
Let $\v$ be an $n\times 1$ vector of i.i.d. zero-mean variance $\sigma^2$ Gaussian random variables and let $A$ be an $m\times n$ matrix of i.i.d. standard normal random variables. Further, let $\g$ and $\h$ be $m\times 1$ and $n\times 1$ vectors of i.i.d. standard normals, respectively and let $\z$ be $n\times 1$ vector of all ones. Consider a $k$-sparse $\xtilde$ defined in (\ref{eq:xtildedef}) and a $\y$ defined in (\ref{eq:systemnoise}) for $\x=\xtilde$. Let the solution of (\ref{eq:socp}) be $\x_{socp}$ and let the so-called error vector of the SOCP from (\ref{eq:socp}) be $\w_{socp}=\x_{socp}-\xtilde$. Let $r_{socp}$ in (\ref{eq:socp}) be a positive scalar. Let $n$ be large and let constants $\alpha=\frac{m}{n}$ and $\beta_w=\frac{k}{n}$ be below the following so-called \emph{fundamental} characterization of $\ell_1$ optimization
\begin{equation}
(1-\beta_w)\frac{\sqrt{\frac{2}{\pi}}e^{-(\erfinv(\frac{1-\alpha_w}{1-\beta_w}))^2}}{\alpha_w}-\sqrt{2}\erfinv (\frac{1-\alpha_w}{1-\beta_w})=0.
\label{eq:fundl1}
\end{equation}
Consider the following optimization problem:
\begin{eqnarray}
\xi_{prim}(\sigma,\g,\h,\xtilde,r_{socp})=\max_{\nu,\lambda} & & \sigma\sqrt{\|\g\|_2^2\nu^2-\|\nu\h+\z-\lambda\|_2^2} -\sum_{i=n-k+1}^{n}\lambda_i\xtilde_i-\nu r_{socp}\nonumber \\
\mbox{subject to}
& & \nu\geq 0\nonumber \\
& & 0 \leq\lambda_i\leq 2,1\leq i\leq n.\label{eq:mainlasso1}
\end{eqnarray}
Let $\hat{\nu}$ and $\hat{\lambda}$ be the solution of (\ref{eq:mainlasso1}). Set
\begin{equation}
\|\hat{\w}\|_2=\sigma\frac{\|\hat{\nu}\h+\z-\hat{\lambda}\|_2}
{\sqrt{\|\g\|_2^2\hat{\nu}^2-\|\hat{\nu}\h+\z-\hat{\lambda}\|_2^2}}.\label{eq:mainlasso2}
\end{equation}
Then:
\begin{multline}
P(\|\xtilde+\w_{socp}\|_1-\|\xtilde\|_1
\in (E\xi_{prim}(\sigma,\g,\h,\xtilde,r_{socp}))-\epsilon_1^{(socp)}|E\xi_{prim}(\sigma,\g,\h,\xtilde,r_{socp}))|,\\
E\xi_{prim}(\sigma,\g,\h,\xtilde,r_{socp}))+\epsilon_1^{(socp)}|E\xi_{prim}(\sigma,\g,\h,\xtilde,r_{socp}))|)=1-e^{-\epsilon_2^{(socp)}n}\label{eq:mainlasso3}
\end{multline}
and
\begin{equation}
P((1-\epsilon_1^{(socp)})E\|\hat{\w}\|_2\leq \|\w_{socp}\|_2
\leq (1+\epsilon_1^{(socp)})E\|\hat{\w}\|_2) =1-e^{-\epsilon_2^{(socp)}n},\label{eq:mainlasso4}
\end{equation}
where $\epsilon_1^{(socp)}>0$ is an arbitrarily small constant and $\epsilon_2^{(socp)}$ is a constant dependent on $\epsilon_1^{(socp)}$ and $\sigma$ but independent of $n$.
\label{thm:mainlasso}
\end{theorem}
\begin{proof}
Presented in \cite{StojnicGenSocp10}.
\end{proof}

\noindent \textbf{Remark:} A pair $(\alpha,\beta_w)$ lies below the fundamental characterization (\ref{eq:fundl1}) if $\alpha>\alpha_w$ and $\alpha_w$ and $\beta_w$ are such that (\ref{eq:fundl1}) holds.

\subsection{Problem dependent properties of the framework} \label{sec:probdepframe}

To facilitate the exposition that will follow we similarly to what was done in \cite{StojnicCSetam09,StojnicEquiv10,StojnicGenSocp10} set
\begin{equation}
\bar{\h}=[|\h|_{(1)}^{(1)},|\h|_{(2)}^{(2)},\dots,|\h|_{(n-k)}^{(n-k)},\h_{n-k+1}^{(k)},\h_{n-k+2}^{(k-1)},\dots,\h_n^{(1)}]^T,\label{eq:defhbar}
\end{equation}
where $[|\h|_{(1)}^{(1)},|\h|_{(2)}^{(2)},\dots,|\h|_{(n-k)}^{(n-k)}]$ are the magnitudes of $[\h_{1},\h_{2},\dots,\h_{n-k}]$ sorted in increasing order and
$[\h_{n-k+1}^{(k)},\h_{n-k+2}^{(k-1)},\dots,\h_n^{(1)}]$ are the elements of $[-\h_{n-k+1},-\h_{n-k+2},\dots,-\h_n]$ sorted in decreasing order (possible ties in the sorting processes are of course broken arbitrarily). One can then rewrite the optimization problem from (\ref{eq:mainlasso1}) in the following way
\begin{eqnarray}
\xi_{prim}(\sigma,\g,\h,\xtilde,r_{socp})=\max_{\nu,\lambda} & & \sigma\sqrt{\|\g\|_2^2\nu^2-\|\nu\bar{\h}-\z+\lambda\|_2^2} -\sum_{i=n-k+1}^{n}\lambda_i\xtilde_i-\nu r_{socp}\nonumber \\
\mbox{subject to}
& & \nu\geq 0\nonumber \\
& & 0 \leq\lambda_i\leq 1,1\leq i\leq n-k\nonumber \\
& & 0 \leq \lambda_i\leq 2, n-k+1\leq i\leq n.\label{eq:mainlasso3}
\end{eqnarray}
In what follows we will restrict our attention to a specific class of unknown vectors $\xtilde$. Namely, we will consider vectors $\xtilde$ that have amplitude of the nonzero components equal to $x_{mag}$. In the noiseless case these problem instances are typically the hardest to solve (at least as long as one uses the $\ell_1$ optimization from (\ref{eq:l1})). We will again emphasize that the fact that magnitudes of the nonzero elements of $\xtilde$ are $x_{mag}$ is not known a priori and can not be used in the solving algorithm (i.e. one can not add constraints that would exploit this knowledge in optimization problem (\ref{eq:socp})). It is just that we will consider how the SOCP from (\ref{eq:socp}) behaves when used to solve problem instances generated by such an $\xtilde$. Also, for such an $\xtilde$ (\ref{eq:mainlasso3}) can be rewritten in the following way
\begin{eqnarray}
\xi_{prim}^{(dep)}(\sigma,\g,\h,x_{mag},r_{socp})=\max_{\nu,\lambda} & & \sigma\sqrt{\|\g\|_2^2\nu^2-\|\nu\bar{\h}-\z+\lambda\|_2^2} -x_{mag}\sum_{i=n-k+1}^{n}\lambda_i-\nu r_{socp}\nonumber \\
\mbox{subject to}
& & \nu\geq 0\nonumber \\
& & 0 \leq\lambda_i\leq 1,1\leq i\leq n-k\nonumber \\
& & 0 \leq \lambda_i\leq 2, n-k+1\leq i\leq n.\label{eq:mainlasso3ver}
\end{eqnarray}
Now, let $\nu_{dep}$ and $\lambda^{(dep)}$ be the solution of (\ref{eq:mainlasso3ver}). Then analogously to (\ref{eq:mainlasso2}) we can set
\begin{equation}
\|\w_{dep}\|_2=\sigma\frac{\|\nu_{dep}\bar{\h}-\z+\lambda^{(dep)}\|_2}
{\sqrt{\|\g\|_2^2\nu_{dep}^2-\|\nu_{dep}\bar{\h}-\z+\lambda^{(dep)}\|_2^2}},\label{eq:mainlasso4}
\end{equation}
In what follows we will determine $\|\w_{dep}\|_2$ and $\xi_{prim}^{(dep)}(\sigma,\g,\h,x_{mag},r_{socp})$ or more precisely their concentrating points
$E\|\w_{dep}\|_2$ and $E\xi_{prim}^{(dep)}(\sigma,\g,\h,x_{mag},r_{socp})$. All other parameters such as $\nu_{dep}$, $\lambda^{(dep)}$ can (and some of them will) be computed through the framework as well. We do however mention right here that what we present below assumes a fair share of familiarity with the techniques introduced in our earlier papers \cite{StojnicCSetam09,StojnicGenLasso10,StojnicGenSocp10}. To shorten the exposition we will skip many details presented in those papers and present only the key differences.

We proceed by following the line of thought presented in \cite{StojnicCSetam09,StojnicGenSocp10}. Since $\lambda^{(dep)}$ is the solution of (\ref{eq:mainlasso3ver}) there will be parameters $c_{1}$, $c_{2}$, and $c_{3}$ such that $$\lambda^{(dep)}=[\lambda_1^{(dep)},\lambda_2^{(dep)},\dots,\lambda_{c_{1}}^{(dep)},0,0,\dots,0,\lambda_{c_{2}+1}^{(dep)},\lambda_{c_{2}+2}^{(dep)},
\dots,\lambda_{n-c_{3}}^{(dep)},2,2,\dots,2]$$ and obviously $c_{1}\leq n-k$, $n-k\leq c_{2}\leq n$, and $0\leq c_{3}\leq n-c_{2}$. At this point let us assume that these parameters are known and fixed. Then following \cite{StojnicCSetam09,StojnicGenSocp10} the optimization problem from (\ref{eq:mainlasso3ver}) can be rewritten in the following way
\begin{eqnarray}
\hspace{-.3in}\xi_{prim}(\sigma,\g,\h,\xtilde,r_{socp})=\max_{\nu,\lambda_{c_{1}+1:n}} & & \sigma\sqrt{\|\g\|_2^2\nu^2-\|\nu\bar{\h}_{c_{1}+1:n}-\z_{c_{1}+1:n}+\lambda_{c_{1}+1:n}\|_2^2} -x_{mag}\sum_{i=n-k+1}^{n}\lambda_i-\nu r_{socp}\nonumber \\
\mbox{subject to}
& & \nu\geq 0\nonumber \\
& & 0 \leq\lambda_i\leq 2, c_{2}\leq i\leq n-c_{3}\nonumber \\
& & \lambda_i=2, n-c_{3}+1\leq i\leq n\nonumber \\
& & \lambda_i=0, c_{1}+1\leq i\leq n-k.\label{eq:mainlasso5}
\end{eqnarray}
To make writing of what will follow somewhat easier we set
\begin{equation}
\xi^{(obj)}=\sigma\sqrt{\|\g\|_2^2\nu^2-\|\nu\bar{\h}_{c_{1}+1:n}-\z_{c_{1}+1:n}+\lambda_{c_{1}+1:n}\|_2^2} -x_{mag}\sum_{i=n-k+1}^{n}\lambda_i-\nu r_{socp}.
\label{eq:defxiobj}
\end{equation}
We then proceed by solving the optimization in (\ref{eq:mainlasso5}) over $\nu$ and $\lambda_{c_{1}+1:n}$. To do so we first look at the derivatives with respect to $\lambda_i,c_{2}+1\leq i\leq n-c_{3}$, of the objective in (\ref{eq:mainlasso5}). Computing the derivatives and equalling them to zero gives
\begin{eqnarray}
\hspace{-.3in}& & \frac{d \xi^{(obj)}}{d \lambda_i}  =  0, c_{2}+1\leq i\leq n-c_{3}\nonumber \\
&\iff &  \sigma \frac{-(\nu\bar{\h}_i-\z_i+\lambda_i)}
{\sqrt{\|\g\|_2^2\nu^2-\|\nu\bar{\h}_{c_{1}+1:n}-\z_{c_{1}+1:n}+\lambda_{c_{1}+1:n}\|_2^2}}-x_{mag} =  0, c_{2}+1\leq i\leq n-c_{3}\nonumber \\
& \iff &  \lambda_i-\z_i+
\nu\bar{\h}_{i}=-\frac{x_{mag}}{\sigma}\sqrt{\|\g\|_2^2\nu^2-\|\nu\bar{\h}_{c_{1}+1:n}-\z_{c_{1}+1:n}+\lambda_{c_{1}+1:n}\|_2^2}, c_{2}+1\leq i\leq n-c_{3}\nonumber \\
& \iff &  \lambda_i=-\frac{x_{mag}}{\sigma}\sqrt{\|\g\|_2^2\nu^2-\|\nu\bar{\h}_{c_{1}+1:n}-\z_{c_{1}+1:n}+\lambda_{c_{1}+1:n}\|_2^2}+\z_i-
\nu\bar{\h}_{i}, c_{2}+1\leq i\leq n-c_{3}.\nonumber \\\label{eq:mainlasso6}
\end{eqnarray}
From the second to last line in the above equation one then has
\begin{multline}
(\lambda_i-\z_i+\nu\bar{\h}_{i})^2=\frac{x_{mag}^2}{\sigma^2}(\|\g\|_2^2\nu^2-\|\nu\bar{\h}_{c_{1}+1:c_{2}}-\z_{c_{1}+1:c_{2}}
+\lambda_{c_{1}+1:c_{2}}\|_2^2\\-\|\nu\bar{\h}_{c_{2}+1:n-c_{3}}-\z_{c_{2}+1:n-c_{3}}
\|_2^2-\|\nu\bar{\h}_{n-c_{3}+1:n}+\z_{n-c_{3}+1:n}\|_2^2)
\end{multline}
and after an easy algebraic transformation
\begin{equation}
\hspace{-.3in}(\lambda_i-\z_i+\nu\bar{\h}_{i})^2=\frac{x_{mag}^2}{\sigma^2+(n-c_{2}-c_{3})x_{mag}^2}(\|\g\|_2^2\nu^2-\|\nu\bar{\h}_{c_{1}+1:c_{2}}-\z_{c_{1}+1:c_{2}}
\|_2^2-\|\nu\bar{\h}_{n-c_{3}+1:n}+\z_{n-c_{3}+1:n}\|_2^2).\label{eq:mainlasso7}
\end{equation}
Using (\ref{eq:mainlasso7}) we further have
\begin{multline}
\sqrt{\|\g\|_2^2\nu^2-\|\nu\bar{\h}_{c_{1}+1:n}-\z_{c_{1}+1:n}
+\lambda_{c_{1}+1:n}\|_2^2}\\
\hspace{-.4in}=\sqrt{\|\g\|_2^2\nu^2-\|\nu\bar{\h}_{c_{1}+1:c_{2}}-\z_{c_{1}+1:c_{2}}\|_2^2-\|\nu\bar{\h}_{c_{2}+1:n-c_{3}}-\z_{c_{2}+1:n-c_{3}}
+\lambda_{c_{2}+1:n-c_{3}}\|_2^2-\|\nu\bar{\h}_{n-c_{3}+1:n}+\z_{n-c_{3}+1:n}\|_2^2}\\
=\frac{\sigma}{\sqrt{\sigma^2+(n-c_{3}-c_{2})x_{mag}^2}}\sqrt{\|\g\|_2^2\nu^2-\|\nu\bar{\h}_{c_{1}+1:c_{2}}-\z_{c_{1}+1:c_{2}}\|_2^2
-\|\nu\bar{\h}_{n-c_{3}+1:n}+\z_{n-c_{3}+1:n}\|_2^2}.\label{eq:mainlasso8}
\end{multline}
Plugging the value for $\lambda_i$ from (\ref{eq:mainlasso5}) in (\ref{eq:defxiobj}) gives
\begin{eqnarray}
\xi^{(obj)} & = & \sigma\sqrt{\|\g\|_2^2\nu^2-\|\nu\bar{\h}_{c_{1}+1:n}-\z_{c_{1}+1:n}+\lambda_{c_{1}+1:n}\|_2^2} -x_{mag}\sum_{i=n-k+1}^{n}\lambda_i-\nu r_{socp}\nonumber \\
& = & \frac{\sigma^2+(n-c_{3}-c_{2})x_{mag}^2}{\sigma}\sqrt{\|\g\|_2^2\nu^2-\|\nu\bar{\h}_{c_{1}+1:n}-\z_{c_{1}+1:n}+\lambda_{c_{1}+1:n}\|_2^2}\nonumber \\
& - & x_{mag}(n-c_{3}-c_{2})+\nu x_{mag}\sum_{i=c_2+1}^{n-c_{3}}\bar{\h}_i-2c_{3}x_{mag}-\nu r_{socp}.
\label{eq:mainlasso9}
\end{eqnarray}
Combining (\ref{eq:mainlasso8}) and (\ref{eq:mainlasso9}) we finally obtain
\begin{eqnarray}
\xi^{(obj)}
& = & \sqrt{\sigma^2+(n-c_{3}-c_{2})x_{mag}^2}
\sqrt{\|\g\|_2^2\nu^2-\|\nu\bar{\h}_{c_{1}+1:c_{2}}-\z_{c_{1}+1:c_{2}}\|_2^2
-\|\nu\bar{\h}_{n-c_{3}+1:n}+\z_{n-c_{3}+1:n}\|_2^2}\nonumber \\
& - & \nu(r_{socp}- x_{mag}\sum_{i=c_2+1}^{n-c_{3}}\bar{\h}_i)-x_{mag}(n-c_{3}-c_{2})-2c_{3}x_{mag}.
\label{eq:mainlasso10}
\end{eqnarray}
Equalling the derivative of $\xi^{(obj)}$ with respect to $\nu$ to zero further gives
\begin{eqnarray}
\hspace{-.3in}& & \frac{d \xi^{(obj)}}{d \nu}  =  0\nonumber \\
&\iff &
\frac{\nu(\|\g\|_2^2-\sum_{i=c_1+1}^{c_2}\bar{\h}_i^2-\sum_{i=n-c_3+1}^{n}\bar{\h}_i^2)+\bar{\h}_{c_{1}+1:c_{2}}^T\z_{c_{1}+1:c_{2}}
-\bar{\h}_{n-c_{3}+1:n}^T\z_{n-c_{3}+1:n}}
{(\sqrt{\sigma^2+(n-c_{3}-c_{2})x_{mag}^2})^{-1}\sqrt{\|\g\|_2^2\nu^2-\|\nu\bar{\h}_{c_{1}+1:c_{2}}-\z_{c_{1}+1:c_{2}}\|_2^2
-\|\nu\bar{\h}_{n-c_{3}+1:n}+\z_{n-c_{3}+1:n}\|_2^2}}\nonumber\\
& & -(r_{socp}- x_{mag}\sum_{i=c_2+1}^{n-c_{3}}\bar{\h}_i) =0.\nonumber \\\label{eq:mainlasso11}
\end{eqnarray}
Let
\begin{eqnarray}
 s_{dep} & = & \bar{\h}_{c_{1}+1:c_{2}}^T\z_{c_{1}+1:c_{2}}
-\bar{\h}_{n-c_{3}+1:n}^T\z_{n-c_{3}+1:n}\nonumber \\
d_{dep} & = & \sum_{i=c_1+1}^{c_2}\bar{\h}_i^2+\sum_{i=n-c_3+1}^{n}\bar{\h}_i^2 \nonumber \\
r_{dep} & = & r_{socp}- x_{mag}\sum_{i=c_2+1}^{n-c_{3}}\bar{\h}_i\nonumber \\
a_{dep} & = & \frac{\|\g\|_2^2-(\sum_{i=c_1+1}^{c_2}\bar{\h}_i^2+\sum_{i=n-c_3+1}^{n}\bar{\h}_i^2)}
{\sqrt{\sigma^2+(n-c_{3}-c_{2})x_{mag}^2}^{-1}r_{dep}}=
\frac{\sqrt{\sigma^2+(n-c_{3}-c_{2})x_{mag}^2}(\|\g\|_2^2-d_{dep})}{r_{dep}}\nonumber \\
b_{dep} & = & \frac{\bar{\h}_{c_{1}+1:c_{2}}^T\z_{c_{1}+1:c_{2}}
-\bar{\h}_{n-c_{3}+1:n}^T\z_{n-c_{3}+1:n}}{\sqrt{\sigma^2+(n-c_{3}-c_{2})x_{mag}^2}^{-1}r_{dep}}=\frac{\sqrt{\sigma^2+(n-c_{3}-c_{2})x_{mag}^2}s_{dep}}{r_{dep}}.
\label{eq:defdep}
\end{eqnarray}
Then combining (\ref{eq:mainlasso11}) and (\ref{eq:defdep}) one obtains
\begin{equation}
(a_{dep}\nu+b_{dep})^2=\|\g\|_2^2\nu^2-\|\nu\bar{\h}_{c_{1}+1:c_{2}}-\z_{c_{1}+1:c_{2}}\|_2^2
-\|\nu\bar{\h}_{n-c_{3}+1:n}+\z_{n-c_{3}+1:n}\|_2^2.\label{eq:mainlasso13}
\end{equation}
After solving (\ref{eq:mainlasso13}) over $\nu$ we have
\begin{equation}
\hspace{-.0in}\nu=\frac{-(a_{dep}b_{dep}-s_{dep})-\sqrt{(a_{dep}b_{dep}-s_{dep})^2-
(b_{dep}^2+\|\z_{c_{1}+1:c_{2}}\|_2^2+\|\z_{n-c_{3}+1:n}\|_2^2)(a_{dep}^2-\|\g\|_2^2+d_{dep})}}
{a_{dep}^2-\|\g\|_2^2+d_{dep}}.\label{eq:mainlasso14}
\end{equation}
Following what was done in \cite{StojnicCSetam09,StojnicGenSocp10}, we have that a combination of (\ref{eq:mainlasso6}) and (\ref{eq:mainlasso14}) gives the following three equations that can be used to determine $c_1$, $c_2$, and $c_3$ (the equations are rather inequalities; since we will assume a large dimensional scenario we will instead of any of the inequalities below write an equality; this will make writing much easier).
\begin{eqnarray}
\nu\bar{\h}_{c_{2}}-\z_{c_{2}}+\frac{x_{mag}}{\sigma}\sqrt{\|\g\|_2^2\nu^2-\|\nu\bar{\h}_{c_{1}+1:n}-\z_{c_{1}+1:n}+\lambda_{c_{1}+1:n}\|_2^2}& =& 0\nonumber \\
\nu\bar{\h}_{n-c_{3}}+\z_{n-c_{3}}+\frac{x_{mag}}{\sigma}\sqrt{\|\g\|_2^2\nu^2-\|\nu\bar{\h}_{c_{1}+1:n}-\z_{c_{1}+1:n}+\lambda_{c_{1}+1:n}\|_2^2} & = & 0\nonumber \\
\hspace{-.5in}\bar{\h}_{c_1}\frac{-(a_{dep}b_{dep}-s_{dep})-\sqrt{(a_{dep}b_{dep}-s_{dep})^2-
(b_{dep}^2+\|\z_{c_{1}+1:c_{2}}\|_2^2+\|\z_{n-c_{3}+1:n}\|_2^2)(a_{dep}^2-\|\g\|_2^2+d_{dep})}}
{a_{dep}^2-\|\g\|_2^2+d_{dep}} & = & 1.\nonumber \\
\label{eq:mainlasso15}
\end{eqnarray}
The last term that appears on the right hand side of the first two of the above equations can be further simplified based on
(\ref{eq:mainlasso8}) in the following way
\begin{multline}
\hspace{-.7in}\sqrt{\|\g\|_2^2\nu^2-\|\nu\bar{\h}_{c_{1}+1:n}-\z_{c_{1}+1:n}+\lambda_{c_{1}+1:n}\|_2^2}=
\frac{\sigma\sqrt{\|\g\|_2^2\nu^2-\|\nu\bar{\h}_{c_{1}+1:c_{2}}-\z_{c_{1}+1:c_{2}}\|_2^2
-\|\nu\bar{\h}_{n-c_{3}+1:n}+\z_{n-c_{3}+1:n}\|_2^2}}{\sqrt{\sigma^2+(n-c_3-c_2)x_{mag}^2}} \\
\hspace{-.5in}=\frac{\sigma\sqrt{\|\g\|_2^2\nu^2-\nu^2d_{dep}+2\nu s_{dep}-(\|\z_{c_{1}+1:c_{2}}\|_2^2+\|\z_{n-c_{3}+1:n}\|_2^2)}}{\sqrt{\sigma^2+(n-c_3-c_2)x_{mag}^2}}=
\frac{\sigma\sqrt{\|\g\|_2^2\nu^2-\nu^2d_{dep}+2\nu s_{dep}-(c_2-c_1+c_3)}}{\sqrt{\sigma^2+(n-c_3-c_2)x_{mag}^2}},\label{eq:mainlasso16}
\end{multline}
where we of course recognized that $\|\z_{c_{1}+1:c_{2}}\|_2^2+\|\z_{n-c_{3}+1:n}\|_2^2=c_2-c_1+c_3$. Combining (\ref{eq:defdep}) and (\ref{eq:mainlasso16}) one can then simplify the equations from (\ref{eq:mainlasso15}) in the following way
\begin{eqnarray}
\hspace{-.5in}\nu\bar{\h}_{c_{2}}-\z_{c_{2}}+\frac{x_{mag}}{\sqrt{\sigma^2+(n-c_3-c_2)x_{mag}^2}}\sqrt{\|\g\|_2^2\nu^2-\nu^2d_{dep}+2\nu s_{dep}-(c_2-c_1+c_3)} & = & 0\nonumber \\
 \nu\bar{\h}_{n-c_{3}}+\z_{n-c_{3}}+\frac{x_{mag}}{\sqrt{\sigma^2+(n-c_3-c_2)x_{mag}^2}}\sqrt{\|\g\|_2^2\nu^2-\nu^2d_{dep}+2\nu s_{dep}-(c_2-c_1+c_3)} & = & 0\nonumber \\
 \bar{\h}_{c_1}\frac{-(a_{dep}b_{dep}-s_{dep})-\sqrt{(a_{dep}b_{dep}-s_{dep})^2-
(b_{dep}^2+(c_2-c_1+c_3))(a_{dep}^2-\|\g\|_2^2+d_{dep})}}
{a_{dep}^2-\|\g\|_2^2+d_{dep}} & = & 1.\nonumber \\
\label{eq:c1c2c3}
\end{eqnarray}
Let $\widehat{c_1}$, $\widehat{c_2}$, and $\widehat{c_3}$ be the solution of (\ref{eq:c1c2c3}). Then
\begin{equation}
\nu_{dep} = \frac{-(\widehat{a_{dep}}\widehat{b_{dep}}-\widehat{s_{dep}})-\sqrt{(\widehat{a_{dep}}\widehat{b_{dep}}-\widehat{s_{dep}})^2-
(\widehat{b_{dep}}^2+(\widehat{c_2}-\widehat{c_1}+\widehat{c_3}))(\widehat{a_{dep}}^2-\|\g\|_2^2+\widehat{d_{dep}})}}
{\widehat{a_{dep}}^2-\|\g\|_2^2+\widehat{d_{dep}}},\label{eq:mainhatnu}
\end{equation}
where $\widehat{s_{dep}}$, $\widehat{d_{dep}}$, $\widehat{a_{dep}}$, and $\widehat{b_{dep}}$ are $s_{dep}$, $d_{dep}$, $a_{dep}$, and $b_{dep}$ from (\ref{eq:defdep}) computed with $\widehat{c_1}$, $\widehat{c_2}$, and $\widehat{c_3}$. From (\ref{eq:mainlasso4}) one then has
\begin{equation}
\|\w_{dep}\|_2=\sigma\frac{\|\nu_{dep}\bar{\h}_{\widehat{c_{1}}+1:n}-\z_{\widehat{c_{1}}+1:n}+\lambda_{\widehat{c_{1}}+1:n}^{(dep)}\|_2}
{\sqrt{\|\g\|_2^2\nu_{dep}^2-\|\nu_{dep}\bar{\h}_{\widehat{c_{1}}+1:n}-\z_{\widehat{c_{1}}+1:n}+\lambda_{\widehat{c_{1}}+1:n}^{(dep)}\|_2^2}}.\label{eq:mainhatw}
\end{equation}
Combining (\ref{eq:mainlasso16}) and (\ref{eq:mainhatw}) one further has
\begin{equation}
\|\w_{dep}\|_2=\sigma\frac{\sqrt{\|\g\|_2^2\nu_{dep}^2(n-\widehat{c_3}-\widehat{c_2})\frac{x_{mag}^2}{\sigma^2}+\nu_{dep}^2\widehat{d_{dep}}-2\nu \widehat{s_{dep}}+(\widehat{c_2}-\widehat{c_1}+\widehat{c_3})}}
{\sqrt{\|\g\|_2^2\nu_{dep}^2-\nu_{dep}^2\widehat{d_{dep}}+2\nu_{dep} \widehat{s_{dep}}-(\widehat{c_2}-\widehat{c_1}+\widehat{c_3})}}.\label{eq:mainhatw1}
\end{equation}
Combination of (\ref{eq:mainhatnu}) and (\ref{eq:mainhatw1}) is conceptually enough to determine $\|\w_{dep}\|_2$ (and then afterwards easily $E\xi_{prim}^{(dep)}(\sigma,\g,\h,x_{mag},r_{socp})$). What is left to be done is a computation of all unknown quantities that appear in (\ref{eq:mainhatnu}) and (\ref{eq:mainhatw1}). We will below show how that can be done. As mentioned earlier what we will present substantially relies on what was shown in \cite{StojnicCSetam09,StojnicGenSocp10} and we assume a familiarity with the procedures presented there.

The first thing to resolve is (\ref{eq:c1c2c3}). Since all random quantities concentrate we will be dealing (as in \cite{StojnicCSetam09,StojnicGenSocp10}) with the expected values. To compute the solution of (\ref{eq:c1c2c3}), $\widehat{c_1}$, $\widehat{c_2}$, and $\widehat{c_3}$,  we will need the following expected values
\begin{eqnarray}
& & E\|\g\|_2^2, E\|\bar{\h}_{c_{1}+1:n-k}\|_2^2, E\|\bar{\h}_{n-k+1:c_2}\|_2^2, E\|\bar{\h}_{n-c_3+1:n}\|_2^2,\nonumber \\
& & E (\bar{\h}_{c_{1}+1:n-k}^T\z_{c_{1}+1:n-k}), E (\bar{\h}_{n-k+1:c_2}^T\z_{n-k+1:c_2}), E (\bar{\h}_{n-c_3+1:n}^T\z_{n-c_{3}+1:n}).\label{eq:mainexp1}
\end{eqnarray}
Clearly, since components of $\g$ are i.i.d. standard normals one easily has
\begin{equation}
E\|\g\|_2^2=m.\label{eq:mainexpg}
\end{equation}
Let $c_{1}=(1-\theta_1)n$, $c_{2}=\theta_2 n$, and $c_{3}=\theta_3 n$ where $\theta_1$, $\theta_2$, and $\theta_3$ are constants independent of $n$. Then as shown in \cite{StojnicCSetam09,StojnicGenSocp10}
\begin{equation}
\lim_{n\rightarrow\infty}\frac{E\|\bar{\h}_{c_{1}+1:n-k}\|_2^2}{n}  =  \frac{1-\beta_w}{\sqrt{2\pi}}\left (2\frac{\sqrt{2(\erfinv(\frac{1-\theta_1}{1-\beta_w}))^2}}{e^{(\erfinv(\frac{1-\theta_1}{1-\beta_w}))^2}}\right )+\theta_1-\beta_w.\label{eq:mainexpnormh1}
\end{equation}
where we of course recall that $\beta_w=\frac{k}{n}$. Also, as shown in \cite{StojnicCSetam09,StojnicGenSocp10}
\begin{equation}
\lim_{n\rightarrow\infty}\frac{E\|\bar{\h}_{n-k+1:c_{2}}\|_2^2}{n}  =  \frac{\beta_w}{\sqrt{2\pi}}\left (\frac{\sqrt{2}\erfinv(2\frac{1-\theta_2}{\beta_w}-1)}{e^{(\erfinv(2\frac{1-\theta_2}{\beta_w}-1))^2}}\right )+\theta_2-1+\beta_w,\label{eq:mainexpnormh2}
\end{equation}
and
\begin{equation}
\lim_{n\rightarrow\infty}\frac{E\|\bar{\h}_{n-c_{3}+1:n}\|_2^2}{n}  =  \frac{\beta_w}{\sqrt{2\pi}}\left (\frac{\sqrt{2}\erfinv(2\frac{\beta_w-\theta_3}{\beta_w}-1)}{e^{(\erfinv(2\frac{\beta_w-\theta_3}{\beta_w}-1))^2}}\right )+\theta_3,\label{eq:mainexpnormh3}
\end{equation}
Following further what was established in \cite{StojnicCSetam09,StojnicGenSocp10} we have
\begin{eqnarray}
\lim_{n\rightarrow\infty}\frac{E(\bar{\h}_{c_{1}+1:n-k}^T\z_{c_{1}+1:n-k})}{n} & = &
\left ((1-\beta_w)\sqrt{\frac{2}{\pi}}e^{-(\erfinv(\frac{1-\theta_1}{1-\beta_w}))^2}\right )\nonumber \\
\lim_{n\rightarrow\infty}\frac{E(\bar{\h}_{n-k+1:c_{2}}^T\z_{n-k+1:c_{2}})}{n} & = &
\left (\beta_w\sqrt{\frac{1}{2\pi}}e^{-(\erfinv(2\frac{1-\theta_2}{\beta_w}-1))^2}\right )\nonumber \\
\lim_{n\rightarrow\infty}\frac{E(\bar{\h}_{n-c_{3}+1:n}^T\z_{n-c_{3}+1:n})}{n} & = &
-\left (\beta_w\sqrt{\frac{1}{2\pi}}e^{-(\erfinv(2\frac{\beta_w-\theta_3}{\beta_w}-1))^2}\right ).
\label{eq:mainexpprod}
\end{eqnarray}
From (\ref{eq:mainexpprod}) we also have
\begin{equation}
\lim_{n\rightarrow\infty}\frac{E(\sum_{i=c_2+1}^{n-c_3}\bar{\h}_{i})}{n} =
\left (\beta_w\sqrt{\frac{1}{2\pi}}e^{-(\erfinv(2\frac{\beta_w-\theta_3}{\beta_w}-1))^2}\right )
-\left (\beta_w\sqrt{\frac{1}{2\pi}}e^{-(\erfinv(2\frac{1-\theta_2}{\beta_w}-1))^2}\right ).\label{eq:mainexpprod1}
\end{equation}
The only other thing that we will need in order to be able to compute $\widehat{c_{1}}$, $\widehat{c_{2}}$, and $\widehat{c_{3}}$ (besides the expectations from (\ref{eq:mainexp1})) are the following inequalities related to the behavior of $\bar{\h}_{c_{1}}$, $\bar{\h}_{c_{2}}$, and $\bar{\h}_{c_{3}}$. Again, as shown in \cite{StojnicCSetam09,StojnicGenSocp10}
\begin{eqnarray}
P(\sqrt{2}\erfinv ((1+\epsilon_1^{\bar{\h}_{c_{1}}})(\frac{1-\theta_1}{1-\beta_w}))\leq \bar{\h}_{c_{1}}) & \leq & e^{-\epsilon_2^{\bar{\h}_{c_{1}}} n}\nonumber \\
P(\sqrt{2}\erfinv ((1+\epsilon_1^{\bar{\h}_{c_{2}}})(2\frac{1-\theta_2}{\beta_w}-1))\leq \bar{\h}_{c_{2}}) & \leq & e^{-\epsilon_2^{\bar{\h}_{c_{2}}} n}\nonumber \\
P(-\sqrt{2}\erfinv ((1+\epsilon_1^{\bar{\h}_{c_{3}}})(2\frac{\beta_w-\theta_3}{\beta_w}-1))\leq \bar{\h}_{n-c_{3}}) & \leq & e^{-\epsilon_2^{\bar{\h}_{n-c_{3}}} n}.\label{eq:mainbrpoint}
\end{eqnarray}
where $\epsilon_1^{\bar{\h}_{c_{1}}}>0$, $\epsilon_1^{\bar{\h}_{c_{2}}}>0$, and $\epsilon_1^{\bar{\h}_{n-c_{3}}}>0$ are arbitrarily small constants and $\epsilon_2^{\bar{\h}_{c_{1}}}$, $\epsilon_2^{\bar{\h}_{c_{2}}}$, and $\epsilon_2^{\bar{\h}_{n-c_{3}}}$ are constants dependent on $\epsilon_1^{\bar{\h}_{c_{1}}}$, $\epsilon_1^{\bar{\h}_{c_{2}}}$, and $\epsilon_1^{\bar{\h}_n-{c_{3}}}$, respectively, but independent of $n$ (essentially one only needs the direction of inequalities as in (\ref{eq:mainbrpoint}); however, a similar reverse inequalities hold as well).

At this point we have all the necessary ingredients to determine $\widehat{c_{1}}$, $\widehat{c_{2}}$, and $\widehat{c_{3}}$ and consequently $\nu_{dep}$, $\|\w_{dep}\|_2$, and $\xi_{prim}^{(dep)}(\sigma,\g,\h,x_{mag},r_{socp})$. We of course recall that in a random setup that we consider quantities $\widehat{c_{1}}$, $\widehat{c_{2}}$, $\widehat{c_{3}}$, $\nu_{dep}$, $\|\w_{dep}\|_2$, and $\xi_{prim}^{(dep)}(\sigma,\g,\h,x_{mag},r_{socp})$ can not really be determined. Instead what we will be determining are their concentrating points. The following theorem then provides a systematic way of doing so.
\begin{theorem}
Assume the setup of Theorem \ref{thm:mainlasso}. Let the nonzero components of $\xtilde$ have magnitude $x_{mag}$ and let $\bar{\h}$ be as defined in (\ref{eq:defhbar}). Further, let $r_{socp}^{(sc)}=\lim_{n\rightarrow\infty}\frac{r_{socp}}{\sqrt{n}}$ and $x_{mag}^{(sc)}=\lim_{n\rightarrow\infty}\frac{x_{mag}}{\sqrt{n}}$. Also, let $\nu_{dep}$, $\|\w_{dep}\|_2$, and $\xi_{prim}^{(dep)}(\sigma,\g,\h,x_{mag},r_{socp})$ be as defined in and right after (\ref{eq:mainlasso3ver}). Let $\alpha=\frac{m}{n}$ and $\beta_w=\frac{k}{n}$ be fixed. Consider the following
\begin{multline*}
S(\theta_1,\theta_2,\theta_3)  =  \lim_{n\rightarrow\infty}\frac{E s_{dep}}{n}  =  \left ((1-\beta_w)\sqrt{\frac{2}{\pi}}e^{-(\erfinv(\frac{1-\theta_1}{1-\beta_w}))^2}\right )+\left (\beta_w\sqrt{\frac{1}{2\pi}}e^{-(\erfinv(2\frac{1-\theta_2}{\beta_w}-1))^2}\right ) \\+\left (\beta_w\sqrt{\frac{1}{2\pi}}e^{-(\erfinv(2\frac{\beta_w-\theta_3}{\beta_w}-1))^2}\right )
\end{multline*}
\begin{multline*}
D(\theta_1,\theta_2,\theta_3)  =  \lim_{n\rightarrow\infty}\frac{E d_{dep}}{n}  =  \frac{1-\beta_w}{\sqrt{2\pi}}\left (2\frac{\sqrt{2(\erfinv(\frac{1-\theta_1}{1-\beta_w}))^2}}{e^{(\erfinv(\frac{1-\theta_1}{1-\beta_w}))^2}}\right )+\theta_1-\beta_w\\+\frac{\beta_w}{\sqrt{2\pi}}\left (\frac{\sqrt{2}\erfinv(2\frac{1-\theta_2}{\beta_w}-1)}{e^{(\erfinv(2\frac{1-\theta_2}{\beta_w}-1))^2}}\right )+\theta_2-1+\beta_w
+\frac{\beta_w}{\sqrt{2\pi}}\left (\frac{\sqrt{2}\erfinv(2\frac{1-\theta_3}{\beta_w}-1)}{e^{(\erfinv(2\frac{\beta_w-\theta_3}{\beta_w}-1))^2}}\right )+\theta_3
\end{multline*}
\begin{multline*}
R(\theta_2,\theta_3)  =  \lim_{n\rightarrow\infty}\frac{E r_{dep}}{\sqrt{n}}  =  r_{socp}^{(sc)}-x_{mag}^{(sc)}\left ( \left (\beta_w\sqrt{\frac{1}{2\pi}}e^{-(\erfinv(2\frac{\beta_w-\theta_3}{\beta_w}-1))^2}\right )
-\left (\beta_w\sqrt{\frac{1}{2\pi}}e^{-(\erfinv(2\frac{1-\theta_2}{\beta_w}-1))^2}\right ) \right)
\end{multline*}
\begin{eqnarray}
A(\theta_1,\theta_2,\theta_3) & = & \lim_{n\rightarrow\infty}\frac{E a_{dep}}{\sqrt{n}}=\frac{\sqrt{\sigma^2+(1-\theta_3-\theta_2)(x_{mag}^{(sc)})^2}(\alpha-D(\theta_1,\theta_2,\theta_3))}{R(\theta_2,\theta_3)}\nonumber \\
B(\theta_1,\theta_2,\theta_3) & = & \lim_{n\rightarrow\infty}\frac{E b_{dep}}{\sqrt{n}}=\frac{\sqrt{\sigma^2+(1-\theta_3-\theta_2)(x_{mag}^{(sc)})^2}S(\theta_1,\theta_2,\theta_3)}{R(\theta_2,\theta_3)}\nonumber \\
F(\theta_1) & = & \sqrt{2}\erfinv (\frac{1-\theta_1}{1-\beta_w})\nonumber \\
G(\theta_2) & = & \sqrt{2}\erfinv (2\frac{1-\theta_2}{\beta_w}-1)\nonumber \\
H(\theta_3) & = & \sqrt{2}\erfinv (2\frac{\beta-\theta_3}{\beta_w}-1).\label{eq:maincompthmcond1}
\end{eqnarray}
Set
\begin{multline}
N(\theta_1,\theta_2,\theta_3)=\frac{-(A(\theta_1,\theta_2,\theta_3)B(\theta_1,\theta_2,\theta_3)-S(\theta_1,\theta_2,\theta_3))}
{A(\theta_1,\theta_2,\theta_3)^2-\alpha+D(\theta_1,\theta_2,\theta_3)}\nonumber \\
\hspace{-.6in}-\frac{\sqrt{(A(\theta_1,\theta_2,\theta_3)B(\theta_1,\theta_2,\theta_3)-S(\theta_1,\theta_2,\theta_3))^2-
(B(\theta_1,\theta_2,\theta_3)^2+\theta_1+\theta_2+\theta_3-1)(A(\theta_1,\theta_2,\theta_3)^2-\alpha+D(\theta_1,\theta_2,\theta_3))}}
{A(\theta_1,\theta_2,\theta_3)^2-\alpha+D(\theta_1,\theta_2,\theta_3)}.
\end{multline}
Let the triplet ($\widehat{\theta_1}$, $\widehat{\theta_2}$, $\widehat{\theta_3}$) be the solution of the following three equations
\begin{eqnarray}
\hspace{-.6in}N(\theta_1,\theta_2,\theta_3)G(\theta_2)+\frac{x_{mag}^{(sc)}\sqrt{N(\theta_1,\theta_2,\theta_3)^2(\alpha^2
-D(\theta_1,\theta_2,\theta_3))+2N(\theta_1,\theta_2,\theta_3) S(\theta_1,\theta_2,\theta_3)-(\theta_1+\theta_2+\theta_3-1)}}
{\sqrt{\sigma^2+(1-\theta_3-\theta_2)(x_{mag}^{(sc)})^2}} & = & 1\nonumber \\
\hspace{-.6in}N(\theta_1,\theta_2,\theta_3)H(\theta_3)-\frac{x_{mag}^{(sc)}\sqrt{N(\theta_1,\theta_2,\theta_3)^2(\alpha^2
-D(\theta_1,\theta_2,\theta_3))+2N(\theta_1,\theta_2,\theta_3) S(\theta_1,\theta_2,\theta_3)-(\theta_1+\theta_2+\theta_3-1)}}
{\sqrt{\sigma^2+(1-\theta_3-\theta_2)(x_{mag}^{(sc)})^2}} & = & 1\nonumber \\
F(\theta_1)N(\theta_1,\theta_2,\theta_3) & = & 1.\nonumber \\\label{eq:mainthmc1c2c3}
\end{eqnarray}
Then the concentrating points of $\nu_{dep}$, $\|\w_{dep}\|_2$, and $\xi_{prim}^{(dep)}(\sigma,\g,\h,x_{mag},r_{socp})$ can be determined as
\begin{equation*}
E\nu_{dep}  =  N(\widehat{\theta_1},\widehat{\theta_2},\widehat{\theta_3})
\end{equation*}
\begin{equation*}
\hspace{-.8in}E\|\w_{dep}\|_2  =  \sigma\frac{\sqrt{N(\widehat{\theta_1},\widehat{\theta_2},\widehat{\theta_3})^2(\alpha(1-\widehat{\theta_3}-\widehat{\theta_2})\frac{(x_{mag}^{(sc)})^2}{\sigma^2}
+
D(\widehat{\theta_1},\widehat{\theta_2},\widehat{\theta_3}))-2N(\widehat{\theta_1},\widehat{\theta_2},\widehat{\theta_3}) S(\widehat{\theta_1},\widehat{\theta_2},\widehat{\theta_3})+(\widehat{\theta_1}+\widehat{\theta_2}+\widehat{\theta_3}-1)}}
{\sqrt{N(\widehat{\theta_1},\widehat{\theta_2},\widehat{\theta_3})^2(\alpha
-
D(\widehat{\theta_1},\widehat{\theta_2},\widehat{\theta_3}))+2N(\widehat{\theta_1},\widehat{\theta_2},\widehat{\theta_3}) S(\widehat{\theta_1},\widehat{\theta_2},\widehat{\theta_3})-(\widehat{\theta_1}+\widehat{\theta_2}+\widehat{\theta_3}-1)}}
\end{equation*}
\begin{multline}
\hspace{-.9in}\lim_{n\rightarrow\infty}\frac{E\xi_{prim}^{(dep)}(\sigma,\g,\h,x_{mag},r_{socp})}{\sqrt{n}}  =  \sigma
\frac{\sqrt{N(\widehat{\theta_1},\widehat{\theta_2},\widehat{\theta_3})^2(\alpha
-
D(\widehat{\theta_1},\widehat{\theta_2},\widehat{\theta_3}))+2N(\widehat{\theta_1},\widehat{\theta_2},\widehat{\theta_3}) S(\widehat{\theta_1},\widehat{\theta_2},\widehat{\theta_3})-(\widehat{\theta_1}+\widehat{\theta_2}+\widehat{\theta_3}-1)}}
{\sqrt{1+(1-\widehat{\theta_2}-\widehat{\theta_3})\frac{(x_{mag}^{(sc)})^2}{\sigma^2}}}\\
-N(\widehat{\theta_1},\widehat{\theta_2},\widehat{\theta_3})r_{socp}^{(sc)}.\label{eq:mainthmnuwgenxiprim}
\end{multline}
\label{thm:maincomperror}
\end{theorem}
\begin{proof}
Follows from Theorem \ref{thm:mainlasso} based on the discussion presented above and a combination of (\ref{eq:defdep}), (\ref{eq:c1c2c3}), (\ref{eq:mainhatnu}), and (\ref{eq:mainhatw1}).
\end{proof}

The results from the above theorem can be used to compute parameters of interest in our derivation for particular values of $\beta_w$, $\alpha$, $\sigma$, $x_{mag}$, and $r_{socp}$. In the following subsection we will present a collection of such results.

\subsubsection{Theoretical predictions} \label{sec:unsignedtheorypred}

In this subsection we present the theoretical predictions one can get based on the result of the previous sections. We will split the presentation of the results into several parts.

\textbf{\underline{\emph{1) $\frac{E\|\w_{dep}\|_2}{\sigma}=\frac{E\|\w_{socp}\|_2}{\sigma}$ as a function of $x_{mag}^{(sc)}$}}}

To present this portion (as well as several others that will follow) of theoretical results we will look at three regimes: 1) low $\alpha$-, medium $\alpha$-, and high $\alpha$-regime. For each of the regimes we will show the theoretical results for $\frac{E\|\w_{dep}\|_2}{\sigma}=\frac{E\|\w_{socp}\|_2}{\sigma}$ as a function of $x_{mag}^{(sc)}$. We will take $\alpha=0.3$ as a representative of the low $\alpha$-regime, $\alpha=0.5$ as a representative of the medium $\alpha$-regime, and $\alpha=0.7$ as a representative of the high $\alpha$-regime. We will consider $r_{socp}=r_{socp}^{(opt)}=\sigma\sqrt{\frac{\alpha n}{1+\rho^2}}$. For each of the $\alpha$-regimes we will look at two different sub-regimes: low $\beta_w$- and high $\beta_w$-regime (which based on results from \cite{StojnicGenSocp10} is equivalent to low $\rho$- and high $\rho$-regimes). For each of these two sub-regimes $\beta_w$ will be selected based on the curves obtained in \cite{StojnicGenSocp10} (or those obtained in \cite{StojnicGenLasso10}) in the following way. In the low $\beta_w$ sub-regime we will set $\rho=2$ and $r_{socp}=r_{socp}^{(opt)}=\sigma\sqrt{\frac{\alpha n}{5}}$ whereas in the high $\beta_w$ sub-regime we will set $\rho=3$ and $r_{socp}=r_{socp}^{(opt)}=\sigma\sqrt{\frac{\alpha n}{10}}$. At the same time from \cite{StojnicGenSocp10} we will have $r_{socp}=r_{socp}^{(opt)}=\sigma\sqrt{(\alpha-\alpha_w)n}$ where $\alpha_w$ and $\beta_w$ are such that (\ref{eq:fundl1}) holds (we also recall on \cite{StojnicGenSocp10} where it was reasoned that the low $\beta$ regime is selected so that the pair $(\alpha,\beta_w)$ is well below the fundamental characterization (\ref{eq:fundl1}) whereas the high $\beta$ regime is selected so that the pair $(\alpha,\beta_w)$ is closer to the fundamental characterization (\ref{eq:fundl1})). The values for $\frac{E\|\w_{dep}\|_2}{\sigma}=\frac{E\|\w_{socp}\|_2}{\sigma}$ one can then get through the results of Theorem \ref{thm:maincomperror} for such $(\alpha,\beta_w)$ pairs, $r_{socp}^{(opt)}$ are shown in Figure \ref{fig:errorvarx} as functions of $x_{mag}^{(sc)}$.
\begin{figure}[htb]
\begin{minipage}[b]{.33\linewidth}
\centering
\centerline{\epsfig{figure=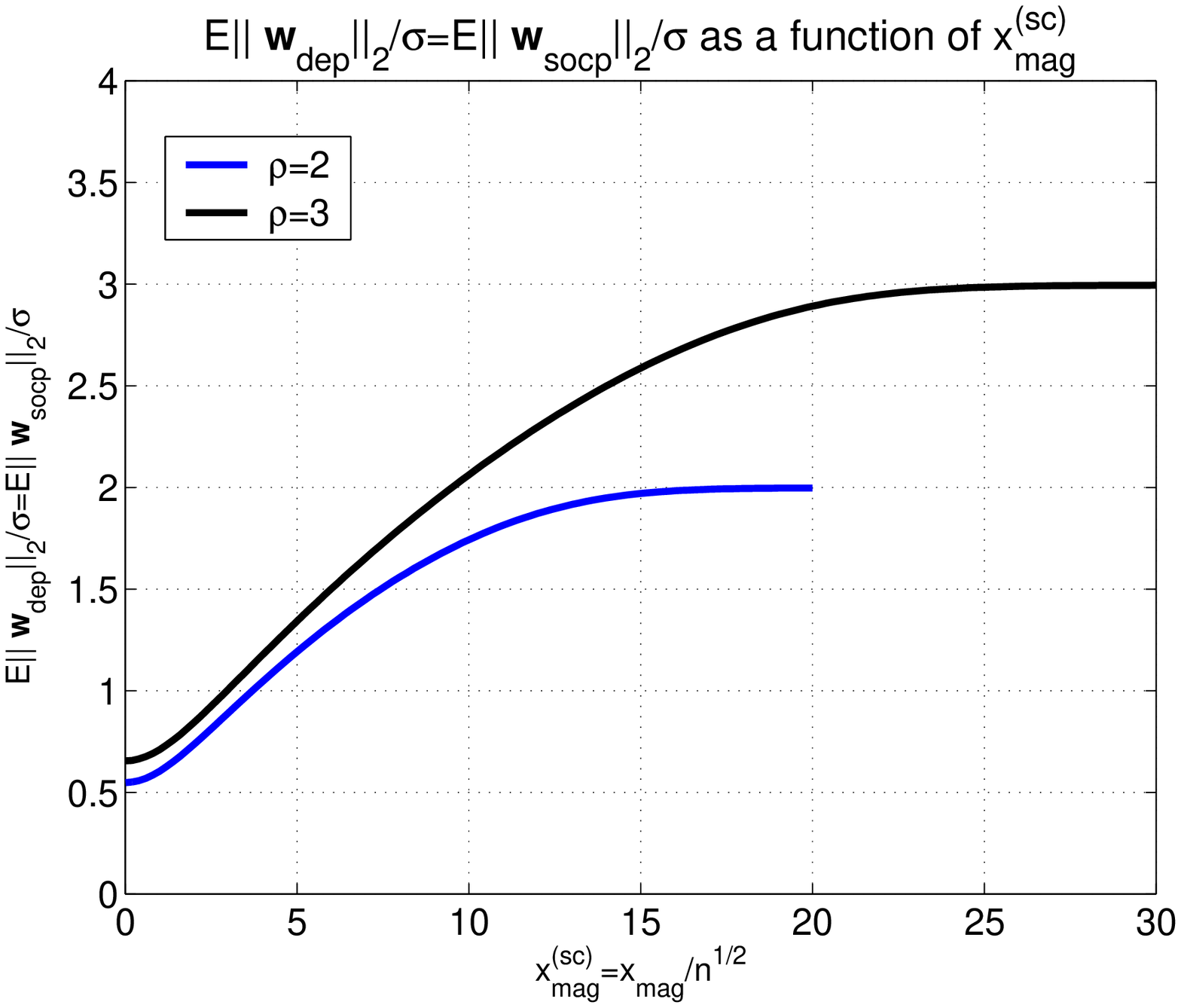,width=5.3cm,height=5.3cm}}
\end{minipage}
\begin{minipage}[b]{.33\linewidth}
\centering
\centerline{\epsfig{figure=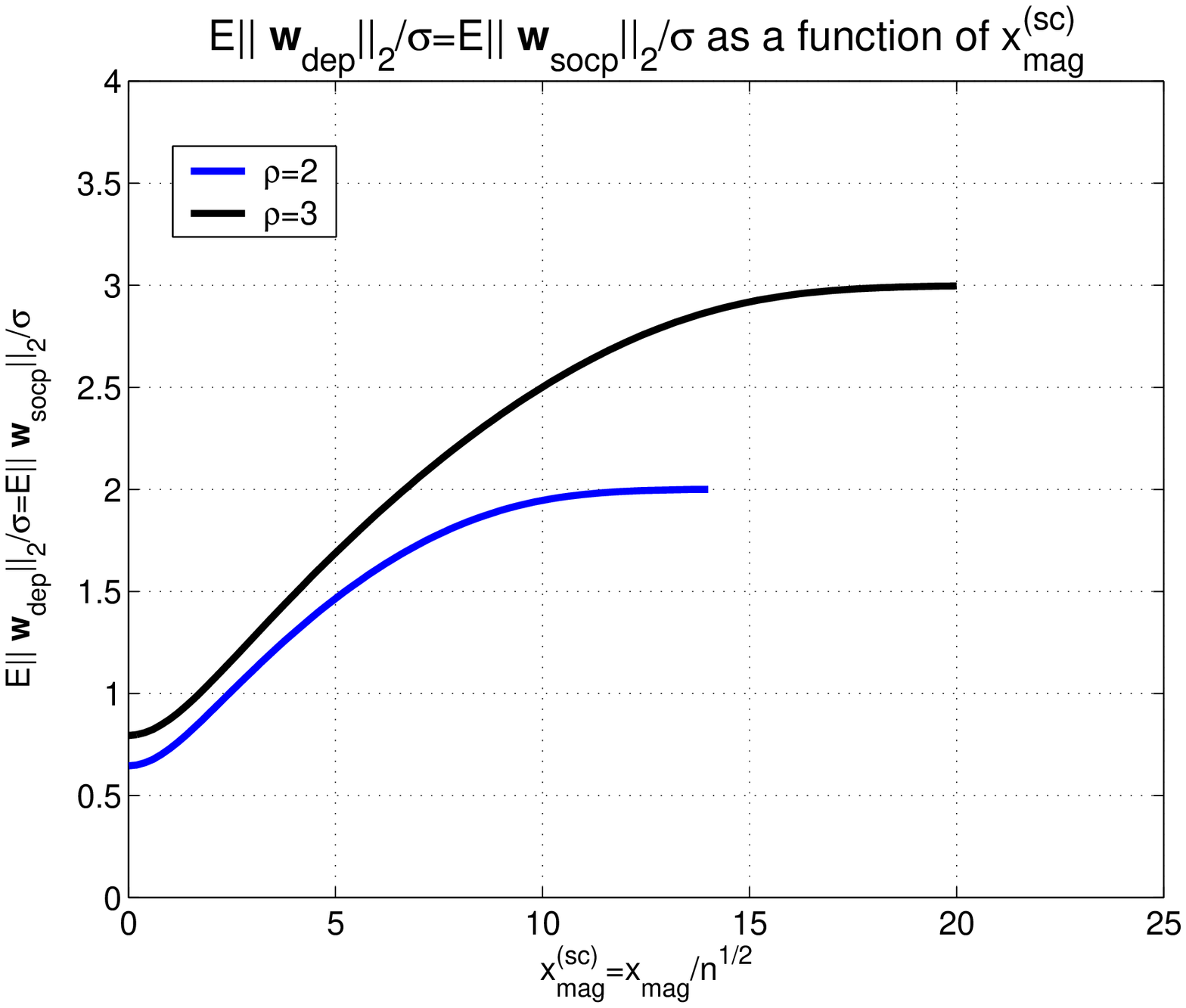,width=5.3cm,height=5.3cm}}
\end{minipage}
\begin{minipage}[b]{.33\linewidth}
\centering
\centerline{\epsfig{figure=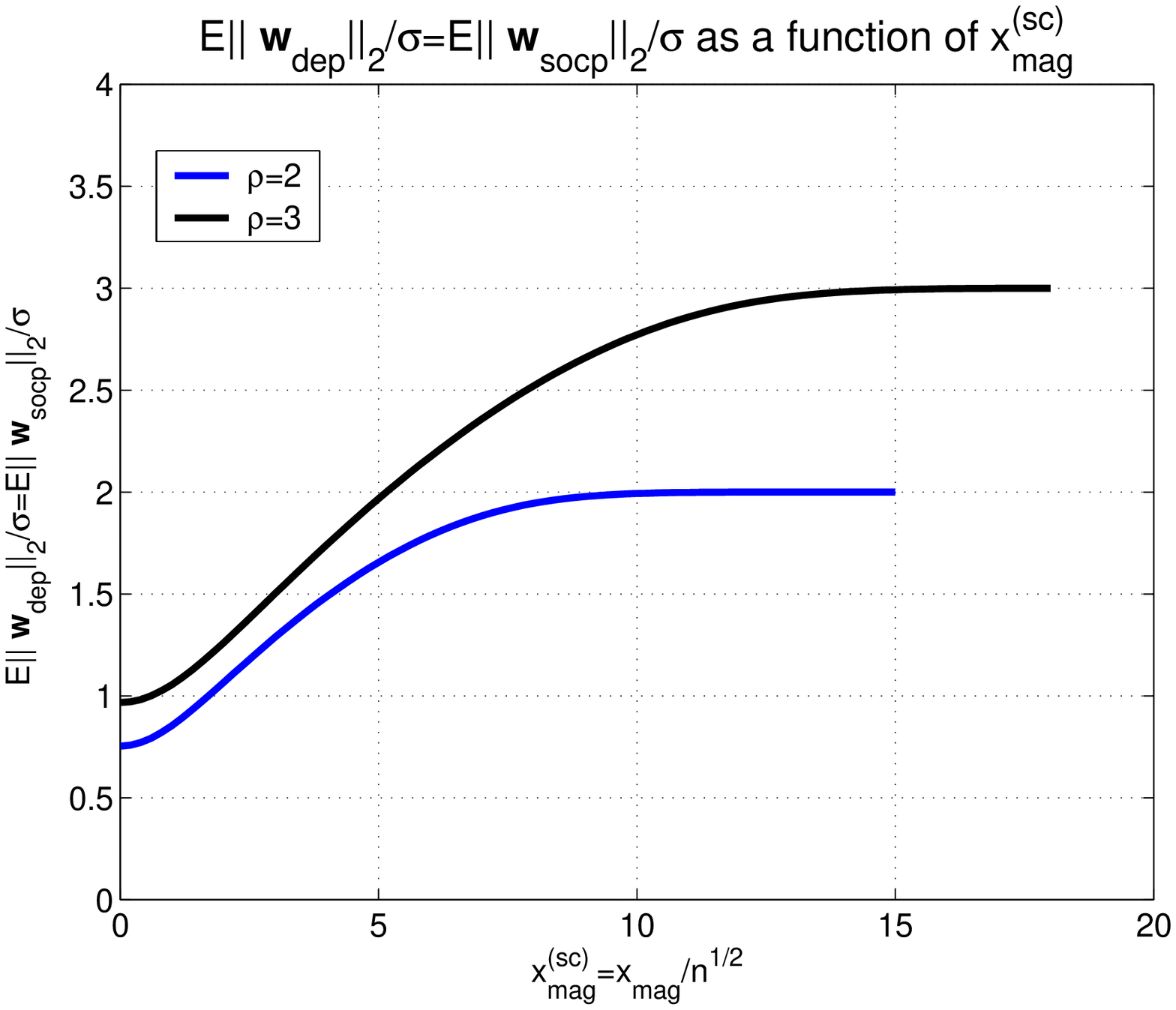,width=5.3cm,height=5.3cm}}
\end{minipage}
\caption{$\frac{E\|\w_{dep}\|_2}{\sigma}=\frac{E\|\w_{socp}\|_2}{\sigma}$ as a function of $x_{mag}^{(sc)}$; $r_{socp}=\sqrt{\frac{\alpha n}{1+\rho^2}}$; left --- $\alpha=0.3$, center --- $\alpha=0.5$, right --- $\alpha=0.7$}
\label{fig:errorvarx}
\end{figure}
As can be seen from Figure \ref{fig:errorvarx}, the values of $\frac{E\|\w_{dep}\|_2}{\sigma}=\frac{E\|\w_{socp}\|_2}{\sigma}$ converge to $\rho$ as $x_{mag}^{(sc)}$ increases. This is of course in agreement with \cite{StojnicGenSocp10} where it was demonstrated that for $r_{socp}^{(opt)}$ one has $\rho=\frac{\|\w_{socp}\|_2}{\sigma}$ with overwhelming probability. Another interesting observation one can make is that the convergence is ``faster" (or happens for smaller $x_{mag}^{(sc)}$) for larger $\alpha$.

\textbf{\underline{\emph{2) $\frac{Ef_{obj}}{\sqrt{n}}=\frac{E\xi_{prim}^{(dep)}(\sigma,\g,\h,x_{mag},r_{socp})}{\sqrt{n}}$ as a function of $x_{mag}^{(sc)}$}}}

Similarly to what was discussed above (and is related to $\|\w_{socp}\|_2$ and $\|\w_{dep}\|_2$) one can determine the concentrating points of $f_{obj}$ and $\xi_{prim}^{(dep)}$ also as functions of $x_{mag}^{(sc)}$. To present these results we restrict ourselves to the medium $\alpha$-regime, or in other words to $\alpha=0.5$. As above we again choose $r_{socp}=r_{socp}^{(opt)}=\sigma\sqrt{\frac{\alpha n}{1+\rho^2}}$ and consider low $\rho=2$- and high $\rho=3$- regime.
The obtained results are shown in Figure \ref{fig:objvarx}.
\begin{figure}[htb]
\begin{minipage}[b]{1\linewidth}
\centering
\centerline{\epsfig{figure=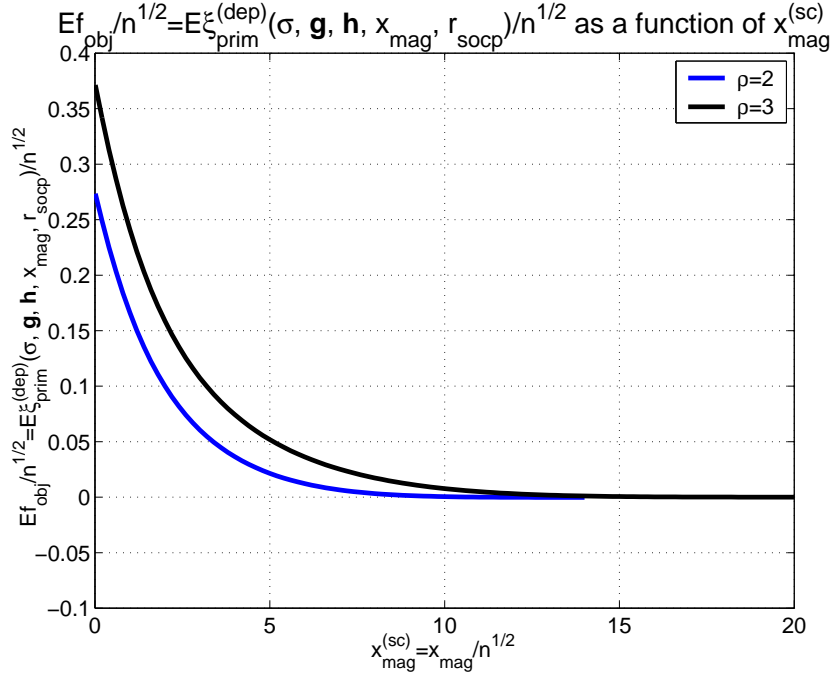,width=11cm,height=9cm}}
\end{minipage}
\caption{$\frac{Ef_{obj}}{\sqrt{n}}=\frac{E\xi_{prim}^{(dep)}(\sigma,\g,\h,x_{mag},r_{socp})}{\sqrt{n}}$ as a function of $x_{mag}^{(sc)}$; $r_{socp}=\sqrt{\frac{\alpha n}{1+\rho^2}}$; $\alpha=0.5$}
\label{fig:objvarx}
\end{figure}
As can be seen from Figure \ref{fig:objvarx} $\frac{Ef_{obj}}{\sqrt{n}}$ is larger for larger $\rho$.

\textbf{\underline{\emph{3) $\frac{E\|\w_{dep}\|_2}{\sigma}=\frac{E\|\w_{socp}\|_2}{\sigma}$ as a function of $x_{mag}^{(sc)}$; varying $r_{socp}$}}}

Another interesting set of results relates to possible variations in the $r_{socp}$ that can be used in (\ref{eq:socp}). The results that we presented above assume an optimal choice for $r_{socp}$ (in a sense defined in \cite{StojnicGenSocp10}). Namely, they assume that for a fixed pair $(\alpha,\beta_w)$ one chooses $r_{socp}=r_{socp}^{(opt)}=\sigma\sqrt{(\alpha-\alpha_w)n}$ where $\alpha_w$ and $\beta_w$ are such that (\ref{eq:fundl1}) holds. In the worst-case scenario (or in the generic scenario as we referred to it in \cite{StojnicGenSocp10}) one has that choice $r_{socp}^{(opt)}$ offers the minimal norm-2 of the error vector. However, such a scenario assumes particular $\xtilde$'s which leaves a possibility that for a wide range of other $\xtilde$'s the performance of the SOCP from (\ref{eq:socp}) in the $\ell_2$ norm of the error vector sense can be more favorable. Of course as shown in Figure (\ref{fig:errorvarx}) this indeed happens to be the case. On the other hand that also leaves an option that one can possibly choose a different $r_{socp}$ and get say a smaller norm-2 of the error vector for various different $\xtilde$. Below we present a few results in this direction.

We will consider again only the medium or $\alpha=0.5$ regime. For two different values of $\rho$, $r_{socp}=r_{socp}^{(opt)}$, and $\beta_w$ we presented the results for $\frac{E\|\w_{dep}\|_2}{\sigma}=\frac{E\|\w_{socp}\|_2}{\sigma}$ in Figure \ref{fig:errorvarx}. In addition to that we now in Figure \ref{fig:errorvarxvarrho} show similar results one can get through Theorem \ref{thm:maincomperror} for two different choices of $r_{socp}$. To be more precise, for $\rho=2$ we choose the same $\alpha$ and $\beta_w$ as in Figure and only vary $r_{socp}$ over $\{\sigma\sqrt{0.05\alpha n},\sigma\sqrt{0.2\alpha n}\sigma\sqrt{0.6\alpha n}\}$. Clearly, choice $\sigma\sqrt{0.05\alpha n}$ is smaller than $r_{socp}^{(opt)}=\sigma\sqrt{0.2\alpha n}$ whereas choice
$\sigma\sqrt{0.6\alpha n}$ is larger than $r_{socp}^{(opt)}=\sigma\sqrt{0.2\alpha n}$. On the other hand for $\rho=3$ we choose the same $\alpha$ and $\beta_w$ as we have chose for $\rho=3$ in Figure \ref{fig:errorvarx} and vary $r_{socp}$ but this time over $\{\sigma\sqrt{0.05\alpha n},\sigma\sqrt{0.1\alpha n}\sigma\sqrt{0.5\alpha n}\}$. Again, clearly, choice $\sigma\sqrt{0.05\alpha n}$ is smaller than $r_{socp}^{(opt)}=\sigma\sqrt{0.1\alpha n}$ whereas choice
$\sigma\sqrt{0.6\alpha n}$ is larger than $r_{socp}^{(opt)}=\sigma\sqrt{0.2\alpha n}$. It is rather obvious but we mention for the completeness that the middle $r_{socp}$ choices for both, $\rho=2$ and $\rho=3$, cases correspond to the center plot in Figure \ref{fig:errorvarx}.
\begin{figure}[htb]
\begin{minipage}[b]{.5\linewidth}
\centering
\centerline{\epsfig{figure=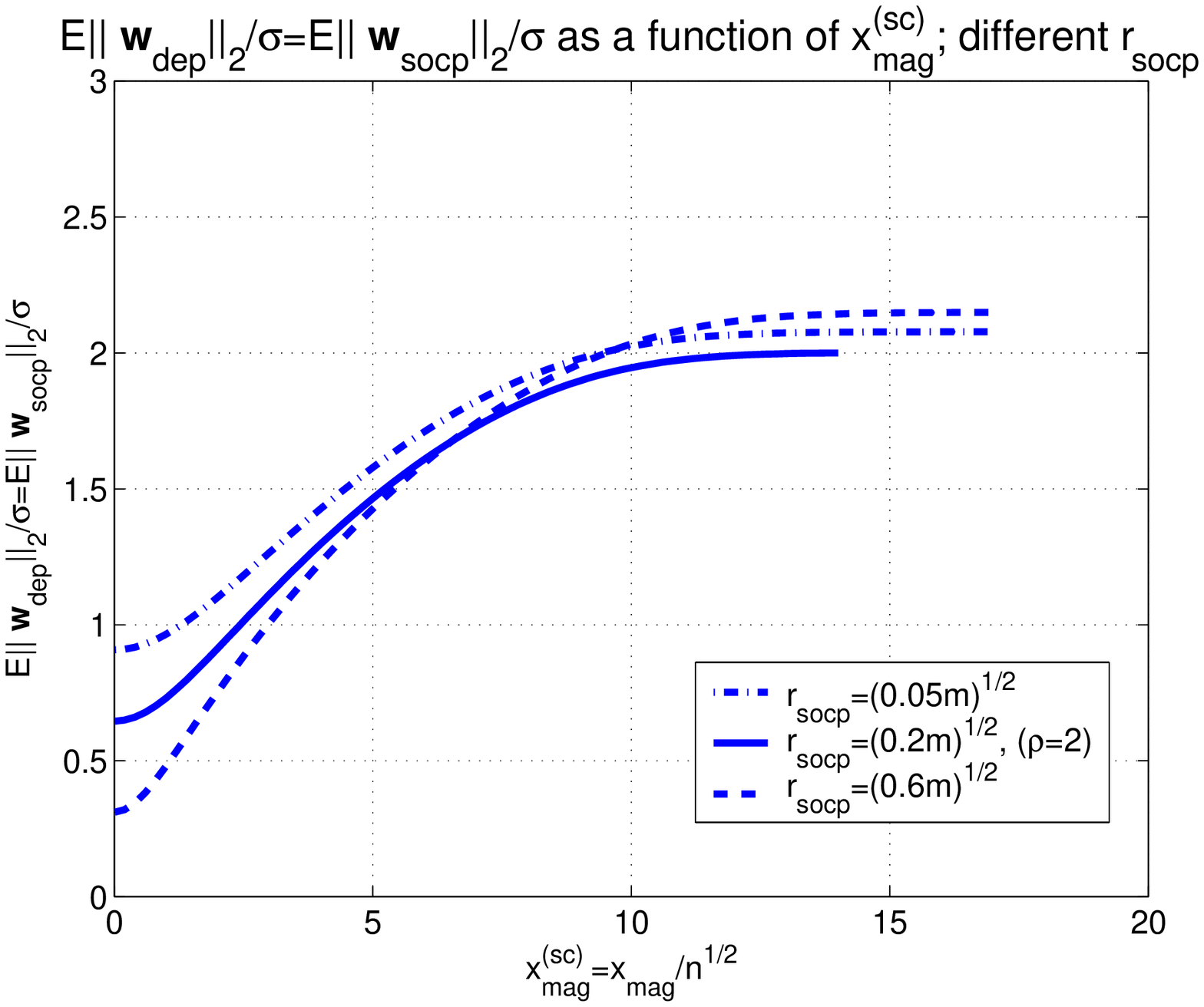,width=8cm,height=7cm}}
\end{minipage}
\begin{minipage}[b]{.5\linewidth}
\centering
\centerline{\epsfig{figure=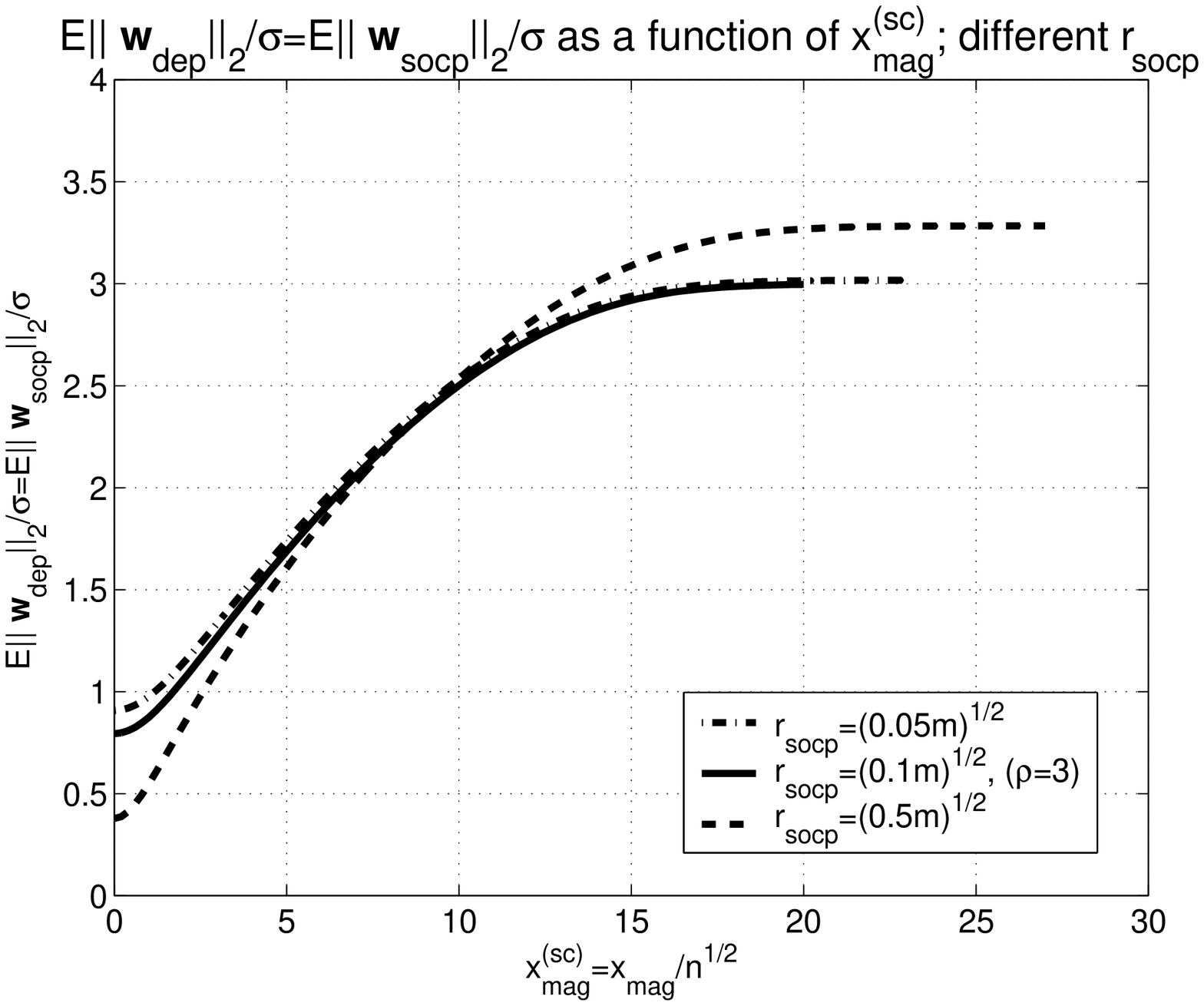,width=8cm,height=7cm}}
\end{minipage}
\caption{$\frac{E\|\w_{dep}\|_2}{\sigma}=\frac{E\|\w_{socp}\|_2}{\sigma}$ as a function of $x_{mag}^{(sc)}$ for different $r_{socp}$; left --- $\rho=2$, $r_{socp}\in\{\sigma\sqrt{0.05\alpha n},\sigma\sqrt{0.2\alpha n}\sigma\sqrt{0.6\alpha n}\}$; right --- $\rho=3$, $r_{socp}\in\{\sigma\sqrt{0.05\alpha n},\sigma\sqrt{0.1\alpha n}\sigma\sqrt{0.5\alpha n}\}$}
\label{fig:errorvarxvarrho}
\end{figure}
We make two interesting observations related to the results presented in Figure \ref{fig:errorvarxvarrho}. The first one is that Figure \ref{fig:errorvarxvarrho} suggests that if $r_{socp}$ is smaller than $r_{socp}^{(opt)}$ then $\frac{E\|\w_{socp}\|_2}{\sigma}$ could be larger than the one that can be obtained for $r_{socp}^{(opt)}$. This actually happens to be the case. A reasoning similar to the one presented in Section 2.4.2 in \cite{StojnicGenSocp10} can show that this is indeed true. Moreover, not only is it true for the $\xtilde$ considered in Theorem \ref{thm:maincomperror} but it is actually true for any $\xtilde$. We skip the details of this simple exercise, though. The second observation is that if $r_{socp}$ is larger than $r_{socp}^{(opt)}$ then for certain $\xtilde$ (but of course not for all of them and certainly not for the worst-case or the generic one) $\frac{E\|\w_{socp}\|_2}{\sigma}$ could be smaller than the one that can be obtained for $r_{socp}^{(opt)}$. This of course suggests that a choice of $r_{socp}$ larger than $r_{socp}^{(opt)}$ could be more favorable in certain applications and for a particular measure of performance. However, if one has no a priori available knowledge about $\xtilde$ then adapting $r_{socp}$ beyond $r_{socp}^{(opt)}$ would be hard.

\textbf{\underline{\emph{4) $\frac{Ef_{obj}}{\sqrt{n}}=\frac{E\xi_{prim}^{(dep)}(\sigma,\g,\h,x_{mag},r_{socp})}{\sqrt{n}}$ as a function of $x_{mag}^{(sc)}$; varying $r_{socp}$}}}

Similarly to what was done above in part 2) one can also determine the theoretical predictions for $\frac{Ef_{obj}}{\sqrt{n}}=\frac{E\xi_{prim}^{(dep)}}{\sqrt{n}}$ for a varying $r_{socp}$. As in parts 2) and 3) above, we restrict our attention only to the medium $\alpha=0.5$ regime. We also assume exactly the same scenarios as in part 3). The obtained results are shown in Figure \ref{fig:objvarxvarrho}.
\begin{figure}[htb]
\begin{minipage}[b]{.5\linewidth}
\centering
\centerline{\epsfig{figure=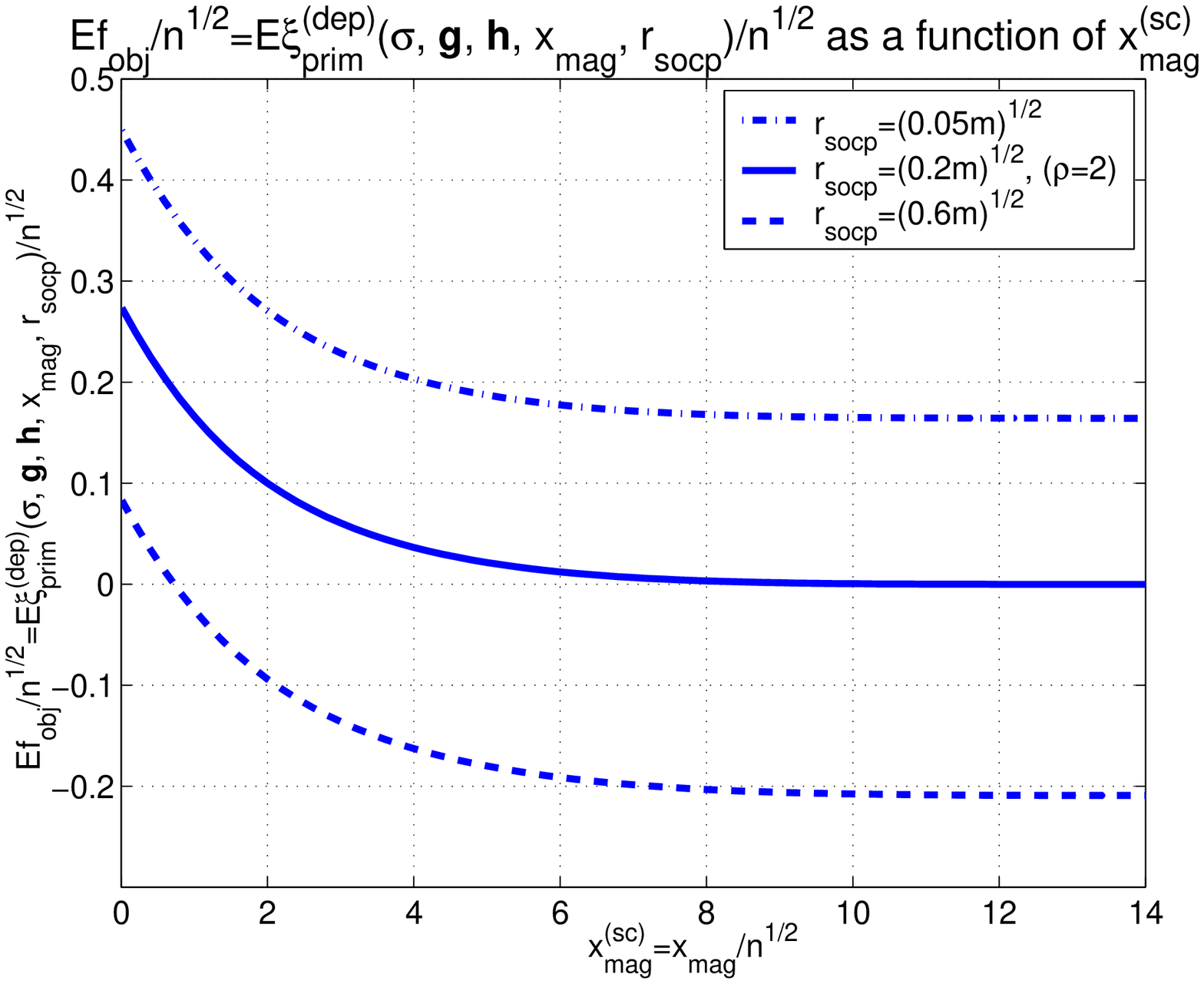,width=8cm,height=7cm}}
\end{minipage}
\begin{minipage}[b]{.5\linewidth}
\centering
\centerline{\epsfig{figure=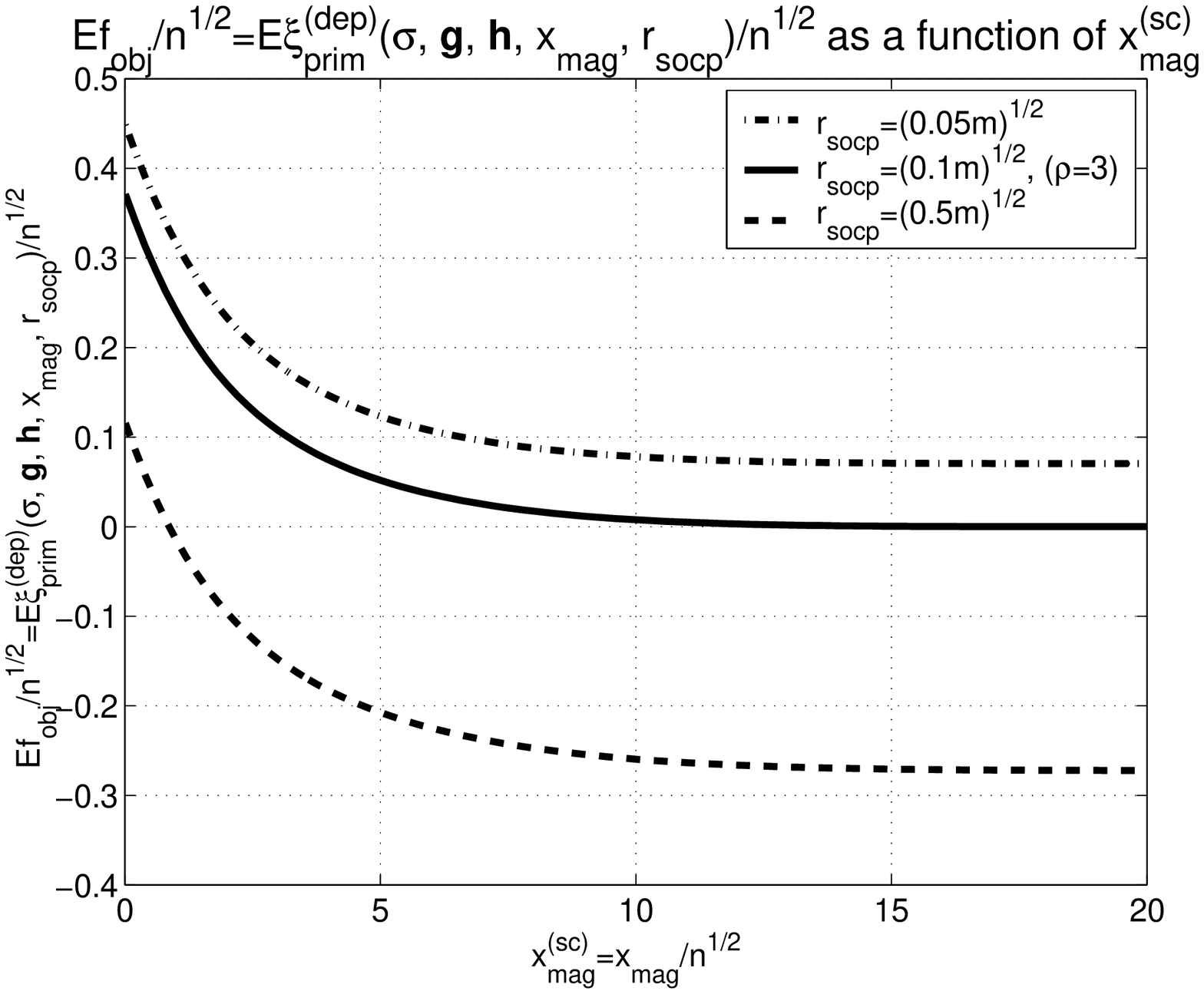,width=8cm,height=7cm}}
\end{minipage}
\caption{$\frac{Ef_{obj}}{\sqrt{n}}=\frac{E\xi_{prim}^{(dep)}(\sigma,\g,\h,x_{mag},r_{socp})}{\sqrt{n}}$ as a function of $x_{mag}^{(sc)}$ for different $r_{socp}$; left --- $\rho=2$, $r_{socp}\in\{\sigma\sqrt{0.05\alpha n},\sigma\sqrt{0.2\alpha n}\sigma\sqrt{0.6\alpha n}\}$; right --- $\rho=3$, $r_{socp}\in\{\sigma\sqrt{0.05\alpha n},\sigma\sqrt{0.1\alpha n}\sigma\sqrt{0.5\alpha n}\}$}
\label{fig:objvarxvarrho}
\end{figure}
As in part 2) Figure \ref{fig:objvarxvarrho} shows that $\frac{Ef_{obj}}{\sqrt{n}}$ is larger for larger $\rho$. On the other hand it also shows that
$\frac{Ef_{obj}}{\sqrt{n}}$ decreases as $r_{socp}$ increases. This also follows rather trivially by the use of arguments from Section 2.4.2. We skip this easy exercise as well.

We conducted massive numerical experiments and found that the results one can get through them are in a firm agreement (as they should be) with what the presented theory predicts. In the next subsection we present a sample of the results obtained through the conducted numerical experiments.

\subsubsection{Numerical experiments} \label{sec:unsignednumexp}

Similarly to what was done in the previous subsection, we will split the presentation of the numerical results in several parts. The numerical results that we will present below are obtained by running the SOCP from (\ref{eq:socp}). To demonstrate the precision of our technique we will in parallel show the results obtained by running (\ref{eq:mainlasso3ver}). To make scaling simpler in all our numerical experiments we set $\sigma=1$.

\textbf{\underline{\emph{1) $\frac{E\|\w_{dep}\|_2}{\sigma}$ and $\frac{E\|\w_{socp}\|_2}{\sigma}$ as functions of $x_{mag}^{(sc)}$}}}

In this part we will show the numerical results that correspond to the theoretical ones given in part 1) in the previous subsection. To shorten a bit the exposition we will restrict our attention again only on the medium or $\alpha=0.5$ regime. We then set all other parameters as in the center plot of Figure \ref{fig:errorvarx} (these parameters are of course different depending if we are considering $\rho=2$ or $\rho=3$; below we will consider both of them).

\underline{\emph{a) Low $(\alpha,\beta_w)$ regime, $\rho=2$}}

We first consider the $\rho=2$ scenario. As mentioned above in our experiments we set $\alpha=0.5$, $r_{socp}=\sqrt{\frac{\alpha n}{1+\rho^2}}=\sqrt{0.2\alpha n}$, and (as shown in \cite{StojnicGenSocp10}) $\beta_w$ such that $(\alpha_w,\beta_w)$ satisfy (\ref{eq:fundl1}) and $\alpha_w=\frac{\rho^2}{1+\rho^2}\alpha$. We then ran (\ref{eq:socp}) $300$ times with $n=800$ for various $x_{mag}^{(sc)}$. In parallel we ran (\ref{eq:mainlasso3ver}) for the exact same parameters with only one difference; namely we ran (\ref{eq:mainlasso3ver}) with $n=2000$. The obtained results for $\frac{E\|\w_{socp}\|_2}{\sigma}$ and $\frac{E\|\w_{dep}\|_2}{\sigma}$ are shown on the left-hand and right-hand side of Figure \ref{fig:errorvarxsim}, respectively (given our assumption that $\sigma=1$ $\frac{E\|\w_{dep}\|_2}{\sigma}$ and $\frac{E\|\w_{socp}\|_2}{\sigma}$ are of course just $E\|\w_{dep}\|_2$ and $E\|\w_{socp}\|_2$, respectively). We also show in Figure \ref{fig:errorvarxsim} the corresponding theoretical predictions obtained in the previous subsection.
\begin{figure}[htb]
\begin{minipage}[b]{.5\linewidth}
\centering
\centerline{\epsfig{figure=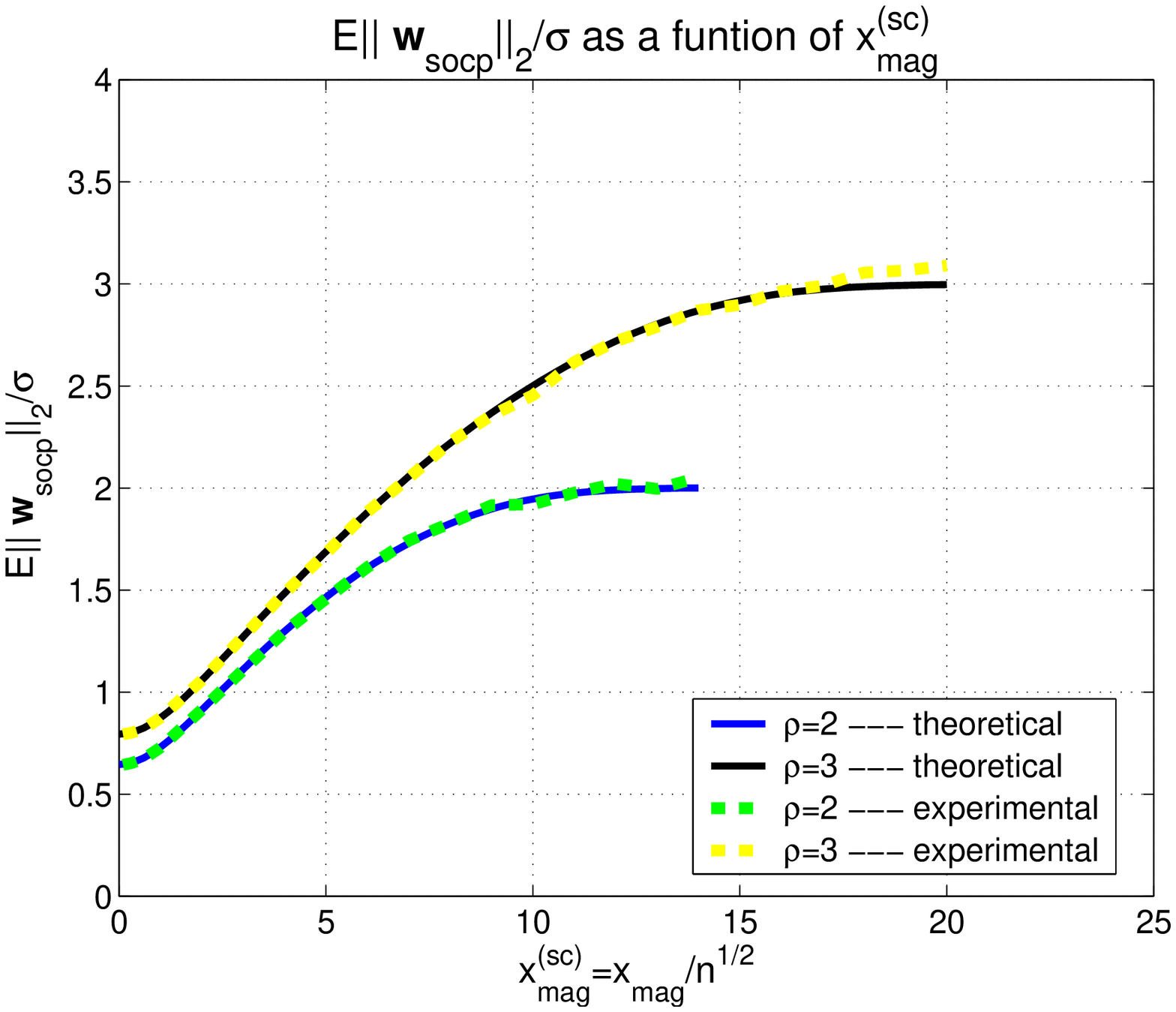,width=8cm,height=7cm}}
\end{minipage}
\begin{minipage}[b]{.5\linewidth}
\centering
\centerline{\epsfig{figure=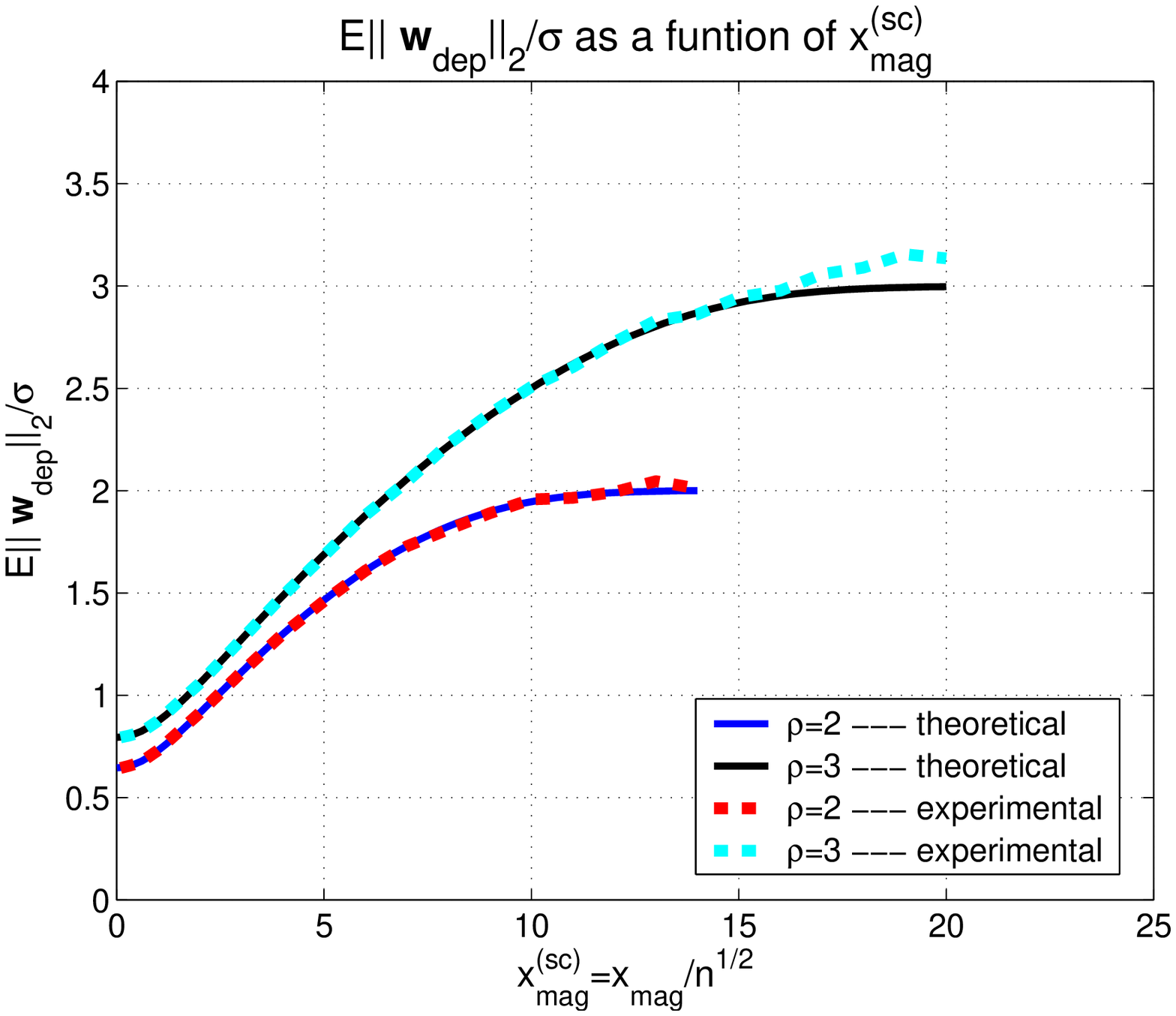,width=8cm,height=7cm}}
\end{minipage}
\caption{Experimental results for $\frac{E\|\w_{socp}\|_2}{\sigma}$ and $\frac{E\|\w_{dep}\|_2}{\sigma}$ as a function of $x_{mag}^{(sc)}$; $\rho=2$, $r_{socp}=\sqrt{0.2 \alpha n}$; $\rho=3$, $r_{socp}=\sqrt{0.1 \alpha n}$; left --- SOCP from (\ref{eq:socp}), right --- (\ref{eq:mainlasso3ver})}
\label{fig:errorvarxsim}
\end{figure}

\underline{\emph{b) High $(\alpha,\beta_w)$ regime, $\rho=3$}}

We also conducted a set of experiments in the so-called ``high" $(\alpha,\beta_w)$ regime. We used exactly the same parameters as in low $(\alpha,\beta_w)$ except that we changed $\rho$ from $2$ to $3$. Consequently we chose $r_{socp}=\sqrt{0.1 \alpha n}$ and $\beta_w$ such that $(\alpha_w,\beta_w)$ satisfy (\ref{eq:fundl1}) and $\alpha_w=\frac{\rho^2}{1+\rho^2}\alpha$. As above we ran $300$ times each (\ref{eq:socp}) and (\ref{eq:mainlasso3ver}). We ran (\ref{eq:socp}) with $n=800$ and (\ref{eq:mainlasso3ver}) with $n=2000$. The numerical results obtained for $\rho=3$ together with the theoretical predictions are again shown in Figure \ref{fig:errorvarxsim}. From Figure \ref{fig:errorvarxsim} we observe a solid agreement between the theoretical predictions and the results obtained through numerical experiments.

\textbf{\underline{\emph{2) $\frac{Ef_{obj}}{\sqrt{n}}$ and $\frac{E\xi_{prim}^{(dep)}(\sigma,\g,\h,x_{mag},r_{socp})}{\sqrt{n}}$ as functions of $x_{mag}^{(sc)}$}}}

In this part we will show the numerical results that correspond to the theoretical ones given in part 2) in the previous subsection. We then set all parameters as in Figure \ref{fig:errorvarx} (these parameters are exactly the same as in experiments whose results we just presented above). Of course we again distinguish two cases: $\rho=2$ and $\rho=3$. For both $\rho=2$ and $\rho=3$ we ran $300$ times each, (\ref{eq:socp}) and (\ref{eq:mainlasso3ver}) and again we ran (\ref{eq:socp}) with $n=800$ and (\ref{eq:mainlasso3ver}) with $n=2000$. The numerical results that we obtained for $\frac{Ef_{obj}}{\sqrt{n}}$ and $\frac{E\xi_{prim}^{(dep)}(\sigma,\g,\h,x_{mag},r_{socp})}{\sqrt{n}}$ are shown in Figure \ref{fig:objvarxsim}.
\begin{figure}[htb]
\begin{minipage}[b]{.5\linewidth}
\centering
\centerline{\epsfig{figure=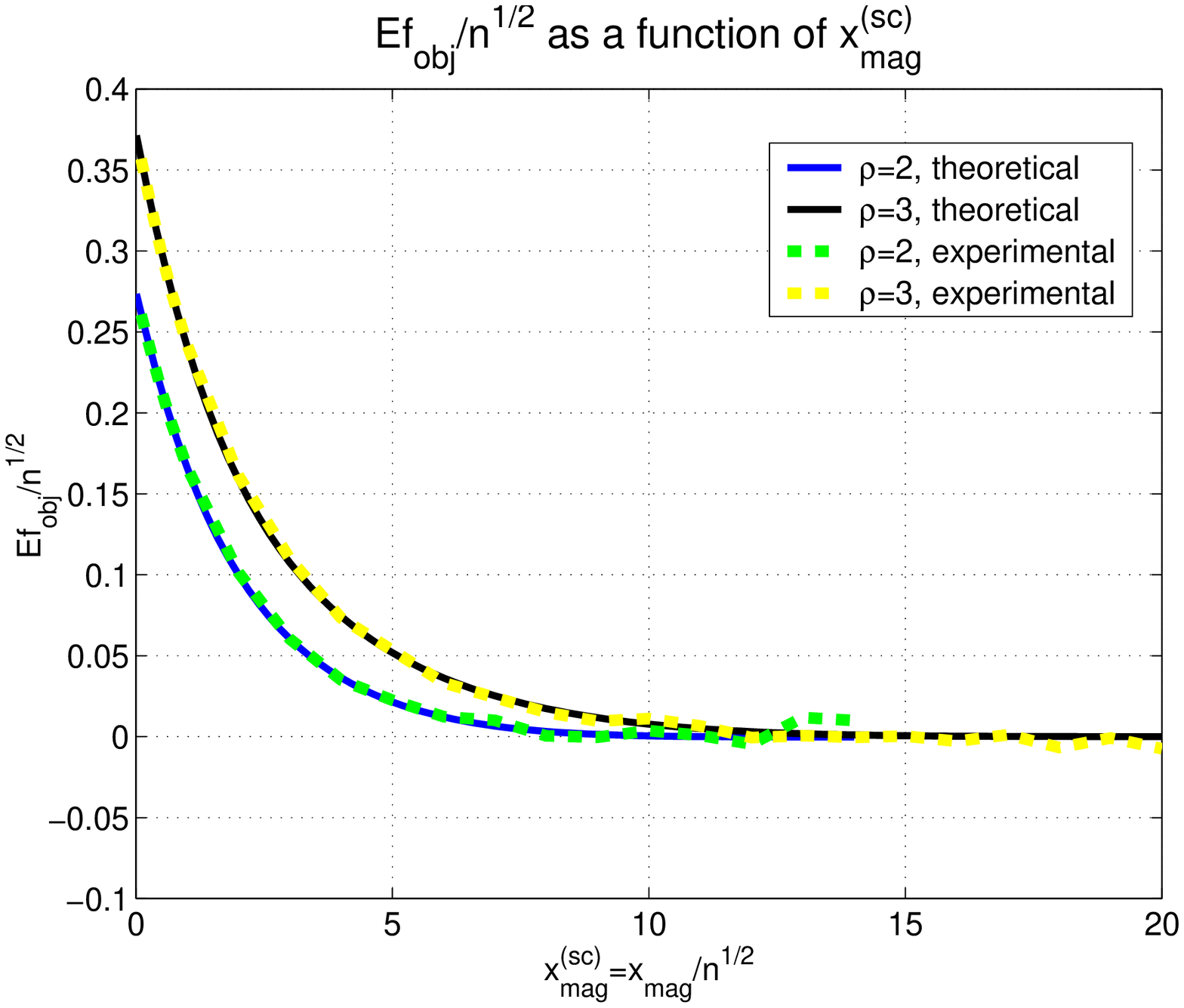,width=8cm,height=7cm}}
\end{minipage}
\begin{minipage}[b]{.5\linewidth}
\centering
\centerline{\epsfig{figure=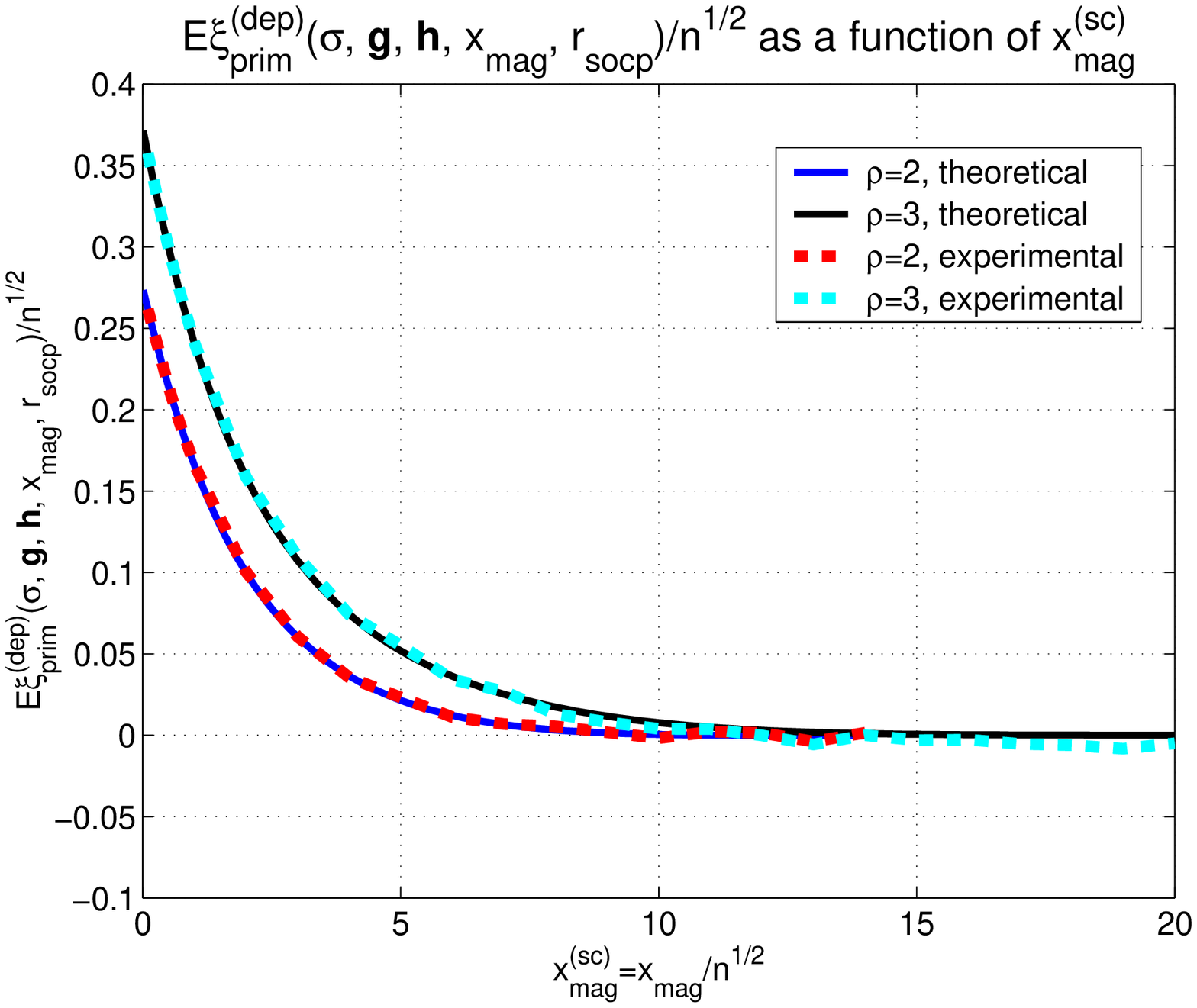,width=8cm,height=7cm}}
\end{minipage}
\caption{Experimental results for $\frac{Ef_{obj}}{\sqrt{n}}$ and $\frac{E\xi_{prim}^{(dep)}(\sigma,\g,\h,x_{mag},r_{socp})}{\sqrt{n}}$ as a function of $x_{mag}^{(sc)}$; $\rho=2$, $r_{socp}=\sqrt{0.2 \alpha n}$; $\rho=3$, $r_{socp}=\sqrt{0.1 \alpha n}$; left --- SOCP from (\ref{eq:socp}); right --- (\ref{eq:mainlasso3ver})}
\label{fig:objvarxsim}
\end{figure}
We again observe a solid agreement between the theoretical predictions and the results obtained through numerical experiments.

\textbf{\underline{\emph{3) $\frac{E\|\w_{dep}\|_2}{\sigma}$ and $\frac{E\|\w_{socp}\|_2}{\sigma}$ as functions of $x_{mag}^{(sc)}$; varying $r_{socp}$}}}

In this part we will show the numerical results that correspond to the theoretical ones given in part 3) in the previous subsection. These results relate to possible variations in the $r_{socp}$ that can be used in (\ref{eq:socp}). We then set all other parameters as in Figure \ref{fig:errorvarxvarrho} (these parameters are of course again different depending if we are considering $\rho=2$ or $\rho=3$).

\underline{\emph{a) Low $(\alpha,\beta_w)$ regime, $\rho=2$}}

We first consider the $\rho=2$ scenario. As in part 1) of this subsection we set $\alpha=0.5$ and choose $\beta_w$ as in part 1). However, differently from part 1) we now consider two different possibilities for $r_{socp}$, namely $r_{socp}=\sqrt{0.05\alpha n}$ and $r_{socp}=\sqrt{0.6 \alpha n}$. We then ran (\ref{eq:socp}) $300$ times with $n=800$ for various $x_{mag}^{(sc)}$. In parallel we ran (\ref{eq:mainlasso3ver}) with $n=2000$. The obtained results for $\frac{E\|\w_{socp}\|_2}{\sigma}$ and $\frac{E\|\w_{dep}\|_2}{\sigma}$ are shown on the left-hand and right-hand side of Figure \ref{fig:errorvarxvarrhosim}, respectively (again, given our assumption that $\sigma=1$ $\frac{E\|\w_{dep}\|_2}{\sigma}$ and $\frac{E\|\w_{socp}\|_2}{\sigma}$ are of course just $E\|\w_{dep}\|_2$ and $E\|\w_{socp}\|_2$, respectively). We also show in Figure \ref{fig:errorvarxvarrhosim} the corresponding theoretical predictions obtained in the previous subsection.
\begin{figure}[htb]
\begin{minipage}[b]{.5\linewidth}
\centering
\centerline{\epsfig{figure=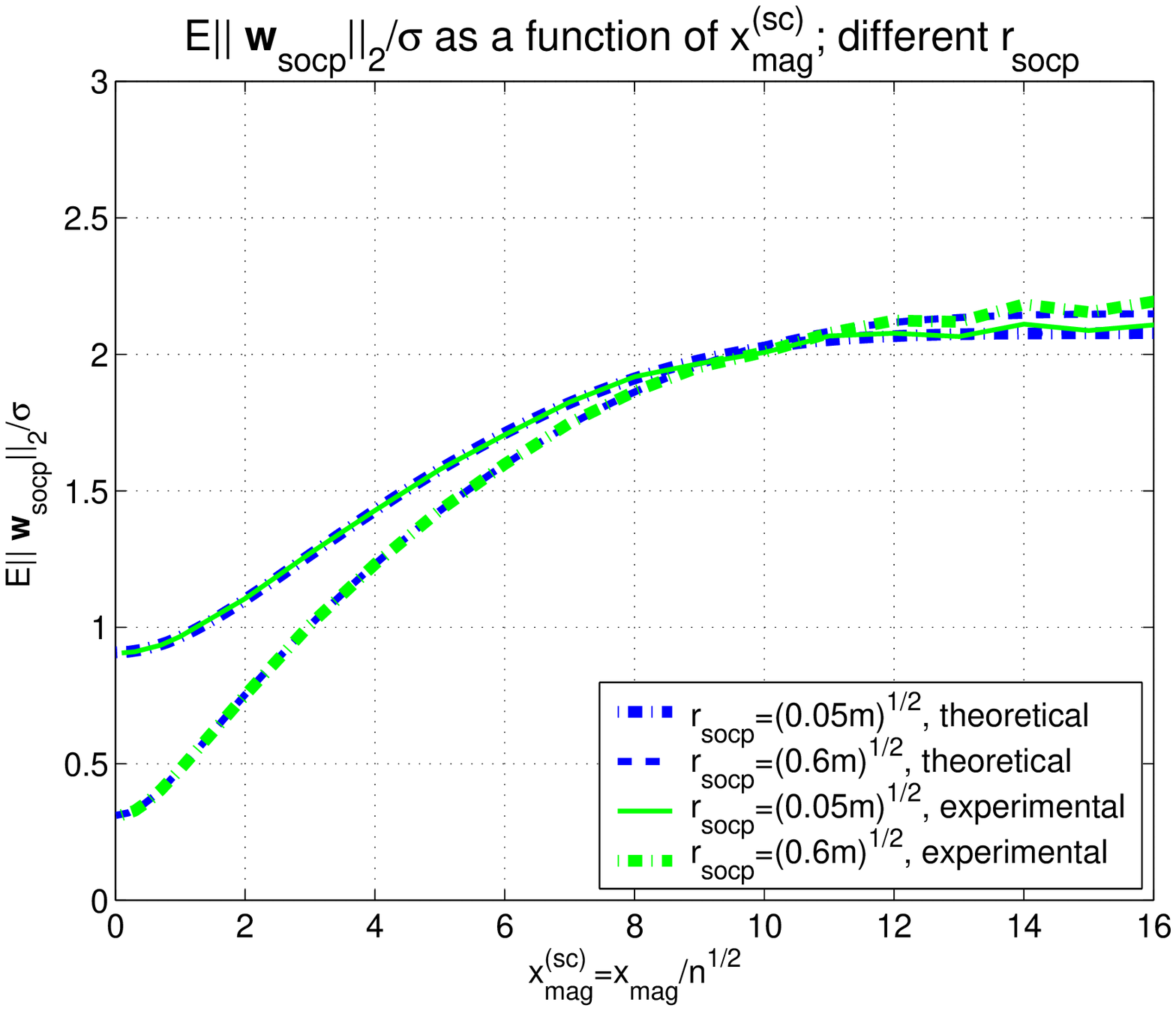,width=8cm,height=7cm}}
\end{minipage}
\begin{minipage}[b]{.5\linewidth}
\centering
\centerline{\epsfig{figure=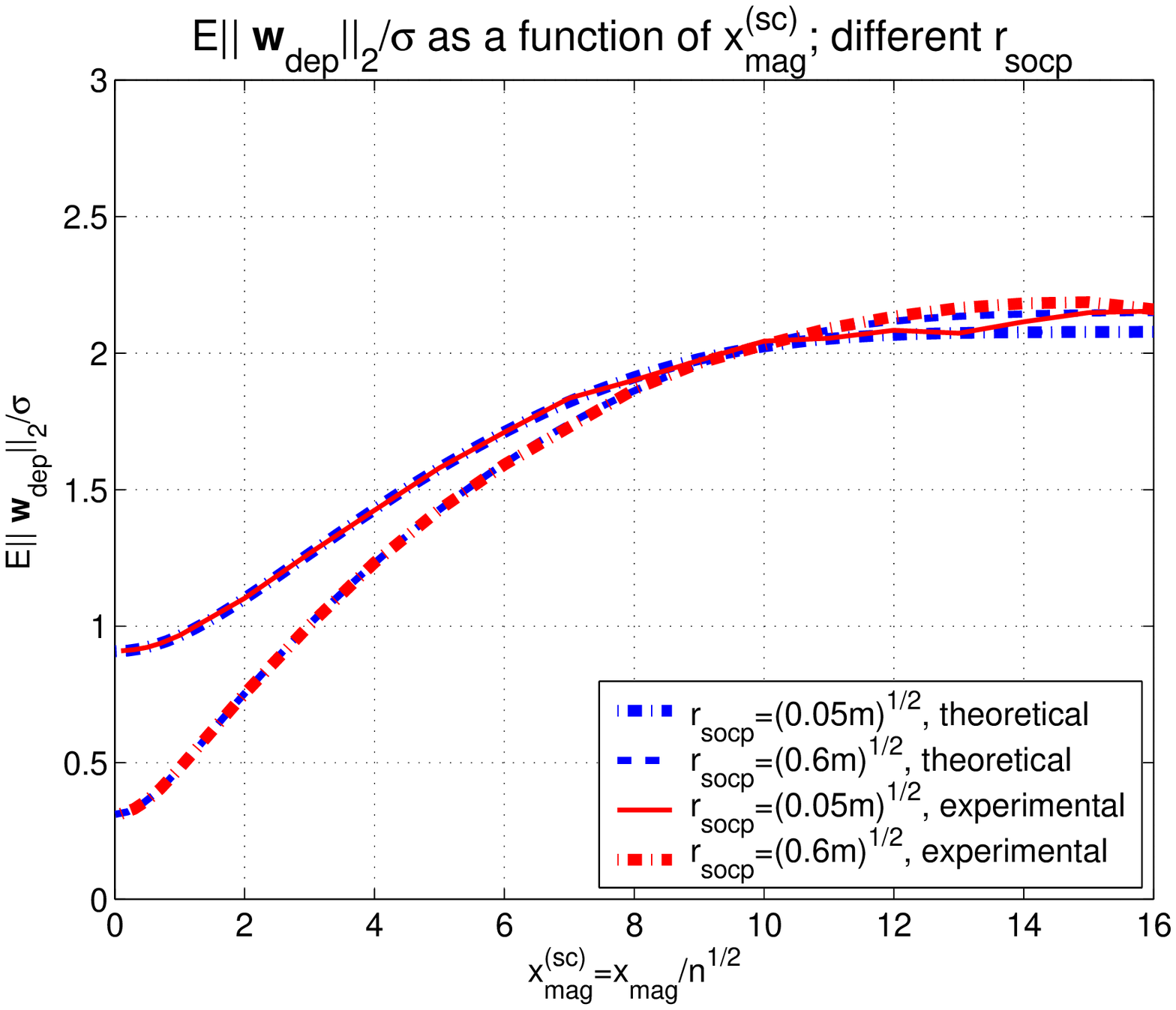,width=8cm,height=7cm}}
\end{minipage}
\caption{Experimental results for $\frac{E\|\w_{socp}\|_2}{\sigma}$ and $\frac{E\|\w_{dep}\|_2}{\sigma}$ as a function of $x_{mag}^{(sc)}$; $\rho=2$; $r_{socp}\in\{\sqrt{0.05 \alpha n},\sqrt{0.6 \alpha n}\}$; left --- SOCP from (\ref{eq:socp}), right --- (\ref{eq:mainlasso3ver})}
\label{fig:errorvarxvarrhosim}
\end{figure}

\underline{\emph{b) High $(\alpha,\beta_w)$ regime, $\rho=3$}}

We also consider the $\rho=3$ scenario. As above, we set $\alpha=0.5$ and choose $\beta_w$ as in part 1) of this subsection. Everything else remain the same as in $\rho=2$ case except the way we vary $r_{socp}$. This time we consider (as in part 3) of the previous section when $\rho=3$ case was considered) $r_{socp}=\sqrt{0.05\alpha n}$ and $r_{socp}=\sqrt{0.5 \alpha n}$. As usual (\ref{eq:socp}) was run $300$ times with $n=800$ for various $x_{mag}^{(sc)}$. In parallel we ran (\ref{eq:mainlasso3ver}) with $n=2000$. The obtained numerical results for $\frac{Ef_{obj}}{\sqrt{n}}$ and $\frac{E\xi_{prim}^{(dep)}(\sigma,\g,\h,x_{mag},r_{socp})}{\sqrt{n}}$ as well as the corresponding theoretical predictions obtained in the previous subsection are shown on the left-hand and right-hand side of Figure \ref{fig:errorvarxvarrhosim1}, respectively.
\begin{figure}[htb]
\begin{minipage}[b]{.5\linewidth}
\centering
\centerline{\epsfig{figure=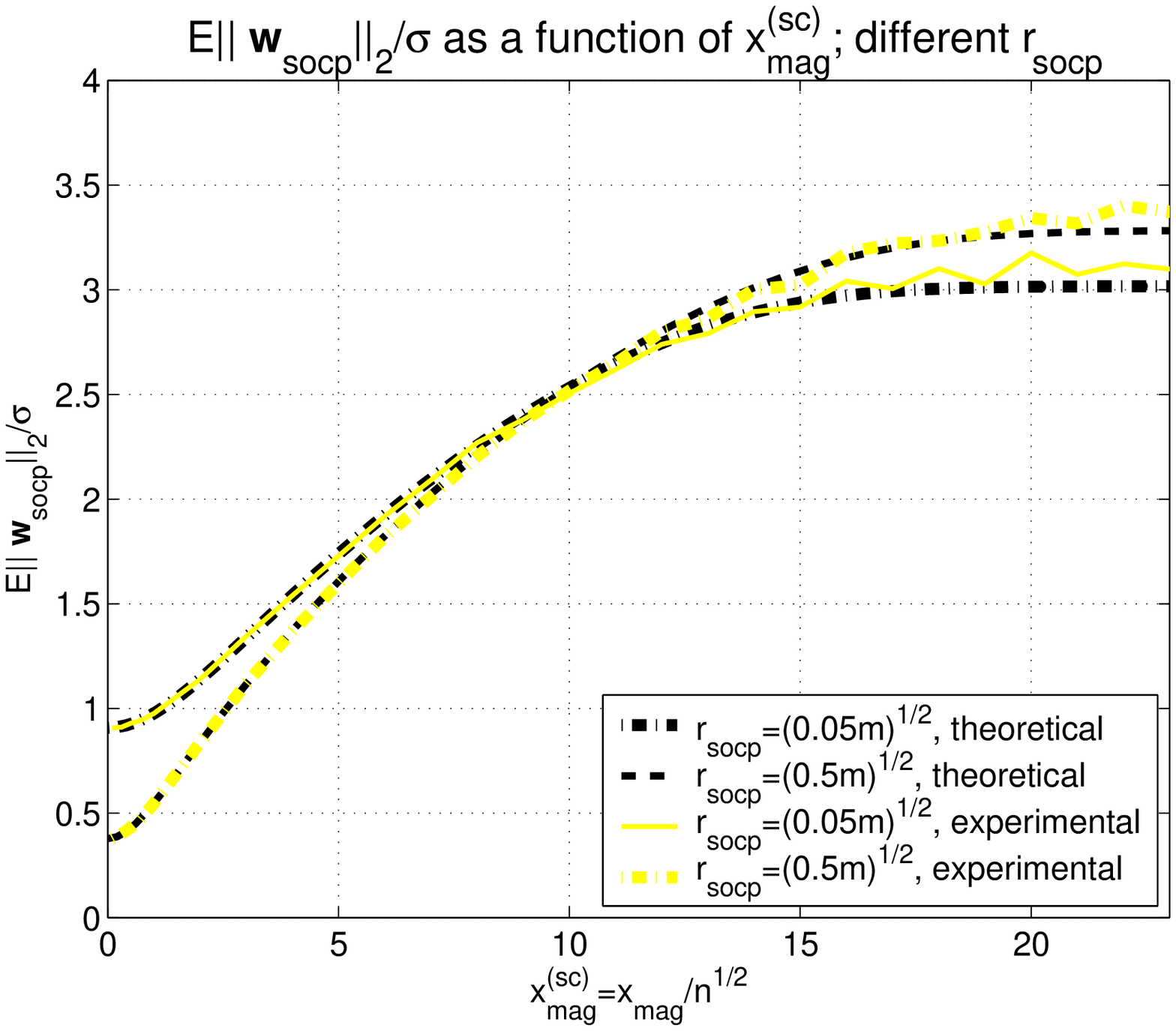,width=8cm,height=7cm}}
\end{minipage}
\begin{minipage}[b]{.5\linewidth}
\centering
\centerline{\epsfig{figure=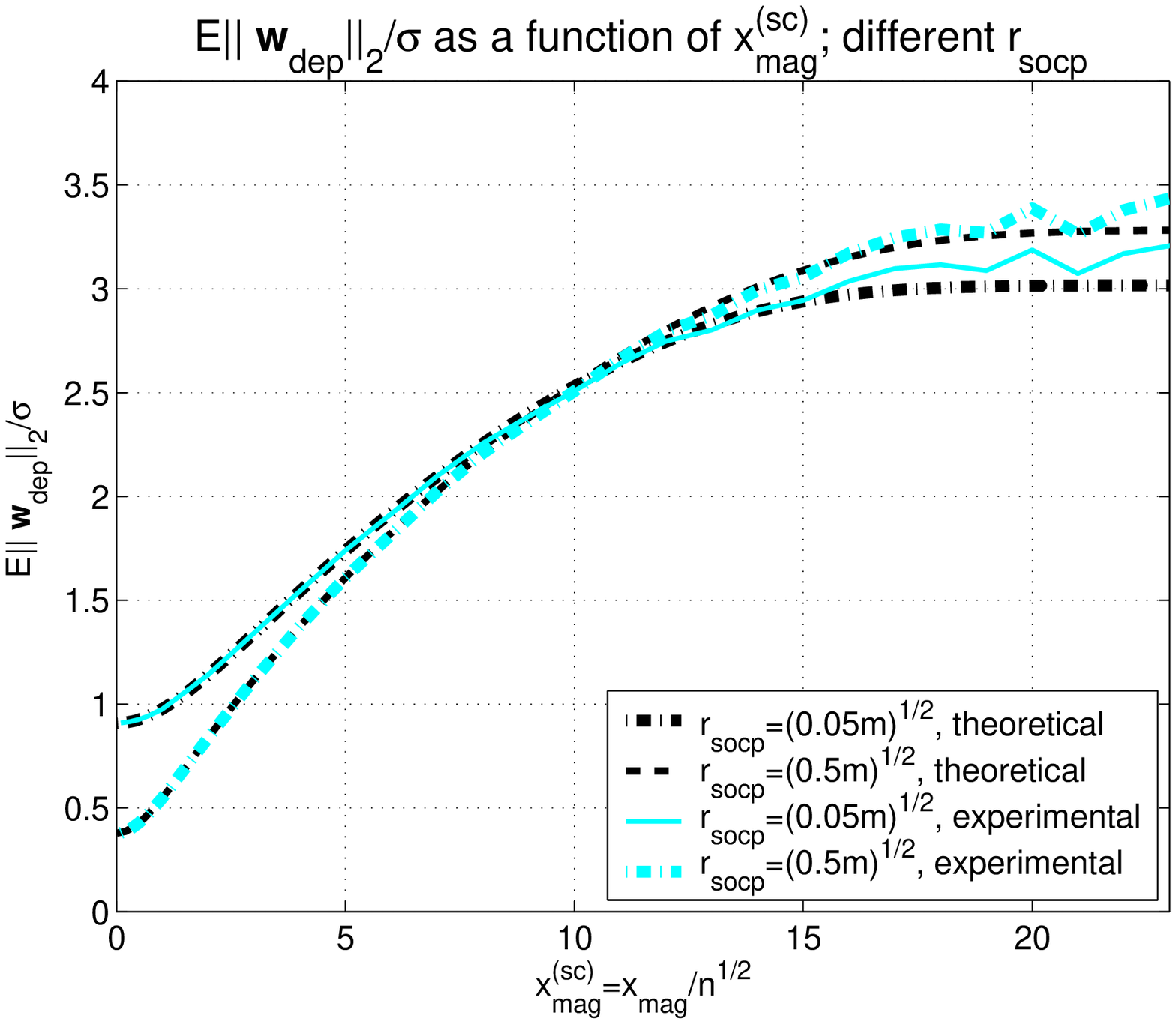,width=8cm,height=7cm}}
\end{minipage}
\caption{Experimental results for $\frac{E\|\w_{socp}\|_2}{\sigma}$ and $\frac{E\|\w_{dep}\|_2}{\sigma}$ as a function of $x_{mag}^{(sc)}$; $\rho=3$; $r_{socp}\in\{\sqrt{0.05 \alpha n},\sqrt{0.5 \alpha n}\}$; left --- SOCP from (\ref{eq:socp}), right --- (\ref{eq:mainlasso3ver})}
\label{fig:errorvarxvarrhosim1}
\end{figure}
We again observe a solid agreement between the theoretical predictions and the results obtained through numerical experiments. Small glitches that happen in large $x_{mag}^{(sc)}$ regime could have been fixed by choosing a larger $n$. We purposely chose a smaller $n$ to show that results are fairly good even when $n$ is not very large. In fact, even a smaller $n$ than the one we have chosen would work quite fine.

\textbf{\underline{\emph{4) $\frac{Ef_{obj}}{\sqrt{n}}$ and $\frac{E\xi_{prim}^{(dep)}(\sigma,\g,\h,x_{mag},r_{socp})}{\sqrt{n}}$ as functions of $x_{mag}^{(sc)}$; varying $r_{socp}$}}}

In this part we will show the numerical results that correspond to the theoretical ones given in part 4) in the previous subsection. These results relate to behavior of $\frac{Ef_{obj}}{\sqrt{n}}$ and $\frac{E\xi_{prim}^{(dep)}(\sigma,\g,\h,x_{mag},r_{socp})}{\sqrt{n}}$ when one varies $r_{socp}$ in (\ref{eq:socp}). We again consider $\rho=2$ or $\rho=3$.

\underline{\emph{a) Low $(\alpha,\beta_w)$ regime, $\rho=2$}}

The setup that we consider is exactly the same as the one considered in part 3a) of this subsection. We set $\alpha=0.5$, choose $\beta_w$ as in part 1), and considered two different possibilities for $r_{socp}$, namely $r_{socp}=\sqrt{0.05\alpha n}$ and $r_{socp}=\sqrt{0.6 \alpha n}$.  The obtained results for $\frac{Ef_{obj}}{\sqrt{n}}$ and $\frac{E\xi_{prim}^{(dep)}(\sigma,\g,\h,x_{mag},r_{socp})}{\sqrt{n}}$ are shown on the left-hand and right-hand side of Figure \ref{fig:objvarxvarrhosim}, respectively. The corresponding theoretical predictions obtained in the previous subsection are also shown in Figure \ref{fig:objvarxvarrhosim}.
\begin{figure}[htb]
\begin{minipage}[b]{.5\linewidth}
\centering
\centerline{\epsfig{figure=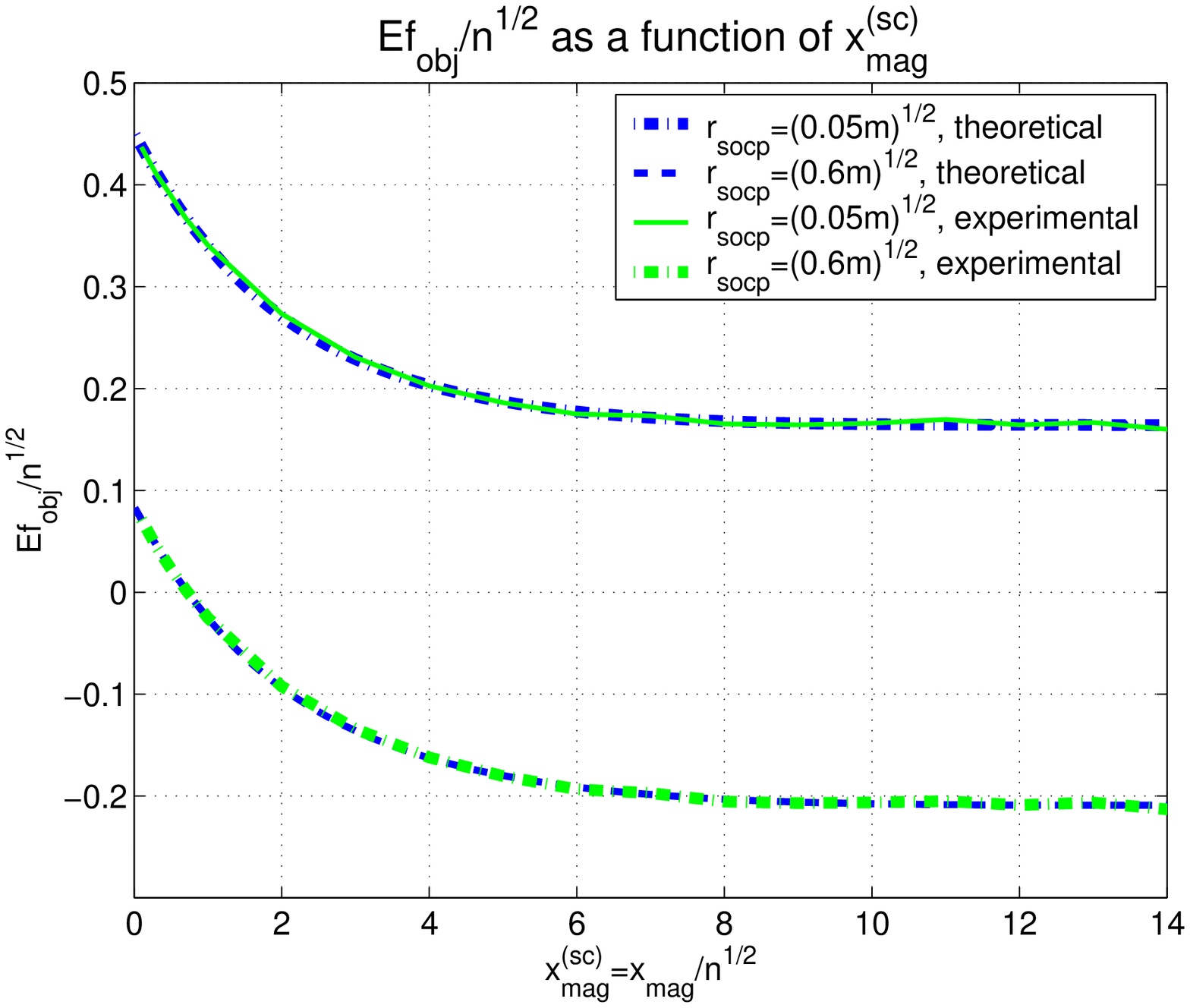,width=8cm,height=7cm}}
\end{minipage}
\begin{minipage}[b]{.5\linewidth}
\centering
\centerline{\epsfig{figure=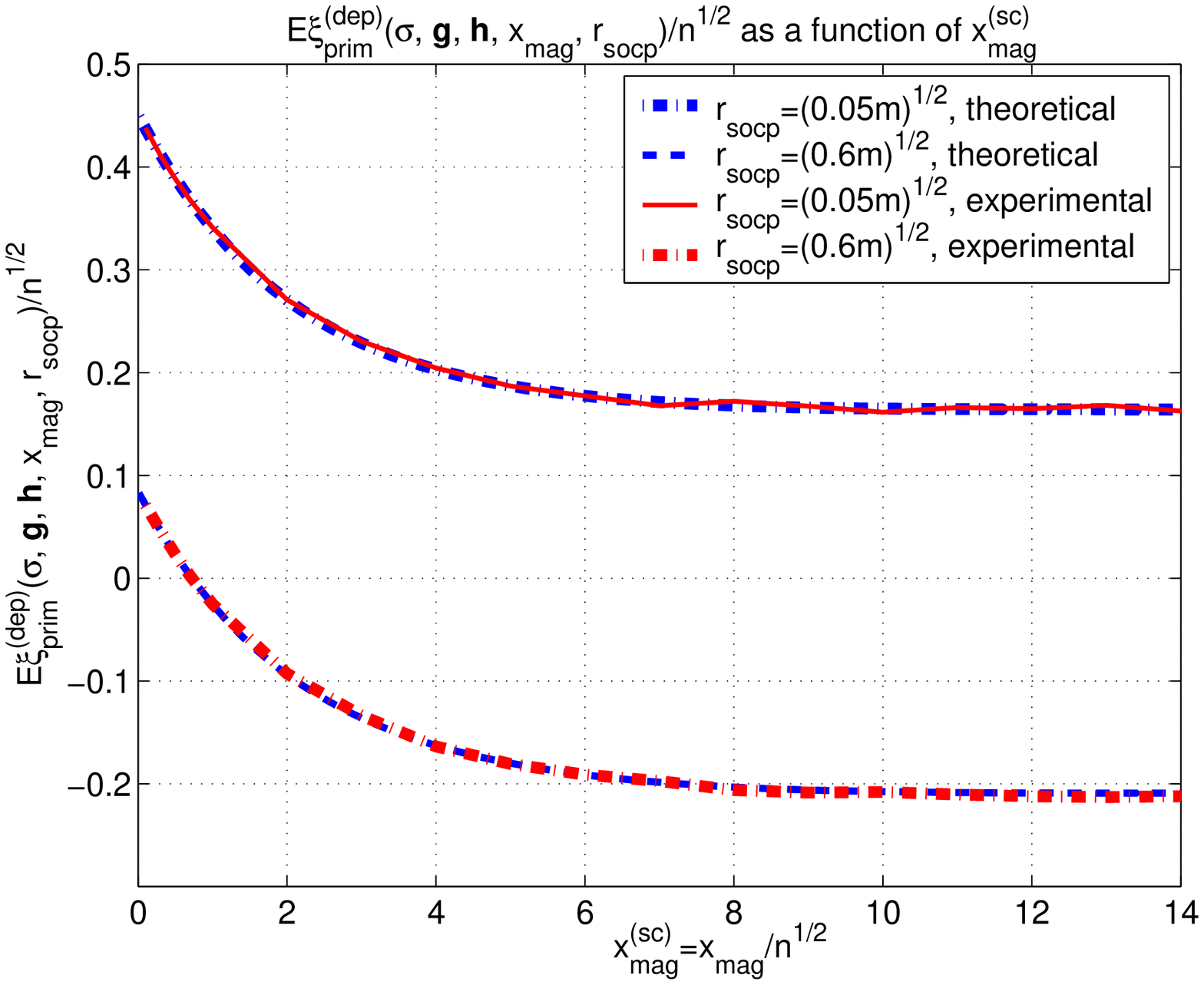,width=8cm,height=7cm}}
\end{minipage}
\caption{Experimental results for $\frac{Ef_{obj}}{\sqrt{n}}$ and $\frac{E\xi_{prim}^{(dep)}(\sigma,\g,\h,x_{mag},r_{socp})}{\sqrt{n}}$ as a function of $x_{mag}^{(sc)}$; $\rho=2$; $r_{socp}\in\{\sqrt{0.05 \alpha n},\sqrt{0.6 \alpha n}\}$; left --- SOCP from (\ref{eq:socp}), right --- (\ref{eq:mainlasso3ver})}
\label{fig:objvarxvarrhosim}
\end{figure}

\underline{\emph{b) High $(\alpha,\beta_w)$ regime, $\rho=3$}}

The setup that we consider is exactly the same as the one considered in part 3b) of this subsection. We set $\alpha=0.5$, choose $\beta_w$ as in part 1), and considered two different possibilities for $r_{socp}$, namely $r_{socp}=\sqrt{0.05\alpha n}$ and $r_{socp}=\sqrt{0.5 \alpha n}$.  The obtained results for $\frac{E\|\w_{socp}\|_2}{\sigma}$ and $\frac{E\|\w_{dep}\|_2}{\sigma}$ are shown on the left-hand and right-hand side of Figure \ref{fig:objvarxvarrhosim1}, respectively. The corresponding theoretical predictions obtained in the previous subsection are also shown in Figure \ref{fig:objvarxvarrhosim1}.
\begin{figure}[htb]
\begin{minipage}[b]{.5\linewidth}
\centering
\centerline{\epsfig{figure=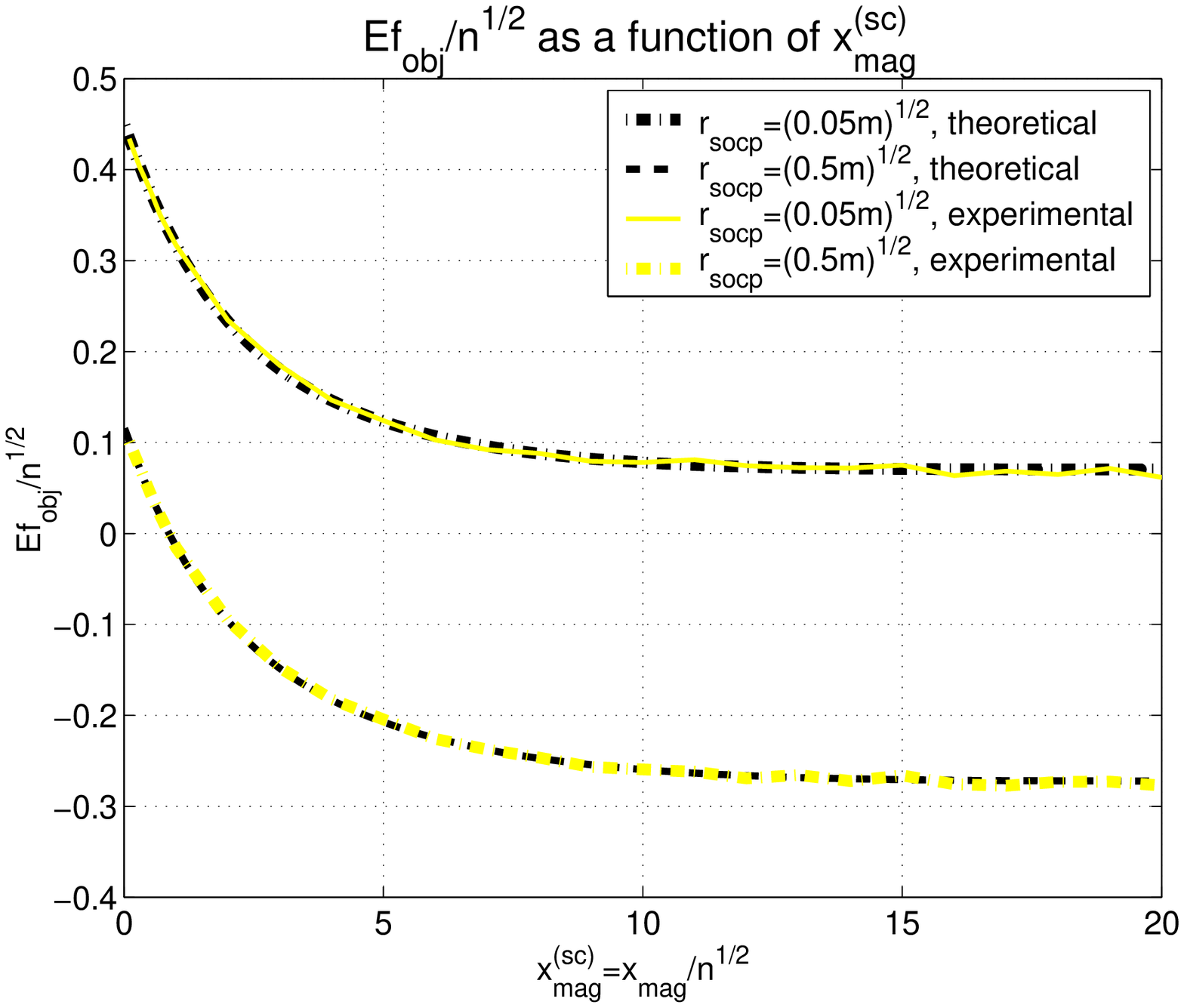,width=8cm,height=7cm}}
\end{minipage}
\begin{minipage}[b]{.5\linewidth}
\centering
\centerline{\epsfig{figure=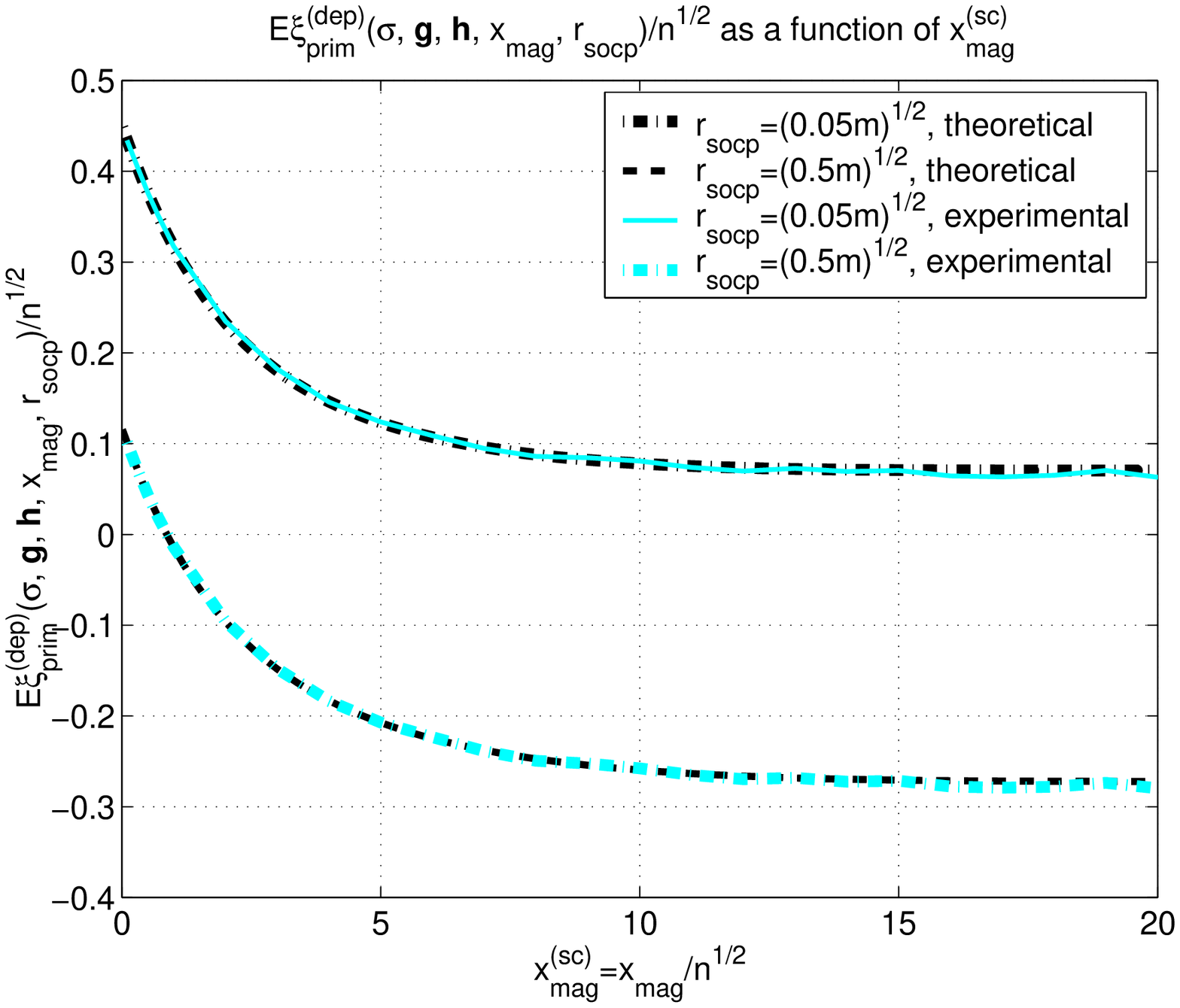,width=8cm,height=7cm}}
\end{minipage}
\caption{Experimental results for $\frac{Ef_{obj}}{\sqrt{n}}$ and $\frac{E\xi_{prim}^{(dep)}(\sigma,\g,\h,x_{mag},r_{socp})}{\sqrt{n}}$ as a function of $x_{mag}^{(sc)}$; $\rho=3$; $r_{socp}\in\{\sqrt{0.05 \alpha n},\sqrt{0.5 \alpha n}\}$; left --- SOCP from (\ref{eq:socp}), right --- (\ref{eq:mainlasso3ver})}
\label{fig:objvarxvarrhosim1}
\end{figure}
We again observe a solid agreement between the theoretical predictions and the results obtained through numerical experiments.

\section{SOCP's problem dependent performance -- signed $\x$} \label{sec:signed}

In this section we show how the SOCP's problem dependent performance analysis developed in the previous section can be specialized to the case when signals are \emph{a priori} known to have nonzero components of certain sign.

\subsection{Basic properties of the SOCP's framework} \label{sec:signedbasic}

All major assumptions stated at the beginning of the previous section will continue to hold in this section as well; namely, we will continue to consider matrices $A$ with i.i.d. standard normal random variables; elements of $\v$ will again be i.i.d. Gaussian random variables with zero mean and variance $\sigma$. The main difference, though, comes in the definition of $\xtilde$. We will in this section assume that $\xtilde$ is the original $\x$ in (\ref{eq:systemnoise}) that we are trying to recover and that it is \emph{any} $k$-sparse vector with a given fixed location of its nonzero elements and with a priori known signs of its elements. Given the statistical context,
it will be fairly easy to see later on that everything that we will present in this section will be irrelevant with respect to what particular location and what particular combination of signs of nonzero elements are chosen. We therefore for the simplicity of the exposition and without loss of generality assume that the components $\x_{1},\x_{2},\dots,\x_{n-k}$ of $\x$ are equal to zero and that the remaining components of $\x$, $\x_{n-k+1},\x_{n-k+2},\dots,\x_n$, are greater than or equal to zero. However, differently from what was assumed in the previous section, we now assume that this information is \emph{a priori} known. That essentially means that
this information is also known to the solving algorithm. Then instead of (\ref{eq:socp}) one can consider its a better (``signed") version
\begin{eqnarray}
\min_{\x} & &  \|\x\|_1\nonumber \\
\mbox{subject to} & & \|\y-A\x\|_2\leq r_{socp+}\nonumber \\
& & \x_{i}\geq 0,1\leq i\leq n.
\label{eq:socpnon}
\end{eqnarray}
Also, one should again note that $r_{socp+}$ in (\ref{eq:socpnon}) is a parameter that critically impacts the outcome of any SOCP type of algorithm (again, for different $r_{socp+}$'s one will have different SOCP's). The analysis that we will present assumes a general $r_{socp+}$. However, we do mention right here that problem (\ref{eq:socpnon}) is not feasible for all choices of $\xtilde$, $\alpha$, $\beta_w^+$, $\sigma$, and $r_{socp+}$. Unless mentioned otherwise what we present below assumes that $\xtilde$, $\alpha$, $\beta_w^+$, $\sigma$, and $r_{socp+}$ are such that (\ref{eq:socpnon}) is feasible with overwhelming probability. For example, as discussed in \cite{StojnicGenSocp10}, a statistical choice $r_{socp+}>\sigma\sqrt{m}$ guarantees feasibility with overwhelming probability. Of course, there are other choices of parameters $\xtilde$, $\alpha$, $\beta_w^+$, $\sigma$, and $r_{socp+}$ that guarantee feasibility as well. Towards the end of this section we will mention some of them and address this general question of feasibility in more detail.

Given the positivity of $\x_i,1\leq i\leq n$, one can, of course, replace $\ell_1$ norm in the objective of (\ref{eq:socpnon}) by the sum of all elements of $\x$. However, to maintain visual similarity between what we will present in this section and what we presented in Section \ref{sec:unsigned} we will keep the $\ell_1$ norm in the objective. Along the same lines, in what follows we will try to mimic the procedure presented in the previous section as much as possible. On such a path we will skip all the obvious parallels and emphasize the points that are different.

As a first step in making the presentation of the ``signed" case as parallel as possible to the presentation of the ``general" case we will again introduce a few definitions that will turn out to be helpful in what follows. First, let us define the optimal value of a slightly changed objective from (\ref{eq:socpnon}) in the following way
\begin{eqnarray}
f_{obj+}=\min_{\x} & & \|\x\|_1-\|\xtilde\|_1\nonumber \\
\mbox{subject to} & & \|\y-A\x\|_2\leq r_{socp+}\nonumber \\
& & \x_{i}\geq 0,1\leq i\leq n. \label{eq:objlassol1non}
\end{eqnarray}
As in the previous section, $f_{obj+}$ is clearly a function of $\sigma,\xtilde,A,\v$, and $r_{socp+}$. To make writing easier we will adopt the same convention as in Section \ref{sec:unsigned} and omit them. Also let $\x_{socp+}$ be the solution of (\ref{eq:socpnon}) (or the solution of (\ref{eq:objlassol1non})) and let $\w_{socp+}\in R^n$  be the so-called error vector defined in the following way
\begin{equation}
\w_{socp+}=\x_{socp+}-\xtilde.\label{eq:xhatdefnon}
\end{equation}
As in Section \ref{sec:unsigned} our main goal in this section will be to provide various characterizations of $\w_{socp+}$ and $f_{obj+}$. Throughout the paper we will heavily rely on the following theorem from \cite{StojnicGenSocp10} that provides a general characterization of $\w_{socp+}$ and $f_{obj+}$.

\begin{theorem}(\cite{StojnicGenSocp10} --- SOCP's performance characterization; signed $\x$)
Let $\v$ be an $n\times 1$ vector of i.i.d. zero-mean variance $\sigma^2$ Gaussian random variables and let $A$ be an $m\times n$ matrix of i.i.d. standard normal random variables. Further, let $\g$ and $\h$ be $m\times 1$ and $n\times 1$ vectors of i.i.d. standard normals, respectively and let $\z$ be $n\times 1$ vector of all ones. Consider a $k$-sparse $\xtilde$ defined in (\ref{eq:xtildedef}) and a $\y$ defined in (\ref{eq:systemnoise}) for $\x=\xtilde$. Let the solution of (\ref{eq:socpnon}) be $\x_{socp+}$ and let the so-called error vector of the SOCP from (\ref{eq:socpnon}) be $\w_{socp+}=\x_{socp+}-\xtilde$. Let $r_{socp+}$ in (\ref{eq:socpnon}) be a positive scalar. Let $n$ be large and let constants $\alpha=\frac{m}{n}$ and $\beta_w^+=\frac{k}{n}$ be below the following so-called \emph{signed fundamental} characterization of $\ell_1$ optimization
\begin{equation}
(1-\beta_w^+)\frac{\sqrt{\frac{1}{2\pi}}e^{-(\erfinv(2\frac{1-\alpha_w^+}{1-\beta_w^+}-1))^2}}{\alpha_w^+}-\sqrt{2}\erfinv (2\frac{1-\alpha_w^+}{1-\beta_w^+}-1)=0.
\label{eq:fundl1non}
\end{equation}
Furthermore, let $\xtilde$, $\alpha$, $\beta_w^+$, $\sigma$, and $r_{socp+}$ be such that (\ref{eq:socpnon}) is feasible with overwhelming probability and $E\xi_{prim+}(\sigma,\g,\h,\xtilde,r_{socp+})$ defined below is finite.
Consider the following optimization problem:
\begin{eqnarray}
\xi_{prim+}(\sigma,\g,\h,\xtilde,r_{socp+})=\max_{\nu,\lambda} & & \sigma\sqrt{\|\g\|_2^2\nu^2-\|\nu\h+\z-\lambda\|_2^2} -\sum_{i=n-k+1}^{n}\lambda_i\xtilde_i-\nu r_{socp+}\nonumber \\
\mbox{subject to}
& & \nu\geq 0\nonumber \\
& & \lambda_i\geq 0,1\leq i\leq n.\label{eq:mainlasso1non}
\end{eqnarray}
Let $\widehat{\nu^+}$ and $\widehat{\lambda^{+}}$ be the solution of (\ref{eq:mainlasso1non}). Set
\begin{equation}
\|\widehat{\w^+}\|_2=\sigma\frac{\|\widehat{\nu^+}\h+\z-\widehat{\lambda^+}\|_2}
{\sqrt{\|\g\|_2^2\widehat{\nu^+}^2-\|\widehat{\nu^+}\h+\z-\widehat{\lambda^+}\|_2^2}}.\label{eq:mainlasso2non}
\end{equation}
Then:
\begin{multline}
P(\|\xtilde+\w_{socp+}\|_1-\|\xtilde\|_1
\in (E\xi_{prim+}(\sigma,\g,\h,\xtilde,r_{socp+}))-\epsilon_1^{(socp)}|E\xi_{prim+}(\sigma,\g,\h,\xtilde,r_{socp+}))|,\\
E\xi_{prim+}(\sigma,\g,\h,\xtilde,r_{socp+}))+\epsilon_1^{(socp)}|E\xi_{prim+}(\sigma,\g,\h,\xtilde,r_{socp+}))|)=1-e^{-\epsilon_2^{(socp)}n}\label{eq:mainlasso3non}
\end{multline}
and
\begin{equation}
P((1-\epsilon_1^{(socp)})E\|\widehat{\w^+}\|_2\leq \|\w_{socp+}\|_2
\leq (1+\epsilon_1^{(socp)})E\|\widehat{\w^+}\|_2) =1-e^{-\epsilon_2^{(socp)}n},\label{eq:mainlasso4non}
\end{equation}
where $\epsilon_1^{(socp)}>0$ is an arbitrarily small constant and $\epsilon_2^{(socp)}$ is a constant dependent on $\epsilon_1^{(socp)}$ and $\sigma$ but independent of $n$.
\label{thm:mainlassonon}
\end{theorem}
\begin{proof}
Presented in \cite{StojnicGenSocp10}.
\end{proof}

\noindent \textbf{Remark:} A pair $(\alpha,\beta_w^+)$ lies below the signed fundamental characterization (\ref{eq:fundl1non}) if $\alpha>\alpha_w^+$ and $\alpha_w^+$ and $\beta_w^+$ are such that (\ref{eq:fundl1non}) holds.

\subsection{Problem dependent properties of the framework} \label{sec:probdepframenon}

To facilitate the exposition that will follow we similarly to what was done in Section \ref{sec:unsigned} (and earlier in \cite{StojnicCSetam09,StojnicEquiv10,StojnicGenSocp10}) set
\begin{equation}
\bar{\h}^+=[\h_{(1)}^{(1)},\h_{(2)}^{(2)},\dots,\h_{(n-k)}^{(n-k)},\h_{n-k+1}^{(k)},\h_{n-k+2}^{(k-1)},\dots,\h_n^{(1)}]^T,\label{eq:defhbarnon}
\end{equation}
where $[\h_{(1)}^{(1)},\h_{(2)}^{(2)},\dots,\h_{(n-k)}^{(n-k)}]$ are the elements of $[-\h_{1},-\h_{2},\dots,-\h_{n-k}]$ sorted in increasing order and
$[\h_{n-k+1}^{(k)},\h_{n-k+2}^{(k-1)},\dots,\h_n^{(1)}]$ are the elements of $[-\h_{n-k+1},-\h_{n-k+2},\dots,-\h_n]$ sorted in decreasing order (possible ties in the sorting processes are of course broken arbitrarily). One can then rewrite the optimization problem from (\ref{eq:mainlasso1non}) in the following way
\begin{eqnarray}
\xi_{prim+}(\sigma,\g,\h,\xtilde,r_{socp+})=\max_{\nu,\lambda} & & \sigma\sqrt{\|\g\|_2^2\nu^2-\|\nu\bar{\h}^+-\z+\lambda\|_2^2} -\sum_{i=n-k+1}^{n}\lambda_i\xtilde_i-\nu r_{socp+}\nonumber \\
\mbox{subject to}
& & \nu\geq 0\nonumber \\
& & \lambda_i\geq 0, 1\leq i\leq n.\label{eq:mainlasso3non}
\end{eqnarray}
In what follows we will restrict our attention to a specific class of unknown vectors $\xtilde$. Namely, we will consider vectors $\xtilde$ that have amplitude of the nonzero components equal to $x_{mag}$. In the noiseless case these problem instances are typically the hardest to solve (at least as long as one uses the signed version of the $\ell_1$ optimization from (\ref{eq:l1})). We will again emphasize that the fact that magnitudes of the nonzero elements of $\xtilde$ are $x_{mag}$ is not known a priori and can not be used in the solving algorithm (i.e. one can not add constraints that would exploit this knowledge in optimization problem (\ref{eq:socpnon})). It is just that we will consider how the SOCP from (\ref{eq:socpnon}) behaves when used to solve problem instances generated by such an $\xtilde$. Also, for such an $\xtilde$ (\ref{eq:mainlasso3non}) can be rewritten in the following way
\begin{eqnarray}
\xi_{prim+}^{(dep)}(\sigma,\g,\h,x_{mag},r_{socp+})=\max_{\nu,\lambda} & & \sigma\sqrt{\|\g\|_2^2\nu^2-\|\nu\bar{\h}^+-\z+\lambda\|_2^2} -x_{mag}\sum_{i=n-k+1}^{n}\lambda_i-\nu r_{socp+}\nonumber \\
\mbox{subject to}
& & \nu\geq 0\nonumber \\
& & \lambda_i\geq 0, 1\leq i\leq n.\label{eq:mainlasso3vernon}
\end{eqnarray}
Now, let $\nu_{dep+}$ and $\lambda^{(dep+)}$ be the solution of (\ref{eq:mainlasso3vernon}). Then analogously to (\ref{eq:mainlasso2non}) we can set
\begin{equation}
\|\w_{dep+}\|_2=\sigma\frac{\|\nu_{dep+}\bar{\h}^+-\z+\lambda^{(dep+)}\|_2}
{\sqrt{\|\g\|_2^2\nu_{dep+}^2-\|\nu_{dep+}\bar{\h}^+-\z+\lambda^{(dep+)}\|_2^2}},\label{eq:mainlasso4non}
\end{equation}
In what follows we will determine $\|\w_{dep+}\|_2$ and $\xi_{prim+}^{(dep)}(\sigma,\g,\h,x_{mag},r_{socp+})$ or more precisely their concentrating points
$E\|\w_{dep+}\|_2$ and $E\xi_{prim+}^{(dep)}(\sigma,\g,\h,x_{mag},r_{socp+})$. All other parameters such as $\nu_{dep+}$, $\lambda^{(dep+)}$ can (and some of them will) be computed through the framework as well.

We proceed by following the line of thought presented in Section \ref{sec:unsigned} and earlier in \cite{StojnicCSetam09,StojnicGenSocp10}. Since $\lambda^{(dep+)}$ is the solution of (\ref{eq:mainlasso3vernon}) there will be parameters $c_{1}^+$ and $c_{2}^+$ such that $$\lambda^{(dep+)}=[\lambda_1^{(dep+)},\lambda_2^{(dep+)},\dots,\lambda_{c_{1}^+}^{(dep+)},0,0,\dots,0,\lambda_{c_{2}^++1}^{(dep+)},\lambda_{c_{2}^++2}^{(dep+)},
\dots,\lambda_{n}^{(dep+)}]$$ and obviously $c_{1}^+\leq n-k$ and $n-k\leq c_{2}^+\leq n$. At this point let us assume that these parameters are known and fixed. Then following \cite{StojnicCSetam09,StojnicGenSocp10} as well as what was presented in Section \ref{sec:unsigned} the optimization problem from (\ref{eq:mainlasso3vernon}) can be rewritten in the following way
\begin{eqnarray}
\hspace{-.45in}\xi_{prim+}(\sigma,\g,\h,\xtilde,r_{socp+})=\max_{\nu,\lambda_{c_{1}^++1:n}} & & \sigma\sqrt{\|\g\|_2^2\nu^2-\|\nu\bar{\h}^+_{c_{1}^++1:n}-\z_{c_{1}^++1:n}+\lambda_{c_{1}^++1:n}\|_2^2} -x_{mag}\sum_{i=n-k+1}^{n}\lambda_i-\nu r_{socp+}\nonumber \\
\mbox{subject to}
& & \nu\geq 0\nonumber \\
& & \lambda_i\geq 0, c_{2}^++1\leq i\leq n\nonumber \\
& & \lambda_i=0, c_{1}^++1\leq i\leq c_{2}^+.\label{eq:mainlasso5non}
\end{eqnarray}
To make writing of what will follow somewhat easier we set
\begin{equation}
\xi^{(obj+)}=\sigma\sqrt{\|\g\|_2^2\nu^2-\|\nu\bar{\h}^+_{c_{1}^++1:n}-\z_{c_{1}^++1:n}+\lambda_{c_{1}^++1:n}\|_2^2} -x_{mag}\sum_{i=n-k+1}^{n}\lambda_i-\nu r_{socp+}.
\label{eq:defxiobjnon}
\end{equation}

Similarly to what was done in Section \ref{sec:unsigned} we then proceed by solving the optimization in (\ref{eq:mainlasso5non}) over $\nu$ and $\lambda_{c_{1}^++1:n}$. To do so we first look at the derivatives with respect to $\lambda_i,c_{2}^++1\leq i\leq n$, of the objective in (\ref{eq:mainlasso5non}). Computing the derivatives and equalling them to zero gives
\begin{eqnarray}
\hspace{-.3in}& & \frac{d \xi^{(obj+)}}{d \lambda_i}  =  0, c_{2}^++1\leq i\leq n\nonumber \\
&\iff &  \sigma \frac{-(\nu\bar{\h}^+_i-\z_i+\lambda_i)}
{\sqrt{\|\g\|_2^2\nu^2-\|\nu\bar{\h}^+_{c_{1}^++1:n}-\z_{c_{1}^++1:n}+\lambda_{c_{1}^++1:n}\|_2^2}}-x_{mag} =  0, c_{2}^++1\leq i\leq n\nonumber \\
& \iff &  \lambda_i-\z_i+
\nu\bar{\h}^+_{i}=-\frac{x_{mag}}{\sigma}\sqrt{\|\g\|_2^2\nu^2-\|\nu\bar{\h}^+_{c_{1}^++1:n}-\z_{c_{1}^++1:n}+\lambda_{c_{1}^++1:n}\|_2^2}, c_{2}^++1\leq i\leq n\nonumber \\
& \iff &  \lambda_i=-\frac{x_{mag}}{\sigma}\sqrt{\|\g\|_2^2\nu^2-\|\nu\bar{\h}^+_{c_{1}^++1:n}-\z_{c_{1}^++1:n}+\lambda_{c_{1}^++1:n}\|_2^2}+\z_i-
\nu\bar{\h}^+_{i}, c_{2}^++1\leq i\leq n.\nonumber \\\label{eq:mainlasso6non}
\end{eqnarray}
From the second to last line in the above equation one then has
\begin{equation*}
(\lambda_i-\z_i+\nu\bar{\h}^+_{i})^2=\frac{x_{mag}^2}{\sigma^2}(\|\g\|_2^2\nu^2-\|\nu\bar{\h}^+_{c_{1}^++1:c_{2}^+}-\z_{c_{1}^++1:c_{2}^+}
+\lambda_{c_{1}^++1:c_{2}^+}\|_2^2-\|\nu\bar{\h}_{c_{2}^++1:n}-\z_{c_{2}^++1:n}+\lambda_{c_{2}^++1:n}
\|_2^2)
\end{equation*}
and after an easy algebraic transformation
\begin{equation}
\hspace{-.3in}(\lambda_i-\z_i+\nu\bar{\h}^+_{i})^2=\frac{x_{mag}^2}{\sigma^2+(n-c_{2}^+)x_{mag}^2}(\|\g\|_2^2\nu^2-\|\nu\bar{\h}^+_{c_{1}^++1:c_{2}^+}-\z_{c_{1}^++1:c_{2}^+}
\|_2^2).\label{eq:mainlasso7non}
\end{equation}
Using (\ref{eq:mainlasso7non}) we further have
\begin{equation}
\hspace{-.3in}\sqrt{\|\g\|_2^2\nu^2-\|\nu\bar{\h}^+_{c_{1}^++1:n}-\z_{c_{1}^++1:n}
+\lambda_{c_{1}^++1:n}\|_2^2}
=\frac{\sigma}{\sqrt{\sigma^2+(n-c_{2}^+)x_{mag}^2}}\sqrt{\|\g\|_2^2\nu^2-\|\nu\bar{\h}^+_{c_{1}^++1:c_{2}^+}-\z_{c_{1}^++1:c_{2}^+}\|_2^2}.
\label{eq:mainlasso8non}
\end{equation}
Plugging the value for $\lambda_i$ from (\ref{eq:mainlasso5non}) in (\ref{eq:defxiobjnon}) gives
\begin{eqnarray}
\xi^{(obj+)} & = & \sigma\sqrt{\|\g\|_2^2\nu^2-\|\nu\bar{\h}^+_{c_{1}^++1:n}-\z_{c_{1}^++1:n}+\lambda_{c_{1}^++1:n}\|_2^2} -x_{mag}\sum_{i=n-k+1}^{n}\lambda_i-\nu r_{socp+}\nonumber \\
& = & \frac{\sigma^2+(n-c_{2}^+)x_{mag}^2}{\sigma}\sqrt{\|\g\|_2^2\nu^2-\|\nu\bar{\h}^+_{c_{1}^++1:n}-\z_{c_{1}^++1:n}+\lambda_{c_{1}^++1:n}\|_2^2}\nonumber \\
& - & x_{mag}(n-c_{2}^+)+\nu x_{mag}\sum_{i=c_2^++1}^{n}\bar{\h}^+_i-\nu r_{socp+}.
\label{eq:mainlasso9non}
\end{eqnarray}
Combining (\ref{eq:mainlasso8non}) and (\ref{eq:mainlasso9non}) we finally obtain the following ``signed" analogue to (\ref{eq:mainlasso10})
\begin{equation}
\hspace{-.3in}\xi^{(obj+)}
 =  \sqrt{\sigma^2+(n-c_{2}^+)x_{mag}^2}
\sqrt{\|\g\|_2^2\nu^2-\|\nu\bar{\h}^+_{c_{1}^++1:c_{2}^+}-\z_{c_{1}^++1:c_{2}^+}\|_2^2}
-  \nu(r_{socp+}- x_{mag}\sum_{i=c_2^++1}^{n}\bar{\h}^+_i)-x_{mag}(n-c_{2}^+).
\label{eq:mainlasso10non}
\end{equation}
Equalling the derivative of $\xi^{(obj+)}$ with respect to $\nu$ to zero further gives
\begin{eqnarray}
\hspace{-.3in}& & \frac{d \xi^{(obj+)}}{d \nu}  =  0\nonumber \\
&\iff &
\frac{\nu(\|\g\|_2^2-\sum_{i=c_1^++1}^{c_2^+}(\bar{\h}^+_i)^2)+(\bar{\h}^+_{c_{1}^++1:c_{2}^+})^T\z_{c_{1}^++1:c_{2}^+}}
{(\sqrt{\sigma^2+(n-c_{2}^+)x_{mag}^2})^{-1}\sqrt{\|\g\|_2^2\nu^2-\|\nu\bar{\h}^+_{c_{1}^++1:c_{2}^+}-\z_{c_{1}^++1:c_{2}^+}\|_2^2}}-(r_{socp+}- x_{mag}\sum_{i=c_2^++1}^{n}\bar{\h}^+_i) =0.\nonumber \\\label{eq:mainlasso11non}
\end{eqnarray}
Let
\begin{eqnarray}
 s_{dep+} & = & (\bar{\h}^+_{c_{1}^++1:c_{2}^+})^T\z_{c_{1}^++1:c_{2}^+}\nonumber \\
d_{dep+} & = & \sum_{i=c_1+1}^{c_2}(\bar{\h}^+_i)^2 \nonumber \\
r_{dep+} & = & r_{socp+}- x_{mag}\sum_{i=c_2+1}^{n}\bar{\h}^+_i\nonumber \\
a_{dep+} & = & \frac{\|\g\|_2^2-(\sum_{i=c_1+1}^{c_2}(\bar{\h}^+_i)^2)}
{\sqrt{\sigma^2+(n-c_{2}^+)x_{mag}^2}^{-1}r_{dep+}}=
\frac{\sqrt{\sigma^2+(n-c_{2}^+)x_{mag}^2}(\|\g\|_2^2-d_{dep+})}{r_{dep+}}\nonumber \\
b_{dep+} & = & \frac{(\bar{\h}^+_{c_{1}^++1:c_{2}^+})^T\z_{c_{1}^++1:c_{2}^+}}
{\sqrt{\sigma^2+(n-c_{2}^+)x_{mag}^2}^{-1}r_{dep+}}=\frac{\sqrt{\sigma^2+(n-c_{2}^+)x_{mag}^2}s_{dep+}}{r_{dep+}}.
\label{eq:defdepnon}
\end{eqnarray}
Then combining (\ref{eq:mainlasso11non}) and (\ref{eq:defdepnon}) one obtains
\begin{equation}
(a_{dep+}\nu+b_{dep+})^2=\|\g\|_2^2\nu^2-\|\nu\bar{\h}^+_{c_{1}^++1:c_{2}^+}-\z_{c_{1}^++1:c_{2}^+}\|_2^2.\label{eq:mainlasso13non}
\end{equation}
After solving (\ref{eq:mainlasso13non}) over $\nu$ we have
\begin{equation}
\hspace{-.0in}\nu=\frac{-(a_{dep+}b_{dep+}-s_{dep+})-\sqrt{(a_{dep+}b_{dep+}-s_{dep+})^2-
(b_{dep+}^2+\|\z_{c_{1}^++1:c_{2}^+}\|_2^2)(a_{dep+}^2-\|\g\|_2^2+d_{dep+})}}
{a_{dep+}^2-\|\g\|_2^2+d_{dep+}}.\label{eq:mainlasso14non}
\end{equation}
Following what was done in Section \ref{sec:unsigned} and earlier in \cite{StojnicCSetam09,StojnicGenSocp10}, we have that a combination of (\ref{eq:mainlasso6non}) and (\ref{eq:mainlasso14non}) gives the following two equations that can be used to determine $c_1$ and $c_2$.
\begin{eqnarray}
\nu\bar{\h}^+_{c_{2}^+}-\z_{c_{2}^+}+\frac{x_{mag}}{\sigma}\sqrt{\|\g\|_2^2\nu^2-\|\nu\bar{\h}^+_{c_{1}^++1:n}-\z_{c_{1}^++1:n}+\lambda_{c_{1}^++1:n}\|_2^2}& =& 0\nonumber \\
\hspace{-.5in}\bar{\h}^+_{c_1^+}\frac{-(a_{dep+}b_{dep+}-s_{dep+})-\sqrt{(a_{dep+}b_{dep+}-s_{dep+})^2-
(b_{dep+}^2+\|\z_{c_{1}^++1:c_{2}^+}\|_2^2)(a_{dep+}^2-\|\g\|_2^2+d_{dep+})}}
{a_{dep+}^2-\|\g\|_2^2+d_{dep+}} & = & 1.\nonumber \\
\label{eq:mainlasso15non}
\end{eqnarray}
The last term that appears on the right hand side of the first of the above equations can be further simplified based on
(\ref{eq:mainlasso8non}) in the following way
\begin{multline}
\hspace{-.0in}\sqrt{\|\g\|_2^2\nu^2-\|\nu\bar{\h}^+_{c_{1}^++1:n}-\z_{c_{1}^++1:n}+\lambda_{c_{1}^++1:n}\|_2^2}=
\frac{\sigma\sqrt{\|\g\|_2^2\nu^2-\nu^2d_{dep+}+2\nu s_{dep+}-(c_2^+-c_1^+)}}{\sqrt{\sigma^2+(n-c_2^+)x_{mag}^2}},\label{eq:mainlasso16non}
\end{multline}
where we of course recognized that $\|\z_{c_{1}^++1:c_{2}^+}\|_2^2=c_2^+-c_1^+$. Combining (\ref{eq:defdepnon}) and (\ref{eq:mainlasso16non}) one can then simplify the equations from (\ref{eq:mainlasso15non}) in the following way
\begin{eqnarray}
\hspace{-.5in}\nu\bar{\h}^+_{c_{2}^+}-\z_{c_{2}^+}+\frac{x_{mag}}{\sqrt{\sigma^2+(n-c_2^+)x_{mag}^2}}\sqrt{\|\g\|_2^2\nu^2-\nu^2d_{dep+}+2\nu s_{dep+}-(c_2^+-c_1^+)} & = & 0\nonumber \\
\bar{\h}^+_{c_1^+}\frac{-(a_{dep+}b_{dep+}-s_{dep+})-\sqrt{(a_{dep+}b_{dep+}-s_{dep+})^2-
(b_{dep+}^2+(c_2^+-c_1^+))(a_{dep+}^2-\|\g\|_2^2+d_{dep+})}}
{a_{dep+}^2-\|\g\|_2^2+d_{dep+}} & = & 1.\nonumber \\
\label{eq:c1c2c3non}
\end{eqnarray}
Let $\widehat{c_1^+}$ and $\widehat{c_2^+}$ be the solution of (\ref{eq:c1c2c3non}). Then
\begin{equation}
\nu_{dep+} = \frac{-(\widehat{a_{dep+}}\widehat{b_{dep+}}-\widehat{s_{dep+}})-\sqrt{(\widehat{a_{dep+}}\widehat{b_{dep+}}-\widehat{s_{dep+}})^2-
(\widehat{b_{dep+}}^2+(\widehat{c_2}-\widehat{c_1}))(\widehat{a_{dep+}}^2-\|\g\|_2^2+\widehat{d_{dep+}})}}
{\widehat{a_{dep+}}^2-\|\g\|_2^2+\widehat{d_{dep+}}},\label{eq:mainhatnunon}
\end{equation}
where $\widehat{s_{dep+}}$, $\widehat{d_{dep+}}$, $\widehat{a_{dep+}}$, and $\widehat{b_{dep+}}$ are $s_{dep+}$, $d_{dep+}$, $a_{dep+}$, and $b_{dep+}$ from (\ref{eq:defdepnon}) computed with $\widehat{c_1}$ and $\widehat{c_2}$. From (\ref{eq:mainlasso4non}) one then has
\begin{equation}
\|\w_{dep+}\|_2=\sigma\frac{\|\nu_{dep+}\bar{\h}^+_{\widehat{c_{1}^+}+1:n}-\z_{\widehat{c_{1}^+}+1:n}+\lambda_{\widehat{c_{1}^+}+1:n}^{(dep+)}\|_2}
{\sqrt{\|\g\|_2^2\nu_{dep+}^2-\|\nu_{dep+}\bar{\h}^+_{\widehat{c_{1}^+}+1:n}-\z_{\widehat{c_{1}^+}+1:n}+\lambda_{\widehat{c_{1}^+}+1:n}^{(dep+)}\|_2^2}}.
\label{eq:mainhatwnon}
\end{equation}
Combining (\ref{eq:mainlasso16non}) and (\ref{eq:mainhatwnon}) one further has
\begin{equation}
\|\w_{dep+}\|_2=\sigma\frac{\sqrt{\|\g\|_2^2\nu_{dep+}^2(n-\widehat{c_2^+})\frac{x_{mag}^2}{\sigma^2}+\nu_{dep+}^2\widehat{d_{dep+}}-2\nu \widehat{s_{dep+}}+(\widehat{c_2^+}-\widehat{c_1^+})}}
{\sqrt{\|\g\|_2^2\nu_{dep+}^2-\nu_{dep+}^2\widehat{d_{dep+}}+2\nu_{dep+} \widehat{s_{dep+}}-(\widehat{c_2^+}-\widehat{c_1^+})}}.\label{eq:mainhatw1non}
\end{equation}
Combination of (\ref{eq:mainhatnunon}) and (\ref{eq:mainhatw1non}) is conceptually enough to determine $\|\w_{dep+}\|_2$ (and then afterwards easily $E\xi_{prim+}^{(dep)}(\sigma,\g,\h,x_{mag},r_{socp+})$). The part that remains though is a computation of all unknown quantities that appear in (\ref{eq:mainhatnunon}) and (\ref{eq:mainhatw1non}). We will below show how that can be done. In doing so we as usual substantially rely on what was shown in \cite{StojnicCSetam09,StojnicGenSocp10} and assume a familiarity with the procedures presented there.

The first thing to resolve is (\ref{eq:c1c2c3non}). Since all random quantities concentrate we will be dealing (as in \cite{StojnicCSetam09,StojnicGenSocp10}) with the expected values. To compute the solution of (\ref{eq:c1c2c3non}), $\widehat{c_1}$ and $\widehat{c_2}$,  we will need the following expected values
\begin{equation}
E\|\g\|_2^2, E\|\bar{\h}^+_{c_{1}^++1:n-k}\|_2^2, E\|\bar{\h}^+_{n-k+1:c_2^+}\|_2^2,
E ((\bar{\h}^+_{c_{1}^++1:n-k})^T\z_{c_{1}^++1:n-k}), E ((\bar{\h}^+_{n-k+1:c_2^+})^T\z_{n-k+1:c_2^+}).\label{eq:mainexp1non}
\end{equation}
As in Section \ref{sec:unsigned} we easily have
\begin{equation}
E\|\g\|_2^2=m.\label{eq:mainexpgnon}
\end{equation}
Let $c_{1}^+=(1-\theta_1^+)n$ and $c_{2}^+=\theta_2^+ n$ where $\theta_{1}^+$ and $\theta_{2}^+$ are constants independent of $n$. Then as shown in \cite{StojnicCSetam09,StojnicGenSocp10}
\begin{equation}
\lim_{n\rightarrow\infty}\frac{E\|\bar{\h}^+_{c_{1}^++1:n-k}\|_2^2}{n}  =  \frac{1-\beta_{w}^+}{\sqrt{2\pi}}\left (\frac{\sqrt{2}(\erfinv(2\frac{1-\theta_{1}^+}{1-\beta_{w}^+}-1))}{e^{(\erfinv(2\frac{1-\theta_{1}^+}{1-\beta_{w}^+}-1))^2}}\right )+\theta_{1}^+-\beta_{w}^+.\label{eq:mainexpnormh1non}
\end{equation}
where we of course recall that $\beta_{w}^+=\frac{k}{n}$. Also, as in Section \ref{sec:unsigned} and earlier in \cite{StojnicCSetam09,StojnicGenSocp10}
we have
\begin{equation}
\lim_{n\rightarrow\infty}\frac{E\|\bar{\h}^+_{n-k+1:c_{2}^+}\|_2^2}{n}  =  \frac{\beta_{w}^+}{\sqrt{2\pi}}\left (\frac{\sqrt{2}\erfinv(2\frac{1-\theta_{2}^+}{\beta_{w}^+}-1)}{e^{(\erfinv(2\frac{1-\theta_{2}^+}{\beta_{w}^+}-1))^2}}\right )+\theta_{2}^+-1+\beta_{w}^+.\label{eq:mainexpnormh2non}
\end{equation}
Following further what was established in \cite{StojnicCSetam09,StojnicGenSocp10} we have
\begin{eqnarray}
\lim_{n\rightarrow\infty}\frac{E((\bar{\h}^+_{c_{1}^++1:n-k})^T\z_{c_{1}^++1:n-k})}{n} & = &
\left ((1-\beta_{w}^+)\sqrt{\frac{1}{2\pi}}e^{-(\erfinv(2\frac{1-\theta_{1}^+}{1-\beta_{w}^+}-1))^2}\right )\nonumber \\
\lim_{n\rightarrow\infty}\frac{E((\bar{\h}^+_{n-k+1:c_{2}^+})^T\z_{n-k+1:c_{2}^+})}{n} & = &
\left (\beta_{w}^+\sqrt{\frac{1}{2\pi}}e^{-(\erfinv(2\frac{1-\theta_{2}^+}{\beta_{w}^+}-1))^2}\right ).
\label{eq:mainexpprodnon}
\end{eqnarray}
From (\ref{eq:mainexpprodnon}) we also have
\begin{equation}
\lim_{n\rightarrow\infty}\frac{E(\sum_{i=c_2+1}^{n}\bar{\h}^+_{i})}{n} =
-\left (\beta_{w}^+\sqrt{\frac{1}{2\pi}}e^{-(\erfinv(2\frac{1-\theta_{2}^+}{\beta_{w}^+}-1))^2}\right ).\label{eq:mainexpprod1non}
\end{equation}
The only other thing that we will need in order to be able to compute $\widehat{c_{1}^+}$ and $\widehat{c_{2}^+}$ (besides the expectations from (\ref{eq:mainexp1non})) are the following inequalities related to the behavior of $\bar{\h}^+_{c_{1}^+}$ and $\bar{\h}^+_{c_{2}^+}$. Again, as shown in \cite{StojnicCSetam09,StojnicGenSocp10}
\begin{eqnarray}
P(\sqrt{2}\erfinv ((1+\epsilon_1^{\bar{\h}^+_{c_{1}^+}})(2\frac{1-\theta_{1}^+}{1-\beta_{w}^+}-1))\leq \bar{\h}^+_{c_{1}^+}) & \leq & e^{-\epsilon_2^{\bar{\h}^+_{c_{1}^+}} n}\nonumber \\
P(\sqrt{2}\erfinv ((1+\epsilon_1^{\bar{\h}^+_{c_{2}^+}})(2\frac{1-\theta_{2}^+}{\beta_{w}^+}-1))\leq \bar{\h}^+_{c_{2}^+}) & \leq & e^{-\epsilon_2^{\bar{\h}^+_{c_{2}^+}} n}.\label{eq:mainbrpointnon}
\end{eqnarray}
where $\epsilon_1^{\bar{\h}^+_{c_{1}^+}}>0$ and $\epsilon_1^{\bar{\h}^+_{c_{2}^+}}>0$ are arbitrarily small constants and $\epsilon_2^{\bar{\h}^+_{c_{1}^+}}$ and $\epsilon_2^{\bar{\h}^+_{c_{2}^+}}$ are constants dependent on $\epsilon_1^{\bar{\h}^+_{c_{1}^+}}$ and $\epsilon_1^{\bar{\h}^+_{c_{2}^+}}$, respectively, but independent of $n$.

At this point we have all the necessary ingredients to determine $\widehat{c_{1}^+}$ and $\widehat{c_{2}^+}$ and consequently $\nu_{dep+}$, $\|\w_{dep+}\|_2$, and $\xi_{prim+}^{(dep)}(\sigma,\g,\h,x_{mag},r_{socp+})$, or to be more precise their concentrating points. The following theorem then provides a systematic way of doing so.
\begin{theorem}
Assume the setup of Theorem \ref{thm:mainlassonon}. Let the nonzero components of $\xtilde$ have magnitude $x_{mag}$ and let $\bar{\h}^+$ be as defined in (\ref{eq:defhbarnon}). Further, let $r_{socp+}^{(sc)}=\lim_{n\rightarrow\infty}\frac{r_{socp+}}{\sqrt{n}}$ and $x_{mag}^{(sc)}=\lim_{n\rightarrow\infty}\frac{x_{mag}}{\sqrt{n}}$. Also, let $\nu_{dep+}$, $\|\w_{dep+}\|_2$, and $\xi_{prim+}^{(dep)}(\sigma,\g,\h,x_{mag},r_{socp+})$ be as defined in and right after (\ref{eq:mainlasso3vernon}). Let $\alpha=\frac{m}{n}$ and $\beta_{w}^+=\frac{k}{n}$ be fixed. Consider the following
\begin{equation*}
S(\theta_{1}^+,\theta_{2}^+)  =  \lim_{n\rightarrow\infty}\frac{E s_{dep+}}{n}  =  \left ((1-\beta_{w}^+)\sqrt{\frac{1}{2\pi}}e^{-(\erfinv(2\frac{1-\theta_{1}^+}{1-\beta_{w}^+}-1))^2}\right )+\left (\beta_{w}^+\sqrt{\frac{1}{2\pi}}e^{-(\erfinv(2\frac{1-\theta_{2}^+}{\beta_{w}^+}-1))^2}\right )
\end{equation*}
\begin{multline*}
D(\theta_{1}^+,\theta_{2}^+)  =  \lim_{n\rightarrow\infty}\frac{E d_{dep+}}{n}  =  \frac{1-\beta_{w}^+}{\sqrt{2\pi}}\left (\frac{\sqrt{2}(\erfinv(2\frac{1-\theta_{1}^+}{1-\beta_{w}^+}-1))}{e^{(\erfinv(2\frac{1-\theta_{1}^+}{1-\beta_{w}^+}-1))^2}}\right )+\theta_{1}^+-\beta_{w}^+\\+\frac{\beta_{w}^+}{\sqrt{2\pi}}\left (\frac{\sqrt{2}\erfinv(2\frac{1-\theta_{2}^+}{\beta_{w}^+}-1)}{e^{(\erfinv(2\frac{1-\theta_{2}^+}{\beta_{w}^+}-1))^2}}\right )+\theta_{2}^+-1+\beta_{w}^+
\end{multline*}
\begin{equation*}
R(\theta_{2}^+)  =  \lim_{n\rightarrow\infty}\frac{E r_{dep+}}{\sqrt{n}}  =  r_{socp+}^{(sc)}+x_{mag}^{(sc)}\left (\beta_{w}^+\sqrt{\frac{1}{2\pi}}e^{-(\erfinv(2\frac{1-\theta_{2}^+}{\beta_{w}^+}-1))^2}\right )
\end{equation*}
\begin{eqnarray}
A(\theta_{1}^+,\theta_{2}^+) & = & \lim_{n\rightarrow\infty}\frac{E a_{dep+}}{\sqrt{n}}=\frac{\sqrt{\sigma^2+(1-\theta_{2}^+)(x_{mag}^{(sc)})^2}(\alpha-D(\theta_{1}^+,\theta_{2}^+))}{R(\theta_{2}^+)}\nonumber \\
B(\theta_{1}^+,\theta_{2}^+) & = & \lim_{n\rightarrow\infty}\frac{E b_{dep+}}{\sqrt{n}}=\frac{\sqrt{\sigma^2+(1-\theta_{2}^+)(x_{mag}^{(sc)})^2}S(\theta_{1}^+,\theta_{2}^+)}{R(\theta_{2}^+)}\nonumber \\
F(\theta_{1}^+) & = & \sqrt{2}\erfinv (2\frac{1-\theta_{1}^+}{1-\beta_{w}^+}-1)\nonumber \\
G(\theta_{2}^+) & = & \sqrt{2}\erfinv (2\frac{1-\theta_{2}^+}{\beta_{w}^+}-1).\label{eq:maincompthmcond1non}
\end{eqnarray}
Set
\begin{multline}
N(\theta_{1}^+,\theta_{2}^+)=\frac{-(A(\theta_{1}^+,\theta_{2}^+)B(\theta_{1}^+,\theta_{2}^+)-S(\theta_{1}^+,\theta_{2}^+))}
{A(\theta_{1}^+,\theta_{2}^+)^2-\alpha+D(\theta_{1}^+,\theta_{2}^+)}\nonumber \\
\hspace{-.6in}-\frac{\sqrt{(A(\theta_{1}^+,\theta_{2}^+)B(\theta_{1}^+,\theta_{2}^+)-S(\theta_{1}^+,\theta_{2}^+))^2-
(B(\theta_{1}^+,\theta_{2}^+)^2+\theta_{1}^++\theta_{2}^+-1)(A(\theta_{1}^+,\theta_{2}^+)^2-\alpha+D(\theta_{1}^+,\theta_{2}^+))}}
{A(\theta_{1}^+,\theta_{2}^+)^2-\alpha+D(\theta_{1}^+,\theta_{2}^+)}.
\end{multline}
Let the pair ($\widehat{\theta_{1}^+}$, $\widehat{\theta_{2}^+}$) be the solution of the following two equations
\begin{eqnarray}
\hspace{-.6in}N(\theta_{1}^+,\theta_{2}^+)G(\theta_{2}^+)+\frac{x_{mag}^{(sc)}\sqrt{N(\theta_{1}^+,\theta_{2}^+)^2(\alpha
-D(\theta_{1}^+,\theta_{2}^+))+2N(\theta_{1}^+,\theta_{2}^+) S(\theta_{1}^+,\theta_{2}^+)-(\theta_{1}^++\theta_{2}^+-1)}}
{\sqrt{\sigma^2+(1-\theta_{2}^+)(x_{mag}^{(sc)})^2}} & = & 1\nonumber \\
F(\theta_{1}^+)N(\theta_{1}^+,\theta_{2}^+) & = & 1.\nonumber \\\label{eq:mainthmc1c2c3non}
\end{eqnarray}
Then the concentrating points of $\nu_{dep+}$, $\|\w_{dep+}\|_2$, and $\xi_{prim+}^{(dep)}(\sigma,\g,\h,x_{mag},r_{socp+})$ can be determined as
\begin{equation*}
E\nu_{dep+}  =  N(\widehat{\theta_{1}^+},\widehat{\theta_{2}^+})
\end{equation*}
\begin{equation*}
\hspace{-.8in}E\|\w_{dep+}\|_2  =  \sigma\frac{\sqrt{N(\widehat{\theta_{1}^+},\widehat{\theta_{2}^+})^2(\alpha(1-\widehat{\theta_{2}^+})\frac{(x_{mag}^{(sc)})^2}{\sigma^2}
+
D(\widehat{\theta_{1}^+},\widehat{\theta_{2}^+}))-2N(\widehat{\theta_{1}^+},\widehat{\theta_{2}^+}) S(\widehat{\theta_{1}^+},\widehat{\theta_{2}^+})+(\widehat{\theta_{1}^+}+\widehat{\theta_{2}^+}-1)}}
{\sqrt{N(\widehat{\theta_{1}^+},\widehat{\theta_{2}^+})^2(\alpha
-
D(\widehat{\theta_{1}^+},\widehat{\theta_{2}^+}))+2N(\widehat{\theta_{1}^+},\widehat{\theta_{2}^+}) S(\widehat{\theta_{1}^+},\widehat{\theta_{2}^+})-(\widehat{\theta_{1}^+}+\widehat{\theta_{2}^+}-1)}}
\end{equation*}
\begin{multline}
\hspace{-.9in}\lim_{n\rightarrow\infty}\frac{E\xi_{prim+}^{(dep)}(\sigma,\g,\h,x_{mag},r_{socp+})}{\sqrt{n}}  =  \sigma
\frac{\sqrt{N(\widehat{\theta_{1}^+},\widehat{\theta_{2}^+})^2(\alpha
-
D(\widehat{\theta_{1}^+},\widehat{\theta_{2}^+}))+2N(\widehat{\theta_{1}^+},\widehat{\theta_{2}^+}) S(\widehat{\theta_{1}^+},\widehat{\theta_{2}^+})-(\widehat{\theta_{1}^+}+\widehat{\theta_{2}^+}-1)}}
{\sqrt{1+(1-\widehat{\theta_{2}^+})\frac{(x_{mag}^{(sc)})^2}{\sigma^2}}}\\
-N(\widehat{\theta_{1}^+},\widehat{\theta_{2}^+})r_{socp+}^{(sc)}.\label{eq:mainthmnuwgenxiprimnon}
\end{multline}
\label{thm:maincomperrornon}
\end{theorem}
\begin{proof}
Follows from Theorem \ref{thm:mainlassonon} based on the discussion presented above and a combination of (\ref{eq:defdepnon}), (\ref{eq:c1c2c3non}), (\ref{eq:mainhatnunon}), and (\ref{eq:mainhatw1non}).
\end{proof}

The results from the above theorem can be used to compute parameters of interest in our derivation for particular values of $\beta_{w}^+$, $\alpha$, $\sigma$, $x_{mag}$, and $r_{socp+}$. In the following subsections we will present a collection of such results. However, before doing so in the next subsection we take a look back and discuss the feasibility of (\ref{eq:socpnon}) in a bit more detail.

\subsubsection{Feasibility of (\ref{eq:socpnon})} \label{sec:feasibilty}

As we have mentioned at the beginning of this section the optimization problem in (\ref{eq:socpnon}) is not necessarily feasible for all possible choices of $A,\v,\sigma,\xtilde$, and $r_{socp+}$. Analogously, (\ref{eq:mainlasso1non}) and (\ref{eq:mainlasso3vernon}) are not necessarily bounded for all choices of $\alpha,\beta_w^+,\sigma,\g,\h,\xtilde$, and $r_{socp+}$. Below we provide a brief discussion on potential unboundedness of (\ref{eq:mainlasso1non}) and (\ref{eq:mainlasso3vernon}). We will split the discussion into two parts. First we focus on a couple of scenarios where the objective of (\ref{eq:mainlasso1non}) and (\ref{eq:mainlasso3vernon}) is bounded. Afterwards we present a procedure that can be used to determine $x_{mag}^{(sc)}$ for which (\ref{eq:mainlasso3vernon}) becomes unbounded and (\ref{eq:socpnon}) infeasible.

\textbf{\underline{\emph{1) Universally feasible scenarios}}}

Clearly, the most critical case for (\ref{eq:mainlasso3vernon}) to be unbounded is that $x_{mag}^{(sc)}=0$. Even if $x_{mag}^{(sc)}=0$ one can distinguish two important scenarios for parameters $\alpha,\beta_w^+,\sigma$, and $r_{socp+}$ such that the objective in (\ref{eq:mainlasso3vernon}) is bounded with overwhelming probability.

\underline{\emph{a) $r_{socp+}>\sigma \sqrt{m}$}}

The first scenario assumes $r_{socp+}>\sigma \sqrt{m}$. Then for any combination of $(\alpha,\beta_w^+)$ that lies on or below the signed fundamental characterization (\ref{eq:fundl1non}) and any $\sigma$ one has that
\begin{equation}\sigma\sqrt{\|\g\|_2^2\nu^2-\|\nu\bar{\h}^+-\z+\lambda\|_2^2}-\nu r_{socp+}=\nu(\sigma\sqrt{\|\g\|_2^2-\|\bar{\h}^+-\frac{\z}{\nu}+\frac{\lambda}{\nu}\|_2^2}-r_{socp+}).\label{eq:feas1}
\end{equation}
Let $\epsilon^{(feas)}_1>0$ be an arbitrarily small constant and let $\epsilon^{(feas)}_2$ be a constant dependent on $\epsilon^{(feas)}_1$ but independent of $n$. Since
$$P(\|\g\|_2\leq (1+\epsilon^{(feas)}_1)\sqrt{m})\geq 1-e^{-\epsilon^{(feas)}_2 n}$$
one has that with overwhelming probability
$$\sigma\sqrt{\|\g\|_2^2\nu^2-\|\nu\bar{\h}^+-\z+\lambda\|_2^2}-\nu r_{socp+}\leq 0$$ which implies that the objective in (\ref{eq:mainlasso3vernon}) is indeed bounded with overwhelming probability.

\underline{\emph{b) $\alpha\leq 0.5$}}

The second scenario that we consider assumes that $\alpha\leq 0.5$. Then for any combination of $(\alpha,\beta_w^+)$ that lies on or below the signed fundamental characterization (\ref{eq:fundl1non}) and any $\sigma$ one again has that the objective in (\ref{eq:mainlasso3vernon}) is bounded with overwhelming probability. To show this we will look at the expression on the right hand side of (\ref{eq:feas1}) instead of looking at the objective of (\ref{eq:mainlasso3vernon}). Clearly, to have that expression unbounded one must have $\nu\rightarrow\infty$ and the term in the parenthesis must be positive. Also the term under the square root must be nonnegative. Since $\z_i=1$ for any $i$ one easily has that $\z/\nu$ would have to converge to $0$. Let us therefore look at the following function $\zeta(\lambda)$
\begin{equation}
\zeta(\lambda)=\|\g\|_2^2-\|\bar{\h}^++\frac{\lambda}{\nu}\|_2^2.\label{eq:feasdefzeta}
\end{equation}
Furthermore let
\begin{equation}
\widehat{\zeta}=\max_{\lambda}\zeta(\lambda)=\max_{\lambda}\|\g\|_2^2-\|\bar{\h}^++\frac{\lambda}{\nu}\|_2^2.\label{eq:feas2}
\end{equation}
Then it is not that hard to see that
\begin{equation}
\widehat{\zeta}=\|\g\|_2^2-\|\bar{\h}^+_{c_1^{f+}+1:c_2^{f+}}\|_2^2,\label{eq:feas2}
\end{equation}
for certain $c_1^{f+}+1\leq n-k$ and $n-k+1\geq c_2^{f+}\leq n$. For any arbitrarily small constants $\epsilon^{(feas)}_3>0$ and $\epsilon^{(feas)}_5>0$ and constants $\epsilon^{(feas)}_4$, $\epsilon^{(feas)}_6$ dependent on $\epsilon^{(feas)}_3$ and $\epsilon^{(feas)}_5$, respectively but independent of $n$ one easily has $$P(\bar{\h}^+_{((1+\epsilon^{(feas)}_3)\frac{n-k}{2})}\geq 0)\geq 1-e^{-\epsilon^{(feas)}_4 n}$$ and
$$P(\bar{\h}^+_{((1-\epsilon^{(feas)}_5)\frac{2n-k}{2})}\geq 0)\geq 1-e^{-\epsilon^{(feas)}_5 n}.$$ One then with overwhelming probability has that $c_1^{f+}<(1+\epsilon^{(feas)}_3)\frac{n-k}{2}$ and $c_2^{f+}>(1-\epsilon^{(feas)}_5)\frac{2n-k}{2}$. From (\ref{eq:feas2}) it then easily follows that with overwhelming probability
\begin{equation}
\widehat{\zeta}=\|\g\|_2^2-\|\bar{\h}^+_{c_1^{f+}+1:c_2^{f+}}\|_2^2<m-\frac{n}{2}.\label{eq:feas3}
\end{equation}
On the other hand if $\alpha< 0.5$ then $m<\frac{n}{2}$ and from (\ref{eq:feasdefzeta}), (\ref{eq:feas2}), and (\ref{eq:feas3}) one has that with overwhelming probability $\zeta(\lambda)$ can not be positive for $\nu\rightarrow\infty$. This in return implies that with overwhelming probability the objective in (\ref{eq:mainlasso3vernon}) is not unbounded.

We also mention that the feasibility of (\ref{eq:socpnon}) in the above mentioned scenarios could be deduced by looking directly at (\ref{eq:socpnon}). For example, if $r_{socp+}>\sigma m$ then $\x=\xtilde$ is with overwhelming probability feasible in (\ref{eq:socpnon}). On the other hand if $\alpha<0.5$ then based on results of \cite{StojnicISIT2010binary} (and earlier \cite{DTbern}) one has that the norm-2 in the constraint of (\ref{eq:socpnon}) can with overwhelming probability be made zero (this is in fact exactly the inverse problem of the one considered in \cite{StojnicISIT2010binary,DTbern}). We will present this small exercise in one of our forthcoming papers since it fits better the topic there.

Also, we should mention that one can define many other scenarios where (\ref{eq:mainlasso3vernon}) is bounded. However, we restricted only to the above two since they are relatively simple to describe and have a nice connection to already known results.

\textbf{\underline{\emph{2) Finding the feasibility breaking point}}}

In the rest of this subsection we will present a general mechanism that can be used to determine a critical $x_{mag}^{(sc)}$ for which the objective in (\ref{eq:mainlasso3vernon}) becomes unbounded and (\ref{eq:socpnon}) infeasible. We start by rewriting the objective of (\ref{eq:mainlasso3vernon}) in the following way
$$\nu(\sigma\sqrt{\|\g\|_2^2-\|\bar{\h}^+-\frac{\z}{\nu}+\frac{\lambda}{\nu}\|_2^2} -x_{mag}\sum_{i=n-k+1}^{n}\frac{\lambda_i}{\nu}- r_{socp+}).$$
To have the above expression unbounded one needs $\lambda_i=\nu\lambda_i^{(\nu)}$, $\nu\rightarrow\infty$, and the expressions under the square root and in the parenthesis to be positive. Let us then consider
\begin{eqnarray}
\max_{\lambda^{(\nu)}} & & \sigma\sqrt{\|\g\|_2^2-\|\bar{\h}^++\lambda^{\nu}\|_2^2} -x_{mag}\sum_{i=n-k+1}^{n}\lambda_i^{(\nu)}- r_{socp+}\nonumber \\
\mbox{subject to} & & \lambda_i^{(\nu)}\geq 0, 1\leq i\leq n.\label{eq:feas4}
\end{eqnarray}
Let $\lambda^{(feas)}$ be the solution of (\ref{eq:mainlasso3vernon}). Following the line of thought presented in Sections \ref{sec:unsigned} and \ref{sec:probdepframenon} there will be parameters $c_{1}^{(feas)}$ and $c_{2}^{(feas)}$ such that $$\lambda^{(feas)}=[\lambda_1^{(feas)},\lambda_2^{(feas)},\dots,\lambda_{c_{1}^{(feas)}}^{(feas)},0,0,\dots,0,\lambda_{c_{2}^{(feas)}+1}^{(feas)},
\lambda_{c_{2}^{(feas)}+2}^{(feas)},
\dots,\lambda_{n}^{(feas)}]$$ and obviously $c_{1}^{(feas)}\leq n-k$ and $n-k\leq c_{2}^{(feas)}\leq n$. At this point let us assume that these parameters are known and fixed. Then following what was presented in Sections \ref{sec:unsigned} and \ref{sec:probdepframenon} the optimization problem from (\ref{eq:feas4}) can be rewritten in the following way
\begin{eqnarray}
\hspace{-.45in}\max_{\lambda_{c_{1}^{(feas)}+1:n}} & & \sigma\sqrt{\|\g\|_2^2-\|\bar{\h}^+_{c_{1}^++1:n}+\lambda_{c_{1}^{(feas)}+1:n}\|_2^2} -x_{mag}\sum_{i=n-k+1}^{n}\lambda_i-r_{socp+}\nonumber \\
\mbox{subject to}
& & \lambda_i\geq 0, c_{2}^{(feas)}+1\leq i\leq n\nonumber \\
& & \lambda_i=0, c_{1}^{(feas)}+1\leq i\leq c_{2}^{(feas)}.\label{eq:feas5}
\end{eqnarray}
Based on the arguments just above (\ref{eq:feas4}) one has that if the optimal value of the objective in (\ref{eq:feas4}) is positive then the optimization problem in (\ref{eq:mainlasso3vernon}) is unbounded. We will then call the largest $x_{mag}^{(sc)}$ for which the optimal value of the objective in (\ref{eq:feas4}) is positive the feasibility breaking point. To determine such an $x_{mag}^{(sc)}$ we proceed by solving the above optimization problem.
To make writing of what will follow easier we set
\begin{equation}
\xi^{(feas)}=\sigma\sqrt{\|\g\|_2^2-\|\bar{\h}^+_{c_{1}^++1:n}+\lambda_{c_{1}^{(feas)}+1:n}\|_2^2} -x_{mag}\sum_{i=n-k+1}^{n}\lambda_i-r_{socp+},
\label{eq:feasdefxiobjnon}
\end{equation}
and
\begin{eqnarray}
\hspace{-.0in}\widehat{\xi^{(feas)}}=\max_{\lambda_{c_{1}^{(feas)}+1:n}} & & \xi^{(feas)}\nonumber \\
\mbox{subject to}
& & \lambda_i\geq 0, c_{2}^{(feas)}+1\leq i\leq n\nonumber \\
& & \lambda_i=0, c_{1}^{(feas)}+1\leq i\leq c_{2}^{(feas)}.\label{eq:feasdefhatxiobjnon}
\end{eqnarray}
Similarly to what was done in Section \ref{sec:probdepframenon} we then proceed by solving the optimization in (\ref{eq:feas5}) (or the one in (\ref{eq:feasdefhatxiobjnon})) over $\lambda_{c_{1}^{(feas)}+1:n}$. To do so we look at the derivatives with respect to $\lambda_i,c_{2}^{(feas)}+1\leq i\leq n$, of the objective in (\ref{eq:feas5}). Computing the derivatives and equalling them to zero gives
\begin{eqnarray}
\hspace{-.3in}& & \frac{d \xi^{(feas)}}{d \lambda_i}  =  0, c_{2}^{(feas)}+1\leq i\leq n\nonumber \\
& \iff &  \sigma \frac{-(\bar{\h}^+_i+\lambda_i)}
{\sqrt{\|\g\|_2^2-\|\bar{\h}^+_{c_{1}^{(feas)}+1:n}+\lambda_{c_{1}^{(feas)}+1:n}\|_2^2}}-x_{mag} =  0, c_{2}^{(feas)}+1\leq i\leq n\nonumber \\
& \iff &  \lambda_i+
\bar{\h}^+_{i}=-\frac{x_{mag}}{\sigma}\sqrt{\|\g\|_2^2-\|\bar{\h}^+_{c_{1}^{(feas)}+1:n}+\lambda_{c_{1}^{(feas)}+1:n}\|_2^2},
c_{2}^{(feas)}+1\leq i\leq n.\nonumber \\\label{eq:feas6}
\end{eqnarray}
From the second to line in the above equation one then has
\begin{equation*}
(\lambda_i+\bar{\h}^+_{i})^2=\frac{x_{mag}^2}{\sigma^2}(\|\g\|_2^2-\|\bar{\h}^+_{c_{1}^{(feas)}+1:c_{2}^{(feas)}}
+\lambda_{c_{1}^{(feas)}+1:c_{2}^{(feas)}}\|_2^2-\|\bar{\h}_{c_{2}^{(feas)}+1:n}+\lambda_{c_{2}^{(feas)}+1:n}
\|_2^2)
\end{equation*}
and after an easy algebraic transformation
\begin{equation}
\hspace{-.0in}(\lambda_i+\bar{\h}^+_{i})^2=\frac{x_{mag}^2}{\sigma^2+(n-c_{2}^{(feas)})x_{mag}^2}(\|\g\|_2^2-\|\bar{\h}^+_{c_{1}^{(feas)}+1:c_{2}^{(feas)}}
\|_2^2).\label{eq:feas7}
\end{equation}
Using (\ref{eq:feas7}) we further have
\begin{equation}
\hspace{-.0in}\sqrt{\|\g\|_2^2-\|\bar{\h}^+_{c_{1}^{(feas)}+1:n}
+\lambda_{c_{1}^{(feas)}+1:n}\|_2^2}
=\frac{\sigma}{\sqrt{\sigma^2+(n-c_{2}^{(feas)})x_{mag}^2}}\sqrt{\|\g\|_2^2-\|\bar{\h}^+_{c_{1}^{(feas)}+1:c_{2}^{(feas)}}\|_2^2}.
\label{eq:feas8}
\end{equation}
Plugging the value for $\lambda_i$ from (\ref{eq:feas7}) in (\ref{eq:feasdefxiobjnon}) gives
\begin{eqnarray}
\xi^{(feas)} & = & \sigma\sqrt{\|\g\|_2^2-\|\bar{\h}^+_{c_{1}^{(feas)}+1:n}+\lambda_{c_{1}^{(feas)}+1:n}\|_2^2} -x_{mag}\sum_{i=n-k+1}^{n}\lambda_i- r_{socp+}\nonumber \\
& = & \frac{\sigma^2+(n-c_{2}^{(feas)})x_{mag}^2}{\sigma}\sqrt{\|\g\|_2^2-\|\bar{\h}^+_{c_{1}^{(feas)}+1:n}+\lambda_{c_{1}^{(feas)}+1:n}\|_2^2}\nonumber \\
& + & x_{mag}\sum_{i=c_2^{(feas)}+1}^{n}\bar{\h}^+_i-r_{socp+}.
\label{eq:feas9}
\end{eqnarray}
Combining (\ref{eq:feas8}) and (\ref{eq:feas9}) we finally obtain
\begin{equation}
\hspace{-.3in}\xi^{(feas)}
 =  \sqrt{\sigma^2+(n-c_{2}^{(feas)})x_{mag}^2}
\sqrt{\|\g\|_2^2-\|\bar{\h}^+_{c_{1}^{(feas)}+1:c_{2}^{(feas)}}\|_2^2}
- (r_{socp+}- x_{mag}\sum_{i=c_2^{(feas)}+1}^{n}\bar{\h}^+_i).
\label{eq:feas10}
\end{equation}
Following what was done in Sections \ref{sec:unsigned} and \ref{sec:probdepframenon} we have that a combination of (\ref{eq:feas7}) and (\ref{eq:feas9}) (together with the fact that the elements of $\lambda$ are nonnegative) gives the following two equations that can be used to determine $c_1^{(feas)}$ and $c_2^{(feas)}$
\begin{eqnarray}
\bar{\h}^+_{c_{2}^{(feas)}}+\frac{x_{mag}}{\sqrt{\sigma^2+(n-c_{2}^{(feas)})x_{mag}^2}}\sqrt{\|\g\|_2^2-\|\bar{\h}^+_{c_{1}^{(feas)}+1:c_{2}^{(feas)}}\|_2^2}& =& 0\nonumber \\
\bar{\h}^+_{c_1^{(feas)}} & \leq & 0.\nonumber \\
\label{eq:feas11}
\end{eqnarray}
Let $\widehat{c_1^{(feas)}}$ and $\widehat{c_2^{(feas)}}$ be the solution of (\ref{eq:feas11}). Then from (\ref{eq:feas10}) we have
\begin{equation}
\hspace{-.0in}\widehat{\xi^{(feas)}}
 =  \sqrt{\sigma^2+(n-\widehat{c_{2}^{(feas)}})x_{mag}^2}
\sqrt{\|\g\|_2^2-\|\bar{\h}^+_{\widehat{c_{1}^{(feas)}}+1:\widehat{c_{2}^{(feas)}}}\|_2^2}
- (r_{socp+}- x_{mag}\sum_{i=\widehat{c_2^{(feas)}}+1}^{n}\bar{\h}^+_i).
\label{eq:feas12}
\end{equation}
Combination of (\ref{eq:feas11}) and (\ref{eq:feas12}) is conceptually enough to determine the feasibility breaking point for $x_{mag}^{(sc)}$. The part that remains is a computation of all unknown quantities that appear in (\ref{eq:feas11}) and (\ref{eq:feas12}). To do it we will as usual substantially rely on what was shown in previous section and basically in \cite{StojnicCSetam09,StojnicGenSocp10}.

The first thing to resolve is (\ref{eq:feas11}). Since all random quantities concentrate we will again deal with the expected values. To compute the solution of (\ref{eq:feas11}), $\widehat{c_1^{(feas)}}$ and $\widehat{c_2^{(feas)}}$,  we will need the following expected values
\begin{equation}
E\|\g\|_2^2, E\|\bar{\h}^+_{c_{1}^{(feas)}+1:n-k}\|_2^2, E\|\bar{\h}^+_{n-k+1:c_2^+}\|_2^2,
E\sum_{i=c_{2}^{(feas)}+1}^{n}\bar{\h}^+_i.\label{eq:feasexp1non}
\end{equation}
As in previous sections we easily have
\begin{equation}
E\|\g\|_2^2=m.\label{eq:feasexpgnon}
\end{equation}
Let $c_{1}^{(feas)}=(1-\theta_1^{(feas)})n$ and $c_{2}^{(feas)}=\theta_2^{(feas)} n$ where $\theta_{1}^{(feas)}$ and $\theta_{2}^{(feas)}$ are constants independent of $n$. Then as shown in \cite{StojnicCSetam09,StojnicGenSocp10}
\begin{equation}
\lim_{n\rightarrow\infty}\frac{E\|\bar{\h}^+_{c_{1}^{(feas)}+1:n-k}\|_2^2}{n}  =  \frac{1-\beta_{w}^+}{\sqrt{2\pi}}\left (\frac{\sqrt{2}(\erfinv(2\frac{1-\theta_{1}^{(feas)}}{1-\beta_{w}^+}-1))}{e^{(\erfinv(2\frac{1-\theta_{1}^{(feas)}}{1-\beta_{w}^+}-1))^2}}\right )+\theta_{1}^{(feas)}-\beta_{w}^+.\label{eq:feasexpnorm1non}
\end{equation}
where we of course recall that $\beta_{w}^+=\frac{k}{n}$. Also, as in Section \ref{sec:unsigned} and earlier in \cite{StojnicCSetam09,StojnicGenSocp10}
we have
\begin{equation}
\lim_{n\rightarrow\infty}\frac{E\|\bar{\h}^+_{n-k+1:c_{2}^{(feas)}}\|_2^2}{n}  =  \frac{\beta_{w}^+}{\sqrt{2\pi}}\left (\frac{\sqrt{2}\erfinv(2\frac{1-\theta_{2}^{(feas)}}{\beta_{w}^+}-1)}{e^{(\erfinv(2\frac{1-\theta_{2}^{(feas)}}{\beta_{w}^+}-1))^2}}\right )+\theta_{2}^{(feas)}-1+\beta_{w}^+.\label{eq:feasexpnormh2non}
\end{equation}
Following further what was established in \cite{StojnicCSetam09,StojnicGenSocp10} we also have
\begin{equation}
\lim_{n\rightarrow\infty}\frac{E(\sum_{i=c_2^{(feas)}+1:n}^{n}\bar{\h}^+_{i})}{n} =
-\left (\beta_{w}^+\sqrt{\frac{1}{2\pi}}e^{-(\erfinv(2\frac{1-\theta_{2}^{(feas)}}{\beta_{w}^+}-1))^2}\right ).\label{eq:feasexpprod1non}
\end{equation}
The only other thing that we will need in order to be able to compute $\widehat{c_{1}^{(feas)}}$ and $\widehat{c_{2}^{(feas)}}$ (besides the expectations from (\ref{eq:feasexp1non})) are the following inequalities related to the behavior of $\bar{\h}^+_{c_{1}^{(feas)}}$ and $\bar{\h}^+_{c_{2}^{(feas)}}$. Again, as shown in \cite{StojnicCSetam09,StojnicGenSocp10}
\begin{eqnarray}
P(\sqrt{2}\erfinv ((1+\epsilon_1^{\bar{\h}^+_{c_{1}^{(feas)}}})(2\frac{1-\theta_{1}^{(feas)}}{1-\beta_{w}^+}-1))\leq \bar{\h}^+_{c_{1}^{(feas)}}) & \leq & e^{-\epsilon_2^{\bar{\h}^+_{c_{1}^{(feas)}}} n}\nonumber \\
P(\sqrt{2}\erfinv ((1+\epsilon_1^{\bar{\h}^+_{c_{2}^{(feas)}}})(2\frac{1-\theta_{2}^{(feas)}}{\beta_{w}^+}-1))\leq \bar{\h}^+_{c_{2}^{(feas)}}) & \leq & e^{-\epsilon_2^{\bar{\h}^+_{c_{2}^{(feas)}}} n}.\label{eq:mainbrpointnon}
\end{eqnarray}
where $\epsilon_1^{\bar{\h}^+_{c_{1}^{(feas)}}}>0$ and $\epsilon_1^{\bar{\h}^+_{c_{2}^{(feas)}}}>0$ are arbitrarily small constants and $\epsilon_2^{\bar{\h}^+_{c_{1}^{(feas)}}}$ and $\epsilon_2^{\bar{\h}^+_{c_{2}^{(feas)}}}$ are constants dependent on $\epsilon_1^{\bar{\h}^+_{c_{1}^{(feas)}}}$ and $\epsilon_1^{\bar{\h}^+_{c_{2}^{(feas)}}}$, respectively, but independent of $n$. Also, we find it useful for what follows to introduce the following definitions
\begin{eqnarray}
s_{feas} & = & \sum_{i=c_2^{(feas)}+1}^{n}\bar{\h}^+_i\nonumber \\
d_{feas} & = & \sum_{i=c_1^{(feas)}+1}^{c_2^{(feas)}}(\bar{\h}^+_i)^2 \nonumber \\
r_{feas} & = & r_{socp+}- x_{mag}\sum_{i=c_2^{(feas)}+1}^{n}\bar{\h}^+_i.
\label{eq:feasdefdepnon}
\end{eqnarray}
At this point we have all the necessary ingredients to determine $\widehat{c_{1}^{(feas)}}$ and $\widehat{c_{2}^{(feas)}}$ and consequently $\widehat{\xi^{(feas)}}$, or to be more precise their concentrating points. The following theorem then provides a systematic way of doing so.
\begin{theorem}
Assume the setup of Theorem \ref{thm:mainlassonon}. Let the nonzero components of $\xtilde$ have magnitude $x_{mag}$ and let $\bar{\h}^+$ be as defined in (\ref{eq:defhbarnon}). Further, let $r_{socp+}^{(sc)}=\lim_{n\rightarrow\infty}\frac{r_{socp+}}{\sqrt{n}}$ and $x_{mag}^{(sc)}=\lim_{n\rightarrow\infty}\frac{x_{mag}}{\sqrt{n}}$. Also, let $\widehat{\xi^{(feas)}}$ be as defined in (\ref{eq:feasdefhatxiobjnon}). Let $\alpha=\frac{m}{n}$ and $\beta_{w}^+=\frac{k}{n}$ be fixed. Consider the following
\begin{equation*}
S(\theta_{1}^{(feas)},\theta_{2}^{(feas)})  =  \lim_{n\rightarrow\infty}\frac{E s_{dep+}}{n}  = -\left (\beta_{w}^+\sqrt{\frac{1}{2\pi}}e^{-(\erfinv(2\frac{1-\theta_{2}^{(feas)}}{\beta_{w}^+}-1))^2}\right )
\end{equation*}
\begin{multline*}
D(\theta_{1}^{(feas)},\theta_{2}^{(feas)})  =  \lim_{n\rightarrow\infty}\frac{E d_{dep+}}{n}  =  \frac{1-\beta_{w}^+}{\sqrt{2\pi}}\left (\frac{\sqrt{2}(\erfinv(2\frac{1-\theta_{1}^{(feas)}}{1-\beta_{w}^+}-1))}{e^{(\erfinv(2\frac{1-\theta_{1}^{(feas)}}{1-\beta_{w}^+}-1))^2}}\right )+\theta_{1}^{(feas)}-\beta_{w}^+\\+\frac{\beta_{w}^+}{\sqrt{2\pi}}\left (\frac{\sqrt{2}\erfinv(2\frac{1-\theta_{2}^{(feas)}}{\beta_{w}^+}-1)}{e^{(\erfinv(2\frac{1-\theta_{2}^{(feas)}}{\beta_{w}^+}-1))^2}}\right )+\theta_{2}^{(feas)}-1+\beta_{w}^+
\end{multline*}
\begin{multline*}
R(\theta_{2}^{(feas)})  =  \lim_{n\rightarrow\infty}\frac{E r_{dep+}}{\sqrt{n}}  = r_{socp+}^{(sc)}-x_{mag}^{(sc)}S(\theta_{1}^{(feas)},\theta_{2}^{(feas)})+x_{mag}^{(sc)}(1-\theta_{2}^{(feas)})\\
=r_{socp+}^{(sc)}+x_{mag}^{(sc)}\left (\beta_{w}^+\sqrt{\frac{1}{2\pi}}e^{-(\erfinv(2\frac{1-\theta_{2}^{(feas)}}{\beta_{w}^+}-1))^2}\right )
\end{multline*}
\begin{eqnarray}
F(\theta_{1}^{(feas)}) & = & \sqrt{2}\erfinv (2\frac{1-\theta_{1}^{(feas)}}{1-\beta_{w}^+}-1)\nonumber \\
G(\theta_{2}^{(feas)}) & = & \sqrt{2}\erfinv (2\frac{1-\theta_{2}^{(feas)}}{\beta_{w}^+}-1).\label{eq:feascompthmcond1non}
\end{eqnarray}
Let the pair ($\widehat{\theta_{1}^{(feas)}}$, $\widehat{\theta_{2}^{(feas)}}$) be the solution of the following two equations
\begin{eqnarray}
\hspace{-.6in}G(\theta_{2}^{(feas)})+\frac{x_{mag}^{(sc)}\sqrt{(\alpha
-D(\theta_{1}^{(feas)},\theta_{2}^{(feas)}))}}
{\sqrt{\sigma^2+(1-\theta_{2}^{(feas)})(x_{mag}^{(sc)})^2}} & = & 0\nonumber \\
F(\theta_{1}^{(feas)}) & = & 0.\nonumber \\\label{eq:feasmainthmc1c2c3non}
\end{eqnarray}
Then the feasibility breaking point for $x_{mag}^{(sc)}$ can be determined as the solution of
\begin{multline}
\hspace{-.2in}\lim_{n\rightarrow\infty}\frac{E\widehat{\xi^{(feas)}}}{\sqrt{n}}  =
\sqrt{\sigma^2+(1-\widehat{\theta_{2}^{(feas)}})(x_{mag}^{(sc)})^2}\sqrt{(\alpha-
D(\widehat{\theta_{1}^{(feas)}},\widehat{\theta_{2}^{(feas)}}))}
-R(\widehat{\theta_{2}^{(feas)}})\\=-(\sigma^2+(1-\widehat{\theta_{2}^{(feas)}})(x_{mag}^{(sc)})^2)
\frac{G(\widehat{\theta_2^{(feas)}})}{x_{mag}^{(sc)}}-R(\widehat{\theta_{2}^{(feas)}})=0.
\label{eq:feasmainthmnuwgenxiprimnon}
\end{multline}
\label{thm:feasmaincomperrornon}
\end{theorem}
\begin{proof}
Follows from Theorem \ref{thm:mainlassonon} based on the discussion presented above and a combination of (\ref{eq:feasdefhatxiobjnon}), (\ref{eq:feas11}), and (\ref{eq:feasdefdepnon}).
\end{proof}

\noindent \textbf{Remark 1:} It is relatively easy to see that $\widehat{\theta_1^{(feas)}}=\frac{1+\beta_w^+}{2}$ which somewhat simplifies the expression for $D(\theta_{1}^{(feas)},\theta_{2}^{(feas)})$. For the completeness though we chose to present the results in the above theorem in a general form.

\noindent \textbf{Remark 2:} Another way to deal with the unboundedness (infeasibility) is to look at (\ref{eq:mainlasso13non}) and recognize that
$(a_{dep+}^2-\|\g\|_2^2+d_{dep+})$ needs to be negative so that (\ref{eq:mainlasso3vernon}) is bounded. We presented the results in Theorem \ref{thm:mainlassonon} by assuming that all relevant parameters are such that (\ref{eq:mainlasso3vernon}) is bounded. Instead one could actually characterize them through one of the above approaches. However, we thought that it would complicate the presentation and opted for the current exposition.

The results from the above theorem can be used to compute the breaking feasibility point for $x_{mag}^{(sc)}$ for particular values of $\beta_{w}^+$, $\alpha$, $\sigma$, and $r_{socp+}$. In a part of the following subsection we will present a subset of such results.

\subsubsection{Theoretical predictions} \label{sec:unsignedtheoryprednon}

In this subsection we present the theoretical predictions one can get based on the result of the previous sections. Similarly to what was done in Section \ref{sec:unsignedtheorypred} we will split the presentation of the results into several parts.

\textbf{\underline{\emph{1) $\frac{E\|\w_{dep+}\|_2}{\sigma}=\frac{E\|\w_{socp+}\|_2}{\sigma}$ as a function of $x_{mag}^{(sc)}$}}}

To present this portion (as well as several others that will follow) of theoretical results we will as in Section \ref{sec:unsignedtheorypred} look at three regimes: 1) low $\alpha$-, medium $\alpha$-, and high $\alpha$-regime. For each of the regimes we will show the theoretical results for $\frac{E\|\w_{dep+}\|_2}{\sigma}=\frac{E\|\w_{socp+}\|_2}{\sigma}$ as a function of $x_{mag}^{(sc)}$. We will again take $\alpha=0.3$ as a representative of the low $\alpha$-regime, $\alpha=0.5$ as a representative of the medium $\alpha$-regime, and $\alpha=0.7$ as a representative of the high $\alpha$-regime. We will consider $r_{socp+}=r_{socp+}^{(opt)}=\sigma\sqrt{\frac{\alpha n}{1+\rho^2}}$. For each of the $\alpha$-regimes we will look at two different sub-regimes: low $\beta_w^+$- and high $\beta_w^+$-regime (which based on results from \cite{StojnicGenSocp10} is equivalent to low $\rho$- and high $\rho$-regimes). For each of these two sub-regimes $\beta_w^+$ will be selected based on the curves obtained in \cite{StojnicGenSocp10} (or those obtained in \cite{StojnicGenLasso10}) in the following way. In the low $\beta_w^+$ sub-regime we will set $\rho=2$ and $r_{socp+}=r_{socp+}^{(opt)}=\sigma\sqrt{\frac{\alpha n}{5}}$ whereas in the high $\beta_w^+$ sub-regime we will set $\rho=3$ and $r_{socp+}=r_{socp+}^{(opt)}=\sigma\sqrt{\frac{\alpha n}{10}}$. At the same time from \cite{StojnicGenSocp10} we will have $r_{socp+}=r_{socp+}^{(opt)}=\sigma\sqrt{(\alpha-\alpha_w^+)n}$ where $\alpha_w^+$ and $\beta_w^+$ are such that (\ref{eq:fundl1non}) holds (we also recall on \cite{StojnicGenSocp10} where it was reasoned that the low $\beta$ regime is selected so that the pair $(\alpha,\beta_w^+)$ is well below the fundamental characterization (\ref{eq:fundl1non}) whereas the high $\beta$ regime is selected so that the pair $(\alpha,\beta_w^+)$ is closer to the fundamental characterization (\ref{eq:fundl1non})). The values for $\frac{E\|\w_{dep+}\|_2}{\sigma}=\frac{E\|\w_{socp+}\|_2}{\sigma}$ one can then get through the results of Theorem \ref{thm:maincomperrornon} for such $(\alpha,\beta_w^+)$ pairs and $r_{socp+}^{(opt)}$ are shown in Figure \ref{fig:errorvarxnon} as functions of $x_{mag}^{(sc)}$.
\begin{figure}[htb]
\begin{minipage}[b]{.33\linewidth}
\centering
\centerline{\epsfig{figure=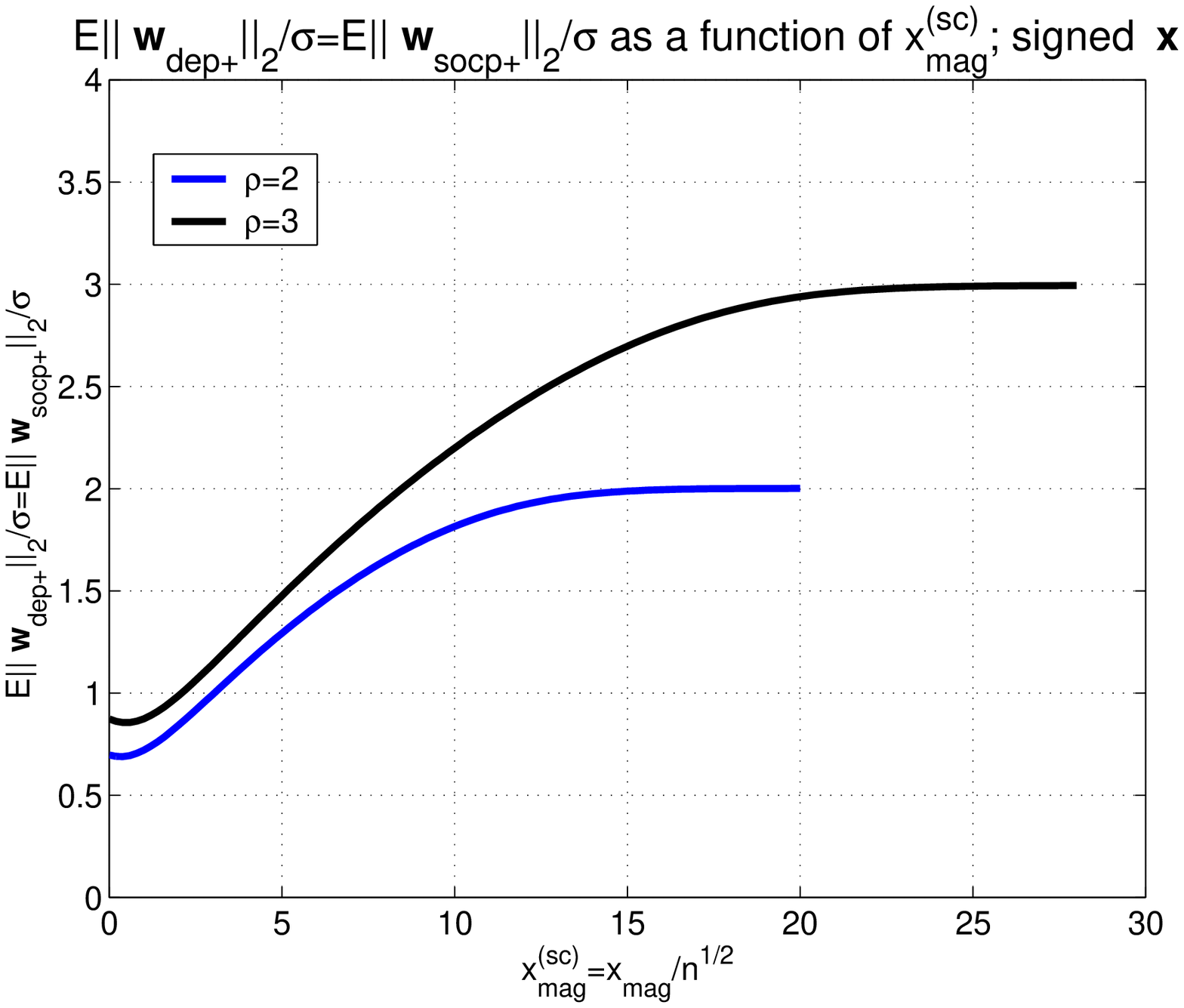,width=5.3cm,height=5.3cm}}
\end{minipage}
\begin{minipage}[b]{.33\linewidth}
\centering
\centerline{\epsfig{figure=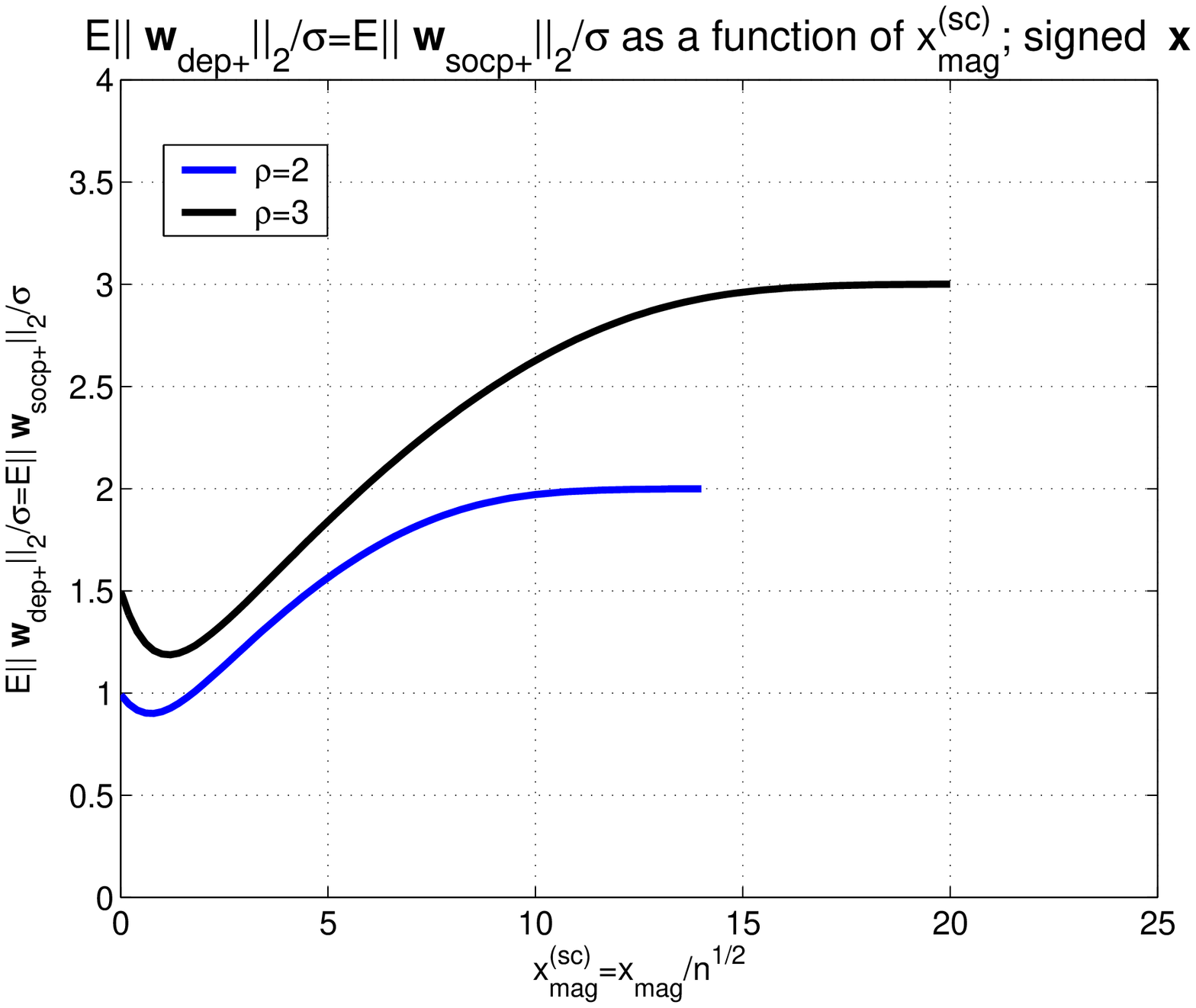,width=5.3cm,height=5.3cm}}
\end{minipage}
\begin{minipage}[b]{.33\linewidth}
\centering
\centerline{\epsfig{figure=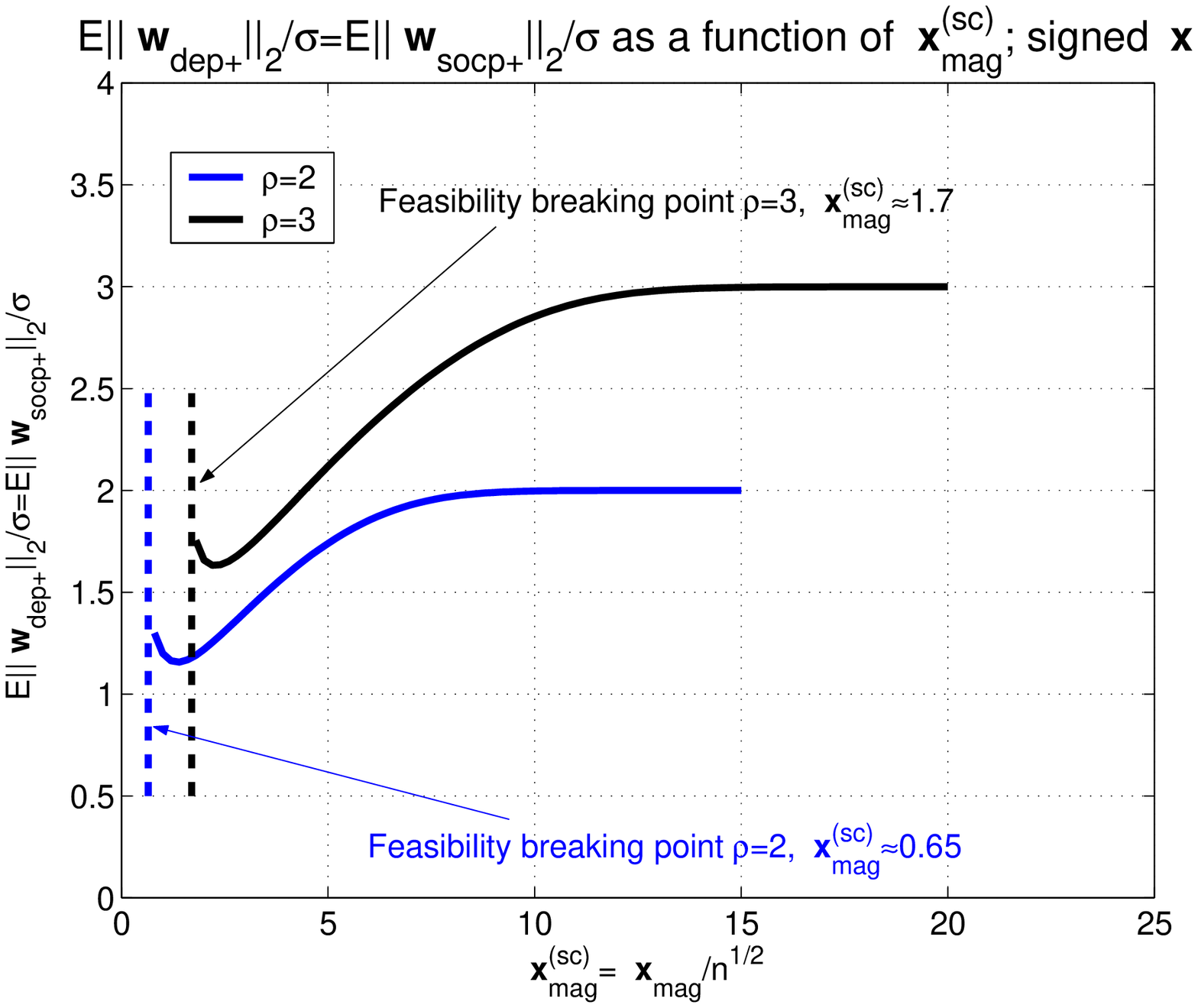,width=5.3cm,height=5.3cm}}
\end{minipage}
\caption{$\frac{E\|\w_{dep+}\|_2}{\sigma}=\frac{E\|\w_{socp+}\|_2}{\sigma}$ as a function of $x_{mag}^{(sc)}$; $r_{socp+}=\sqrt{\frac{\alpha n}{1+\rho^2}}$; left --- $\alpha=0.3$, center --- $\alpha=0.5$, right --- $\alpha=0.7$}
\label{fig:errorvarxnon}
\end{figure}
As can be seen from Figure \ref{fig:errorvarxnon}, the values of $\frac{E\|\w_{dep+}\|_2}{\sigma}=\frac{E\|\w_{socp+}\|_2}{\sigma}$ converge to $\rho$ as $x_{mag}^{(sc)}$ increases. This is of course in agreement with \cite{StojnicGenSocp10} where it was demonstrated that for $r_{socp+}^{(opt)}$ one has $\rho=\frac{\|\w_{socp+}\|_2}{\sigma}$ with overwhelming probability. Also, this is in a agreement with a similar conclusion made in Section \ref{sec:unsignedtheorypred}. Furthermore, as it was the case in Section \ref{sec:unsignedtheorypred}, the convergence is ``faster" (or happens for smaller $x_{mag}^{(sc)}$) for larger $\alpha$.

On the other hand one should observe from Figure \ref{fig:errorvarxnon} that for $\alpha=0.7$ the optimization problem in (\ref{eq:mainlasso3vernon}) is infeasible with overwhelming probability for $x_{mag}^{(sc)}$ below $\approx 1.7$ in high $\beta_w^+$ regime and for $x_{mag}^{(sc)}$ below $\approx 0.65$ in low $\beta_w^+$ regime.

\textbf{\underline{\emph{2) $\frac{Ef_{obj+}}{\sqrt{n}}=\frac{E\xi_{prim+}^{(dep)}(\sigma,\g,\h,x_{mag},r_{socp+})}{\sqrt{n}}$ as a function of $x_{mag}^{(sc)}$}}}

Similarly to what was discussed above (and is related to $\|\w_{socp+}\|_2$ and $\|\w_{dep+}\|_2$) one can determine the concentrating points of $f_{obj+}$ and $\xi_{prim+}^{(dep)}$ also as functions of $x_{mag}^{(sc)}$. As in Section \ref{sec:unsignedtheorypred}, to present these results we restrict ourselves to the medium $\alpha$-regime, or in other words to $\alpha=0.5$. As in part 1) above we again choose $r_{socp+}=r_{socp+}^{(opt)}=\sigma\sqrt{\frac{\alpha n}{1+\rho^2}}$ and consider low $\rho=2$- and high $\rho=3$- regime.
The obtained results are shown in Figure \ref{fig:objvarxnon}.
\begin{figure}[htb]
\begin{minipage}[b]{1\linewidth}
\centering
\centerline{\epsfig{figure=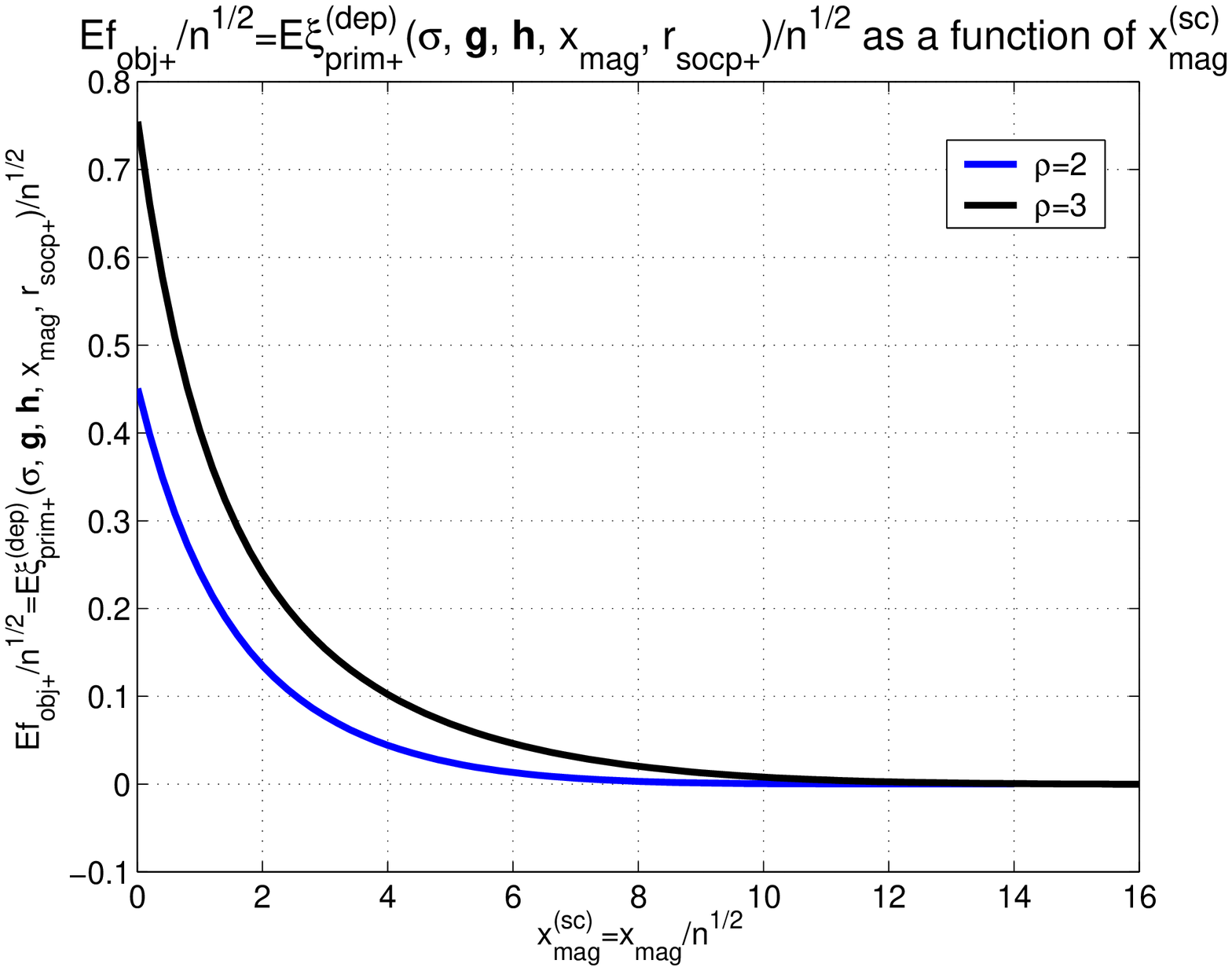,width=11cm,height=9cm}}
\end{minipage}
\caption{$\frac{Ef_{obj+}}{\sqrt{n}}=\frac{E\xi_{prim+}^{(dep)}(\sigma,\g,\h,x_{mag},r_{socp+})}{\sqrt{n}}$ as a function of $x_{mag}^{(sc)}$; $r_{socp+}=\sqrt{\frac{\alpha n}{1+\rho^2}}$; $\alpha=0.5$}
\label{fig:objvarxnon}
\end{figure}
As can be seen from Figure \ref{fig:objvarxnon} $\frac{Ef_{obj+}}{\sqrt{n}}$ is larger for larger $\rho$.

\textbf{\underline{\emph{3) $\frac{E\|\w_{dep+}\|_2}{\sigma}=\frac{E\|\w_{socp+}\|_2}{\sigma}$ as a function of $x_{mag}^{(sc)}$; varying $r_{socp+}$}}}

Another interesting set of results relates to possible variations in the $r_{socp+}$ that can be used in (\ref{eq:socp}). The results that we presented above assume an optimal choice for $r_{socp+}$ (in a sense defined in \cite{StojnicGenSocp10}). Namely, they assume that for a fixed pair $(\alpha,\beta_w^+)$ one chooses $r_{socp+}=r_{socp+}^{(opt)}=\sigma\sqrt{(\alpha-\alpha_w^+)n}$ where $\alpha_w^+$ and $\beta_w^+$ are such that (\ref{eq:fundl1non}) holds. In the worst-case scenario (or in the generic scenario as we referred to it in \cite{StojnicGenSocp10}) one has that choice $r_{socp+}^{(opt)}$ offers the minimal norm-2 of the error vector. However, such a scenario assumes particular $\xtilde$'s which leaves a possibility that for a wide range of other $\xtilde$'s the performance of the SOCP from (\ref{eq:socp}) in the $\ell_2$ norm of the error vector sense can be more favorable. Of course as shown in Figure (\ref{fig:errorvarxnon}) this indeed happens to be the case. On the other hand that also leaves an option that one can possibly choose a different $r_{socp+}$ and get say a smaller norm-2 of the error vector for various different $\xtilde$. Below we present a few results in this direction.

We will consider again only the medium or $\alpha=0.5$ regime. For two sets of different values of $\rho$, $r_{socp+}=r_{socp+}^{(opt)}$, and $\beta_w^+$ we presented the results for $\frac{E\|\w_{dep+}\|_2}{\sigma}=\frac{E\|\w_{socp+}\|_2}{\sigma}$ in Figure \ref{fig:errorvarxnon}. In addition to that we now in Figure \ref{fig:errorvarxvarrhonon} show similar results one can get through Theorem \ref{thm:maincomperrornon} for two different choices of $r_{socp+}$. To be more precise, for $\rho=2$ we choose the same $\alpha$ and $\beta_w^+$ as in Figure and only vary $r_{socp+}$ over $\{\sigma\sqrt{0.05\alpha n},\sigma\sqrt{0.2\alpha n}\sigma\sqrt{0.6\alpha n}\}$. Clearly, choice $\sigma\sqrt{0.05\alpha n}$ is smaller than $r_{socp+}^{(opt)}=\sigma\sqrt{0.2\alpha n}$ whereas choice
$\sigma\sqrt{0.6\alpha n}$ is larger than $r_{socp+}^{(opt)}=\sigma\sqrt{0.2\alpha n}$. On the other hand for $\rho=3$ we choose the same $\alpha$ and $\beta_w^+$ as we have chosen for $\rho=3$ in Figure \ref{fig:errorvarxnon} and vary $r_{socp+}$ but this time over $\{\sigma\sqrt{0.05\alpha n},\sigma\sqrt{0.1\alpha n}\sigma\sqrt{0.5\alpha n}\}$. Again, clearly, choice $\sigma\sqrt{0.05\alpha n}$ is smaller than $r_{socp+}^{(opt)}=\sigma\sqrt{0.1\alpha n}$ whereas choice
$\sigma\sqrt{0.6\alpha n}$ is larger than $r_{socp+}^{(opt)}=\sigma\sqrt{0.2\alpha n}$. It is rather obvious but we mention for the completeness that the middle $r_{socp+}$ choices for both, $\rho=2$ and $\rho=3$, cases correspond to the center plot in Figure \ref{fig:errorvarxnon}.
\begin{figure}[htb]
\begin{minipage}[b]{.5\linewidth}
\centering
\centerline{\epsfig{figure=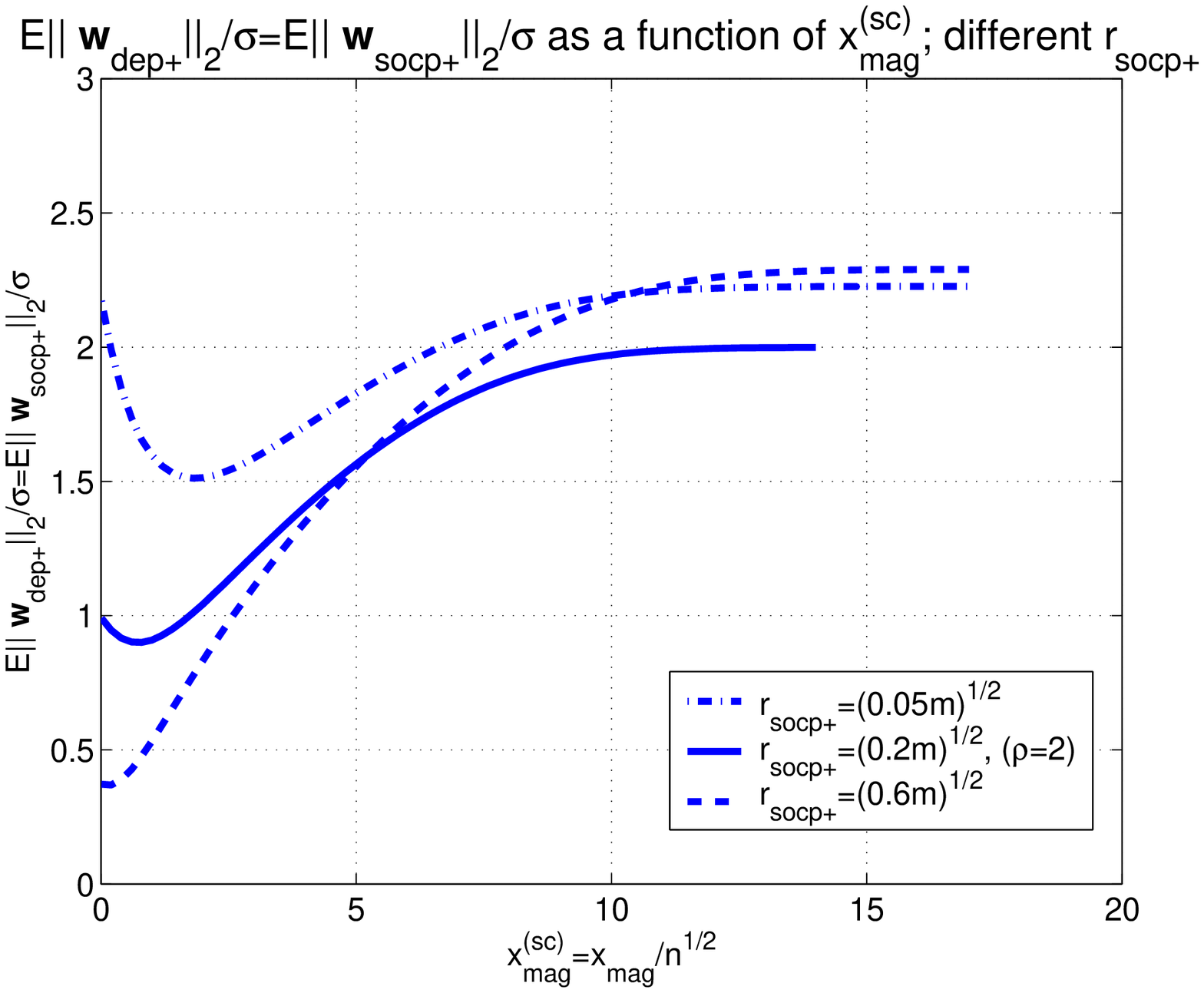,width=8cm,height=7cm}}
\end{minipage}
\begin{minipage}[b]{.5\linewidth}
\centering
\centerline{\epsfig{figure=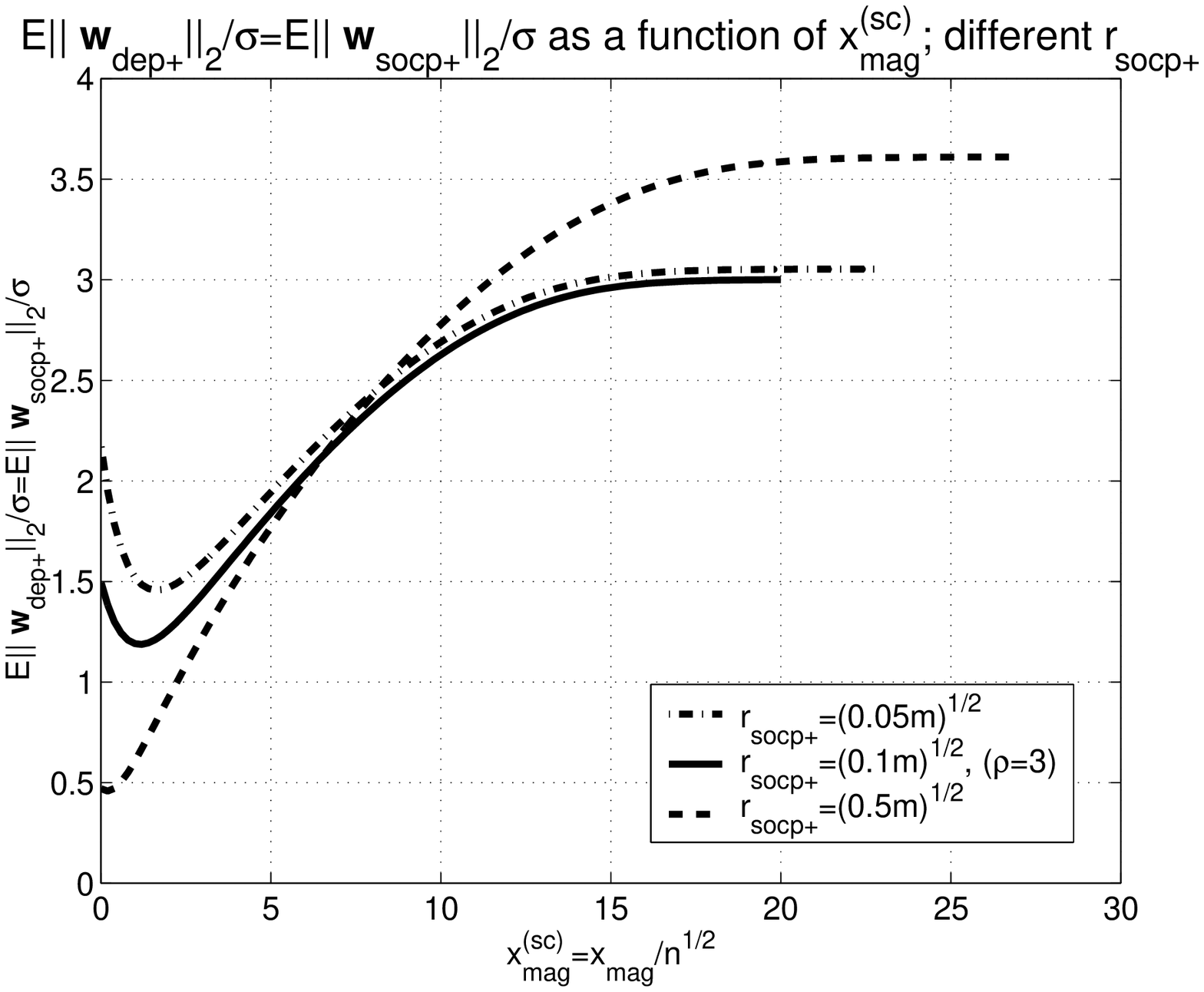,width=8cm,height=7cm}}
\end{minipage}
\caption{$\frac{E\|\w_{dep+}\|_2}{\sigma}=\frac{E\|\w_{socp+}\|_2}{\sigma}$ as a function of $x_{mag}^{(sc)}$ for different $r_{socp+}$; left --- $\rho=2$, $r_{socp+}\in\{\sigma\sqrt{0.05\alpha n},\sigma\sqrt{0.2\alpha n}\sigma\sqrt{0.6\alpha n}\}$; right --- $\rho=3$, $r_{socp+}\in\{\sigma\sqrt{0.05\alpha n},\sigma\sqrt{0.1\alpha n}\sigma\sqrt{0.5\alpha n}\}$}
\label{fig:errorvarxvarrhonon}
\end{figure}
Based on the results presented in Figure \ref{fig:errorvarxvarrhonon} one can then make the same two observations that we made earlier related to the results presented in Figure \ref{fig:errorvarxvarrho}. The first one is that Figure \ref{fig:errorvarxvarrhonon} suggests that if $r_{socp+}$ is smaller than $r_{socp+}^{(opt)}$ then $\frac{E\|\w_{socp+}\|_2}{\sigma}$ could be larger than the one that can be obtained for $r_{socp+}^{(opt)}$. This actually happens to be the case. As in Section \ref{sec:unsignedtheorypred} we skip the details of this simple exercise. The second observation is that if $r_{socp+}$ is larger than $r_{socp+}^{(opt)}$ then for certain $\xtilde$ $\frac{E\|\w_{socp+}\|_2}{\sigma}$ could be smaller than the one that can be obtained for $r_{socp+}^{(opt)}$. This of course suggests that a choice of $r_{socp+}$ larger than $r_{socp+}^{(opt)}$ could be more favorable in certain applications and for a particular measure of performance. However, if one has no a priori available knowledge about $\xtilde$ then adapting $r_{socp+}$ beyond $r_{socp+}^{(opt)}$ would be hard.

\textbf{\underline{\emph{4) $\frac{Ef_{obj+}}{\sqrt{n}}=\frac{E\xi_{prim+}^{(dep)}(\sigma,\g,\h,x_{mag},r_{socp+})}{\sqrt{n}}$ as a function of $x_{mag}^{(sc)}$; varying $r_{socp+}$}}}

Similarly to what was done above in part 2) one can also determine the theoretical predictions for $\frac{Ef_{obj+}}{\sqrt{n}}=\frac{E\xi_{prim+}^{(dep)}}{\sqrt{n}}$ for a varying $r_{socp+}$. As in parts 2) and 3) above, we restrict our attention only to the medium $\alpha=0.5$ regime. We also assume exactly the same scenarios as in part 3). The obtained results are shown in Figure \ref{fig:objvarxvarrho}.
\begin{figure}[htb]
\begin{minipage}[b]{.5\linewidth}
\centering
\centerline{\epsfig{figure=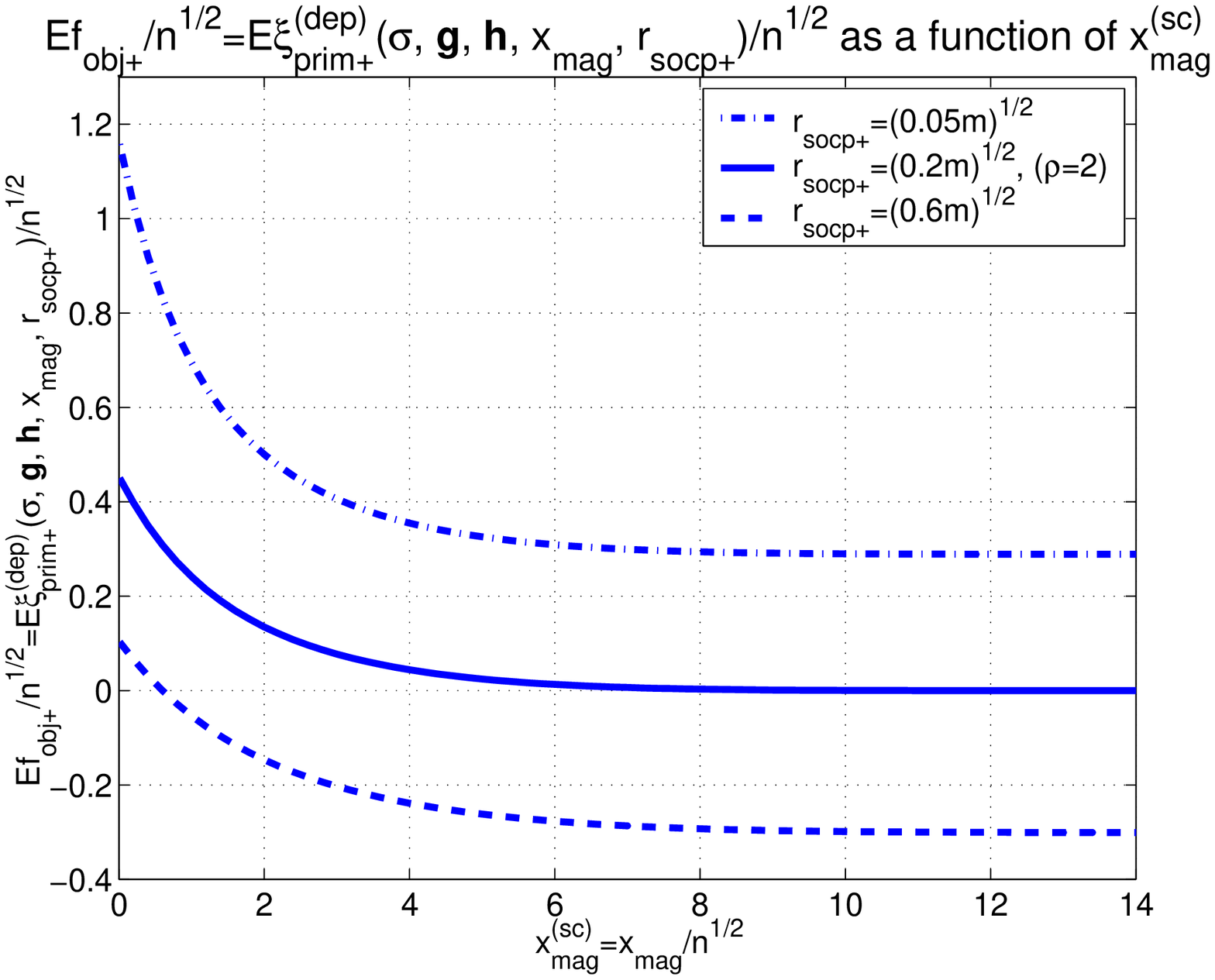,width=8cm,height=7cm}}
\end{minipage}
\begin{minipage}[b]{.5\linewidth}
\centering
\centerline{\epsfig{figure=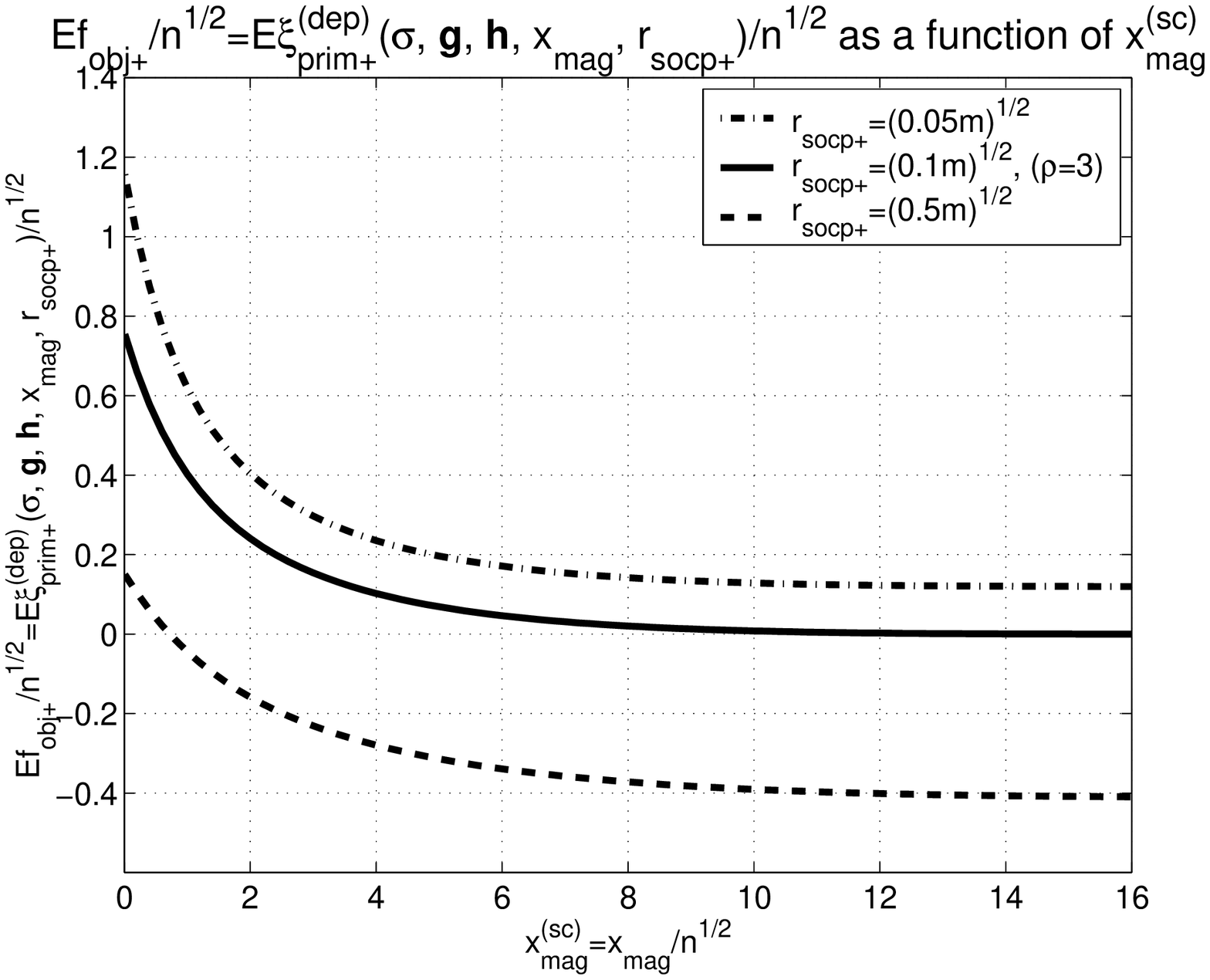,width=8cm,height=7cm}}
\end{minipage}
\caption{$\frac{Ef_{obj+}}{\sqrt{n}}=\frac{E\xi_{prim+}^{(dep)}(\sigma,\g,\h,x_{mag},r_{socp+})}{\sqrt{n}}$ as a function of $x_{mag}^{(sc)}$ for different $r_{socp+}$; left --- $\rho=2$, $r_{socp+}\in\{\sigma\sqrt{0.05\alpha n},\sigma\sqrt{0.2\alpha n}\sigma\sqrt{0.6\alpha n}\}$; right --- $\rho=3$, $r_{socp+}\in\{\sigma\sqrt{0.05\alpha n},\sigma\sqrt{0.1\alpha n}\sigma\sqrt{0.5\alpha n}\}$}
\label{fig:objvarxvarrhonon}
\end{figure}
As in part 2) Figure \ref{fig:objvarxvarrhonon} shows that $\frac{Ef_{obj+}}{\sqrt{n}}$ is larger for larger $\rho$. On the other hand it also shows that
$\frac{Ef_{obj+}}{\sqrt{n}}$ decreases as $r_{socp+}$ increases. This also follows rather trivially from the structure of (\ref{eq:socpnon}) or in a way discussed in the corresponding part of Section \ref{sec:unsignedtheorypred}.

As in Section \ref{sec:unsignedtheorypred} we conducted massive numerical experiments and again found that the results one can get through them are in a firm agreement with what the presented theory predicts. In the next subsection we present a sample of the results obtained through the conducted numerical experiments.

\subsubsection{Numerical experiments} \label{sec:unsignednumexpnon}

 As in earlier subsection we will split the presentation of the numerical results in several parts. The numerical results that we will present below are obtained by running the SOCP from (\ref{eq:socpnon}). To demonstrate the precision of our technique we will in parallel show the results obtained by running (\ref{eq:mainlasso3vernon}). To make scaling simpler in all our numerical experiments we again set $\sigma=1$.

\textbf{\underline{\emph{1) $\frac{E\|\w_{dep+}\|_2}{\sigma}$ and $\frac{E\|\w_{socp+}\|_2}{\sigma}$ as functions of $x_{mag}^{(sc)}$}}}

In this part we will show the numerical results that correspond to the theoretical ones given in part 1) in the previous subsection. We will restrict our attention again only on the medium or $\alpha=0.5$ regime (in a later section we will show the results one can get for $\alpha=0.7$ regime). We then set all other parameters as in the center plot of Figure \ref{fig:errorvarxnon} (these parameters are of course different depending if we are considering $\rho=2$ or $\rho=3$; below we will consider both of them).

\underline{\emph{a) Low $(\alpha,\beta_w^+)$ regime, $\rho=2$}}

We first consider the $\rho=2$ scenario. As mentioned above in our experiments we set $\alpha=0.5$, $r_{socp+}=\sqrt{\frac{\alpha n}{1+\rho^2}}=\sqrt{0.2\alpha n}$, and (as shown in \cite{StojnicGenSocp10}) $\beta_w^+$ such that $(\alpha_w^+,\beta_w^+)$ satisfy (\ref{eq:fundl1non}) and $\alpha_w^+=\frac{\rho^2}{1+\rho^2}\alpha$. We then ran (\ref{eq:socpnon}) $300$ times with $n=600$ for various $x_{mag}^{(sc)}$. In parallel we ran (\ref{eq:mainlasso3ver}) for the exact same parameters with only one difference; namely we ran (\ref{eq:mainlasso3vernon}) with $n=2000$. The obtained results for $\frac{E\|\w_{socp+}\|_2}{\sigma}$ and $\frac{E\|\w_{dep+}\|_2}{\sigma}$ are shown on the left-hand and right-hand side of Figure \ref{fig:errorvarxsimnon}, respectively (given our assumption that $\sigma=1$ $\frac{E\|\w_{dep+}\|_2}{\sigma}$ and $\frac{E\|\w_{socp+}\|_2}{\sigma}$ are of course just $E\|\w_{dep+}\|_2$ and $E\|\w_{socp+}\|_2$, respectively). We also show in Figure \ref{fig:errorvarxsimnon} the corresponding theoretical predictions obtained in the previous subsection.
\begin{figure}[htb]
\begin{minipage}[b]{.5\linewidth}
\centering
\centerline{\epsfig{figure=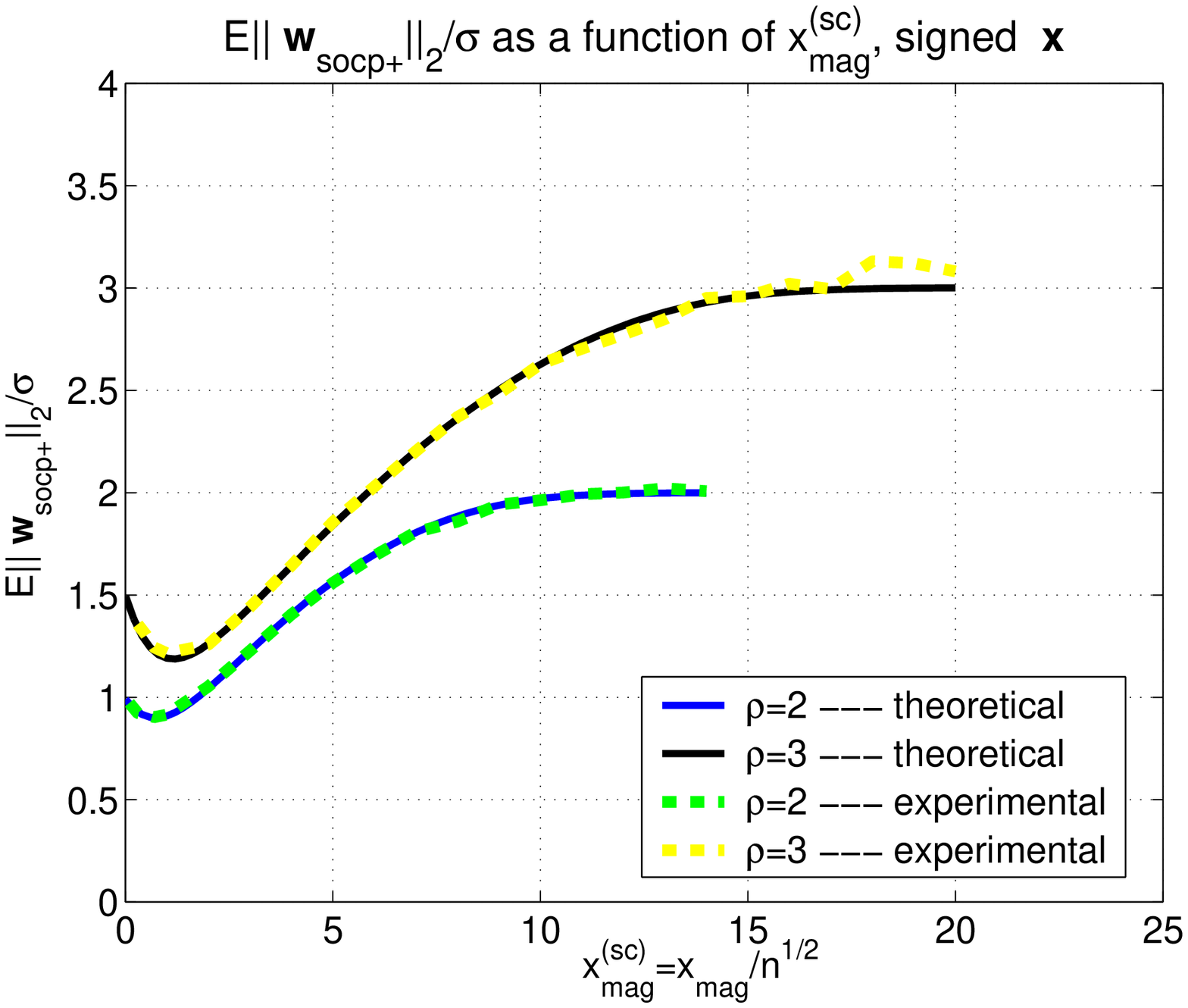,width=8cm,height=7cm}}
\end{minipage}
\begin{minipage}[b]{.5\linewidth}
\centering
\centerline{\epsfig{figure=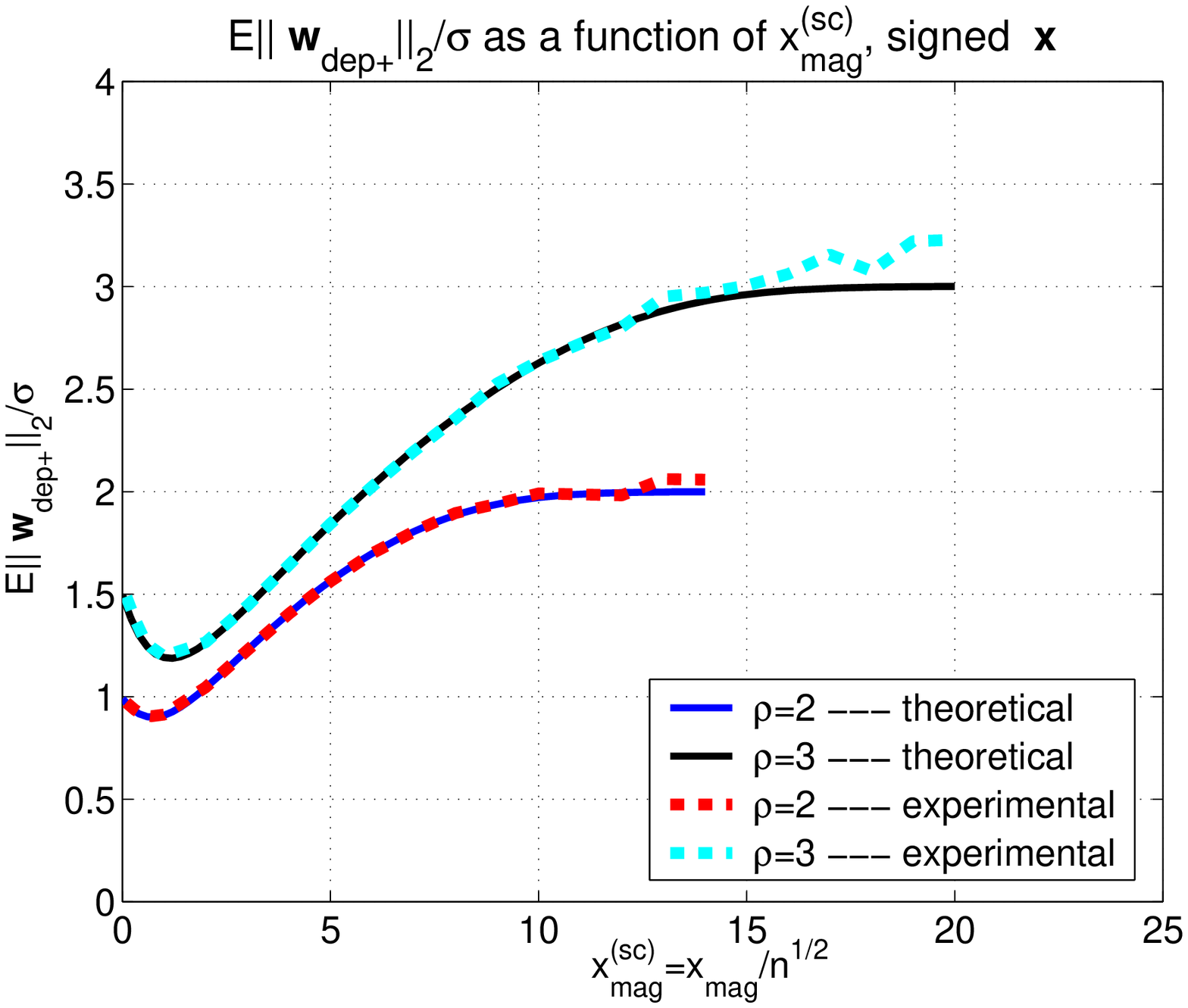,width=8cm,height=7cm}}
\end{minipage}
\caption{Experimental results for $\frac{E\|\w_{socp+}\|_2}{\sigma}$ and $\frac{E\|\w_{dep+}\|_2}{\sigma}$ as a function of $x_{mag}^{(sc)}$; $\rho=2$, $r_{socp+}=\sqrt{0.2 \alpha n}$; $\rho=3$, $r_{socp+}=\sqrt{0.1 \alpha n}$; left --- SOCP from (\ref{eq:socpnon}), right --- (\ref{eq:mainlasso3vernon})}
\label{fig:errorvarxsimnon}
\end{figure}

\underline{\emph{b) High $(\alpha,\beta_w^+)$ regime, $\rho=3$}}

We also conducted a set of experiments in the so-called ``high" $(\alpha,\beta_w^+)$ regime. We used exactly the same parameters as in low $(\alpha,\beta_w^+)$ except that we changed $\rho$ from $2$ to $3$. Consequently we chose $r_{socp+}=\sqrt{0.1 \alpha n}$ and $\beta_w^+$ such that $(\alpha_w^+,\beta_w^+)$ satisfy (\ref{eq:fundl1non}) and $\alpha_w^+=\frac{\rho^2}{1+\rho^2}\alpha$. As above we ran $300$ times each (\ref{eq:socpnon}) and (\ref{eq:mainlasso3vernon}). We ran (\ref{eq:socpnon}) with $n=600$ and (\ref{eq:mainlasso3vernon}) with $n=2000$. The numerical results obtained for $\rho=3$ together with the theoretical predictions are again shown in Figure \ref{fig:errorvarxsimnon}. From Figure \ref{fig:errorvarxsimnon} we observe a solid agreement between the theoretical predictions and the results obtained through numerical experiments.

\textbf{\underline{\emph{2) $\frac{Ef_{obj+}}{\sqrt{n}}$ and $\frac{E\xi_{prim+}^{(dep)}(\sigma,\g,\h,x_{mag},r_{socp+})}{\sqrt{n}}$ as functions of $x_{mag}^{(sc)}$}}}

In this part we will show the numerical results that correspond to the theoretical ones given in part 2) in the previous subsection. We then set all parameters as in Figure \ref{fig:errorvarxnon} (these parameters are exactly the same as in experiments whose results we just presented above). Of course we again distinguish two cases: $\rho=2$ and $\rho=3$. For both $\rho=2$ and $\rho=3$ we ran $300$ times each, (\ref{eq:socpnon}) and (\ref{eq:mainlasso3vernon}) and again we ran (\ref{eq:socpnon}) with $n=600$ and (\ref{eq:mainlasso3vernon}) with $n=2000$. The numerical results that we obtained for $\frac{Ef_{obj+}}{\sqrt{n}}$ and $\frac{E\xi_{prim+}^{(dep)}(\sigma,\g,\h,x_{mag},r_{socp+})}{\sqrt{n}}$ are shown in Figure \ref{fig:objvarxsimnon}.
\begin{figure}[htb]
\begin{minipage}[b]{.5\linewidth}
\centering
\centerline{\epsfig{figure=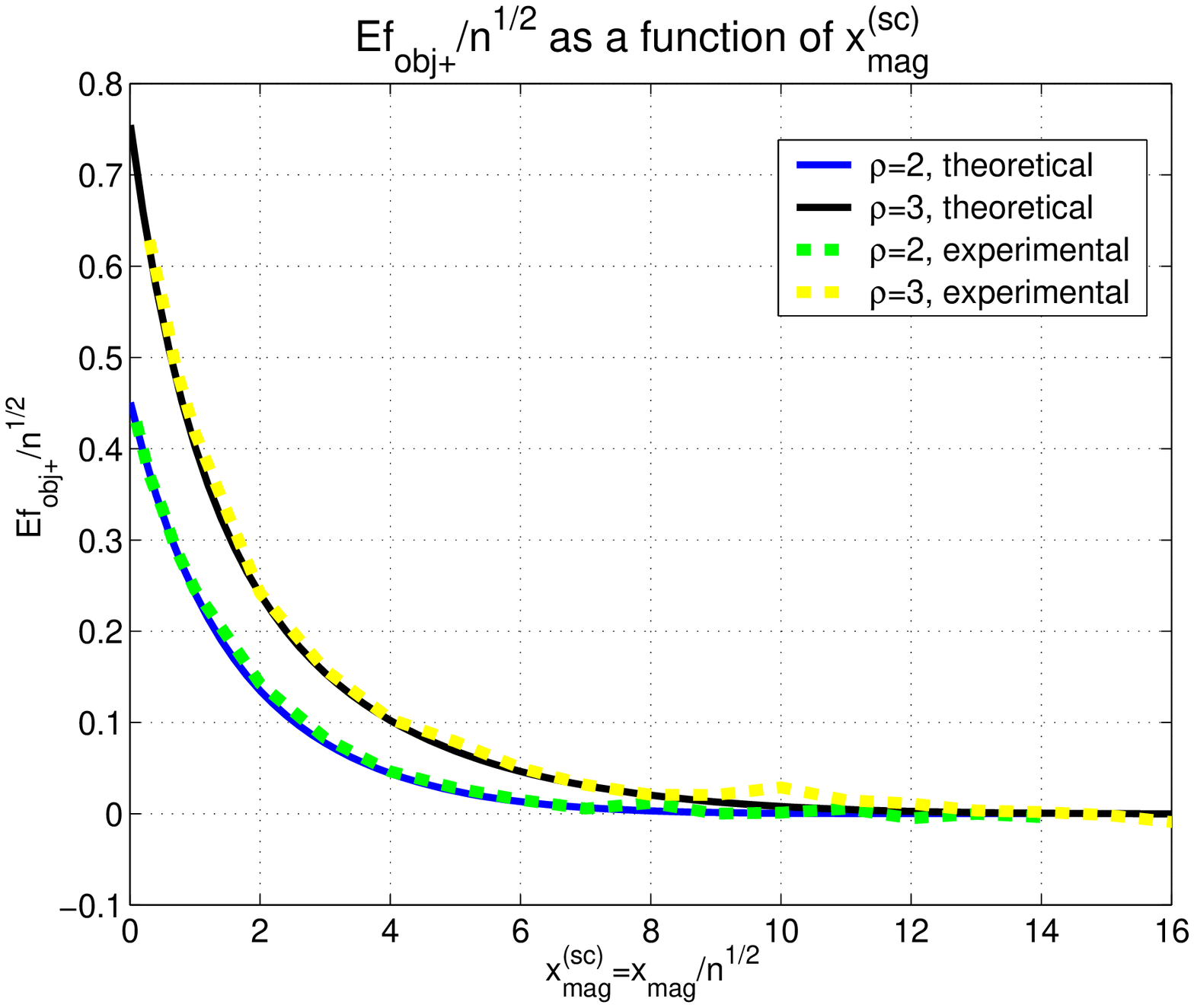,width=8cm,height=7cm}}
\end{minipage}
\begin{minipage}[b]{.5\linewidth}
\centering
\centerline{\epsfig{figure=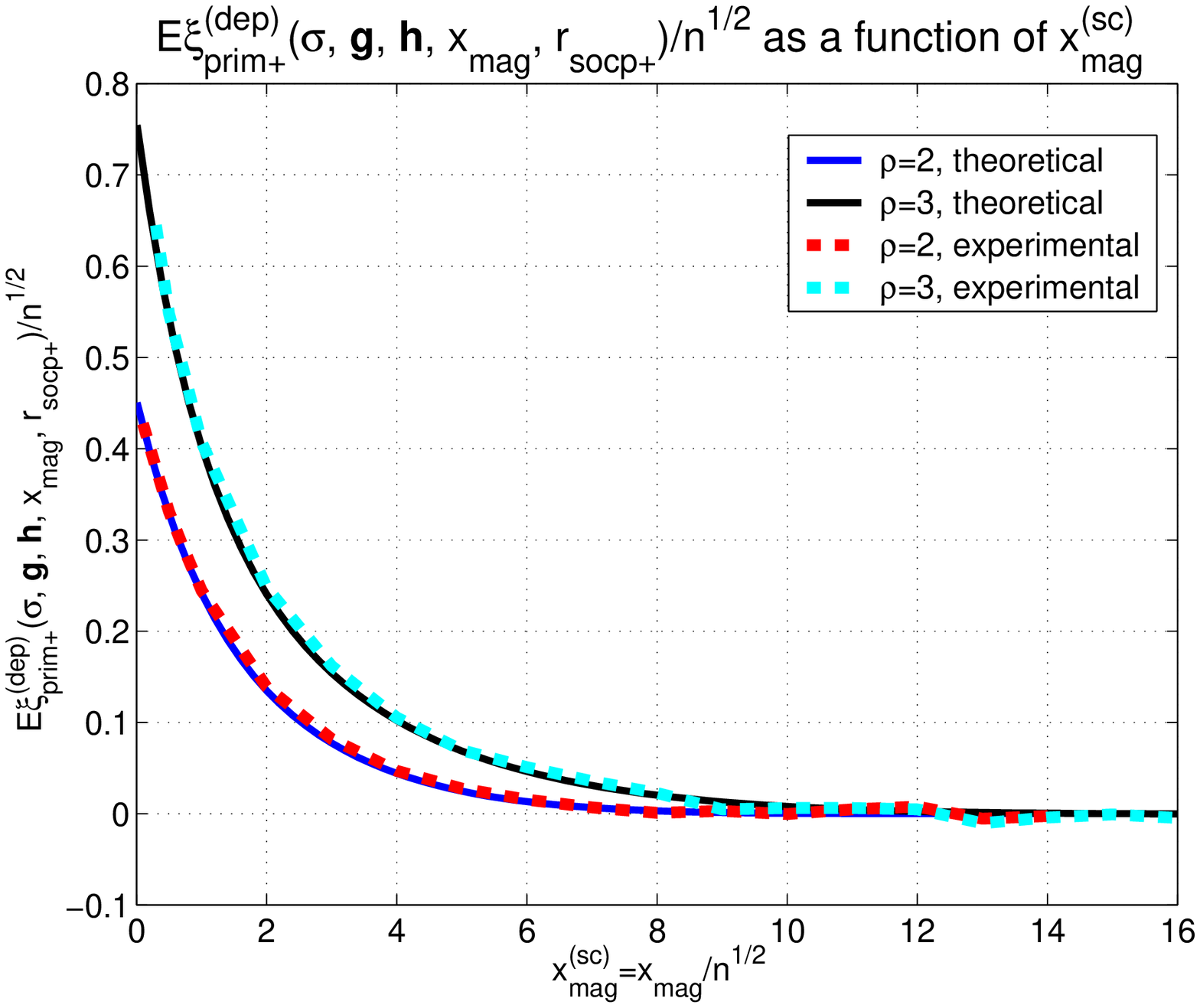,width=8cm,height=7cm}}
\end{minipage}
\caption{Experimental results for $\frac{Ef_{obj+}}{\sqrt{n}}$ and $\frac{E\xi_{prim+}^{(dep)}(\sigma,\g,\h,x_{mag},r_{socp+})}{\sqrt{n}}$ as a function of $x_{mag}^{(sc)}$; $\rho=2$, $r_{socp+}=\sqrt{0.2 \alpha n}$; $\rho=3$, $r_{socp+}=\sqrt{0.1 \alpha n}$; left --- SOCP from (\ref{eq:socpnon}); right --- (\ref{eq:mainlasso3vernon})}
\label{fig:objvarxsimnon}
\end{figure}
We again observe a solid agreement between the theoretical predictions and the results obtained through numerical experiments.

\textbf{\underline{\emph{3) $\frac{E\|\w_{dep+}\|_2}{\sigma}$ and $\frac{E\|\w_{socp+}\|_2}{\sigma}$ as functions of $x_{mag}^{(sc)}$; varying $r_{socp+}$}}}

In this part we will show the numerical results that correspond to the theoretical ones given in part 3) in the previous subsection. These results relate to possible variations in the $r_{socp+}$ that can be used in (\ref{eq:socpnon}). We then set all other parameters as in Figure \ref{fig:errorvarxvarrhonon} (these parameters are of course again different depending if we are considering $\rho=2$ or $\rho=3$).

\underline{\emph{a) Low $(\alpha,\beta_w^+)$ regime, $\rho=2$}}

We first consider the $\rho=2$ scenario. As in part 1) of this subsection we set $\alpha=0.5$ and choose $\beta_w^+$ as in part 1). However, differently from part 1) we now consider two different possibilities for $r_{socp+}$, namely $r_{socp+}=\sqrt{0.05\alpha n}$ and $r_{socp+}=\sqrt{0.6 \alpha n}$. We then ran (\ref{eq:socpnon}) $300$ times with $n=600$ for various $x_{mag}^{(sc)}$. In parallel we ran (\ref{eq:mainlasso3vernon}) with $n=2000$. The obtained results for $\frac{E\|\w_{socp+}\|_2}{\sigma}$ and $\frac{E\|\w_{dep+}\|_2}{\sigma}$ are shown on the left-hand and right-hand side of Figure \ref{fig:errorvarxvarrhosimnon}, respectively. We also show in Figure \ref{fig:errorvarxvarrhosimnon} the corresponding theoretical predictions obtained in the previous subsection.
\begin{figure}[htb]
\begin{minipage}[b]{.5\linewidth}
\centering
\centerline{\epsfig{figure=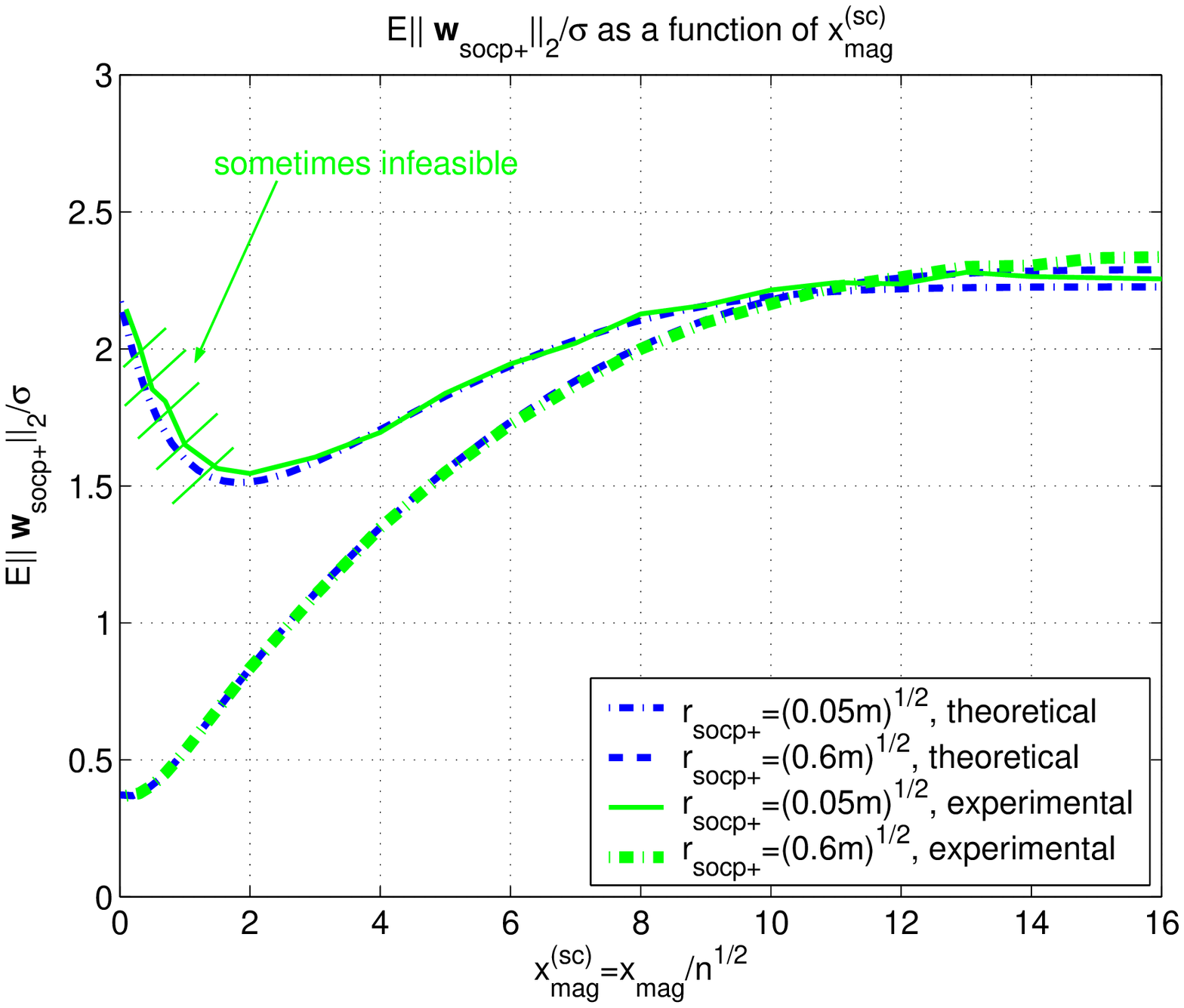,width=8cm,height=7cm}}
\end{minipage}
\begin{minipage}[b]{.5\linewidth}
\centering
\centerline{\epsfig{figure=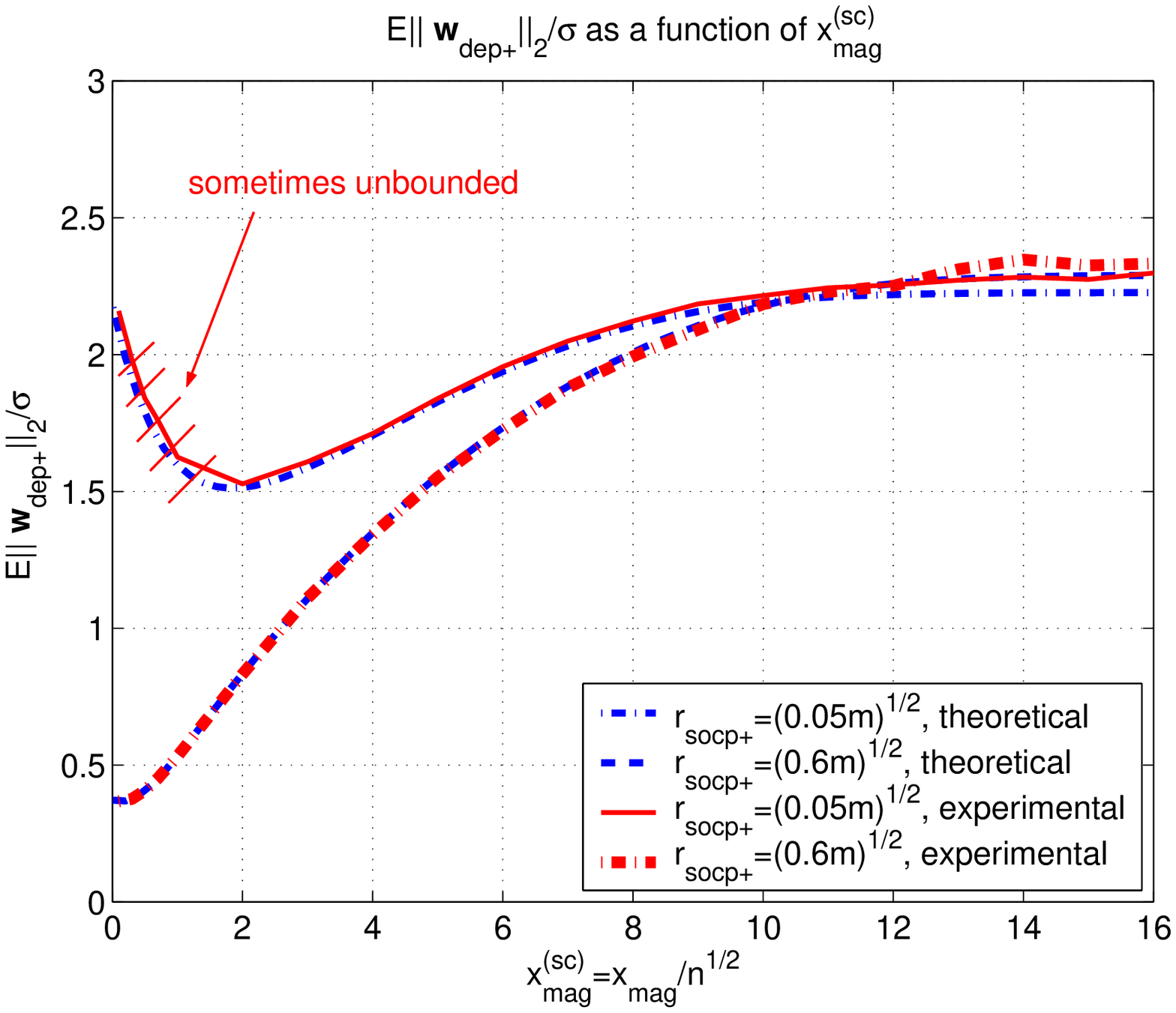,width=8cm,height=7cm}}
\end{minipage}
\caption{Experimental results for $\frac{E\|\w_{socp+}\|_2}{\sigma}$ and $\frac{E\|\w_{dep+}\|_2}{\sigma}$ as a function of $x_{mag}^{(sc)}$; $\rho=2$; $r_{socp+}\in\{\sqrt{0.05 \alpha n},\sqrt{0.6 \alpha n}\}$; left --- SOCP from (\ref{eq:socpnon}), right --- (\ref{eq:mainlasso3vernon})}
\label{fig:errorvarxvarrhosimnon}
\end{figure}

\underline{\emph{b) High $(\alpha,\beta_w^+)$ regime, $\rho=3$}}

We also consider the $\rho=3$ scenario. As above, we set $\alpha=0.5$ and choose $\beta_w^+$ as in part 1) of this subsection. Everything else remain the same as in $\rho=2$ case except the way we vary $r_{socp+}$. This time we consider (as in part 3) of the previous section when $\rho=3$ case was considered) $r_{socp+}=\sqrt{0.05\alpha n}$ and $r_{socp+}=\sqrt{0.5 \alpha n}$. As usual (\ref{eq:socpnon}) was run $300$ times with $n=600$ for various $x_{mag}^{(sc)}$. In parallel we ran (\ref{eq:mainlasso3vernon}) $300$ times with $n=2000$. The obtained numerical results for $\frac{Ef_{obj+}}{\sqrt{n}}$ and $\frac{E\xi_{prim+}^{(dep)}(\sigma,\g,\h,x_{mag},r_{socp+})}{\sqrt{n}}$ as well as the corresponding theoretical predictions obtained in the previous subsection are shown on the left-hand and right-hand side of Figure \ref{fig:errorvarxvarrhosim1non}, respectively.
\begin{figure}[htb]
\begin{minipage}[b]{.5\linewidth}
\centering
\centerline{\epsfig{figure=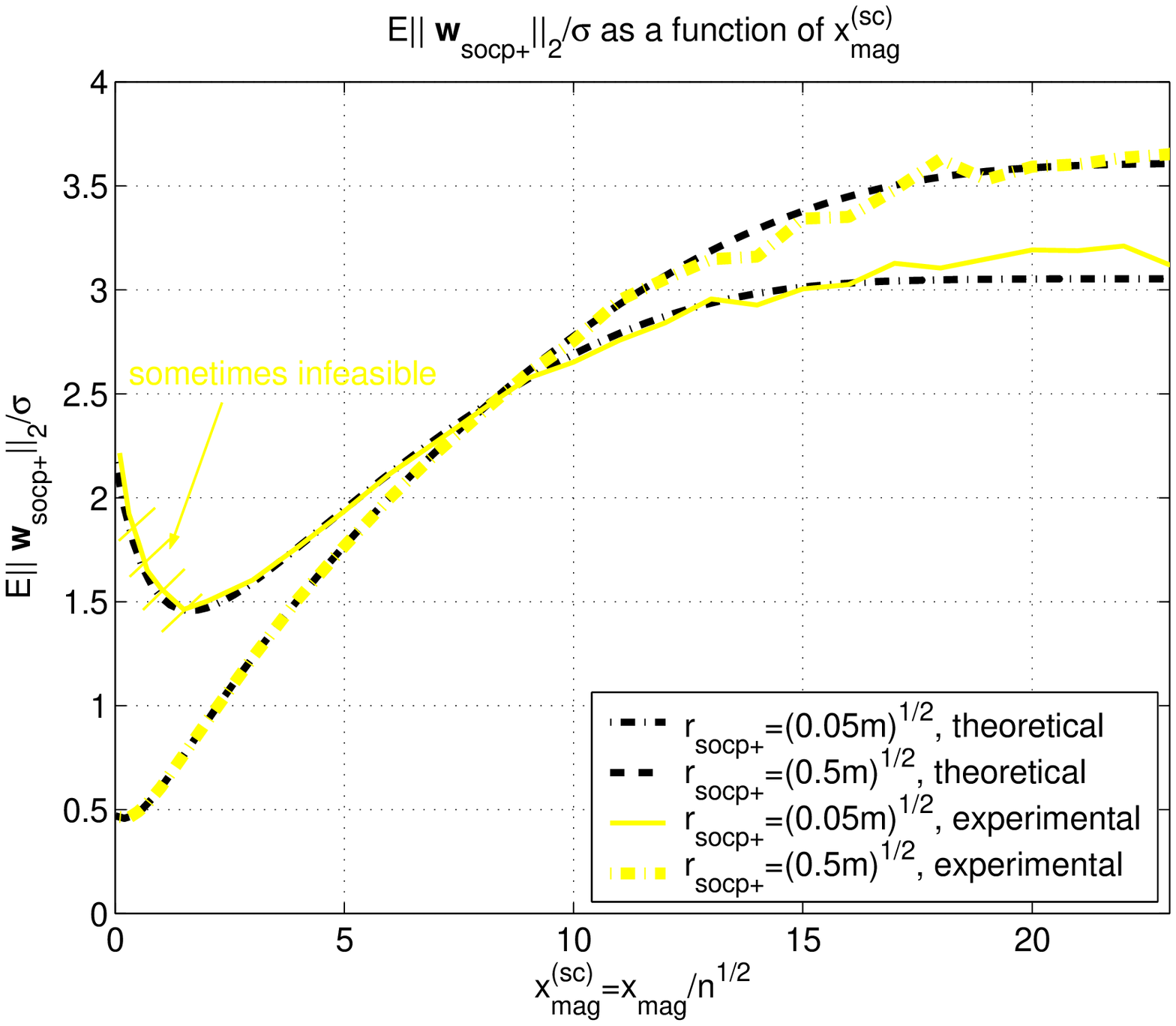,width=8cm,height=7cm}}
\end{minipage}
\begin{minipage}[b]{.5\linewidth}
\centering
\centerline{\epsfig{figure=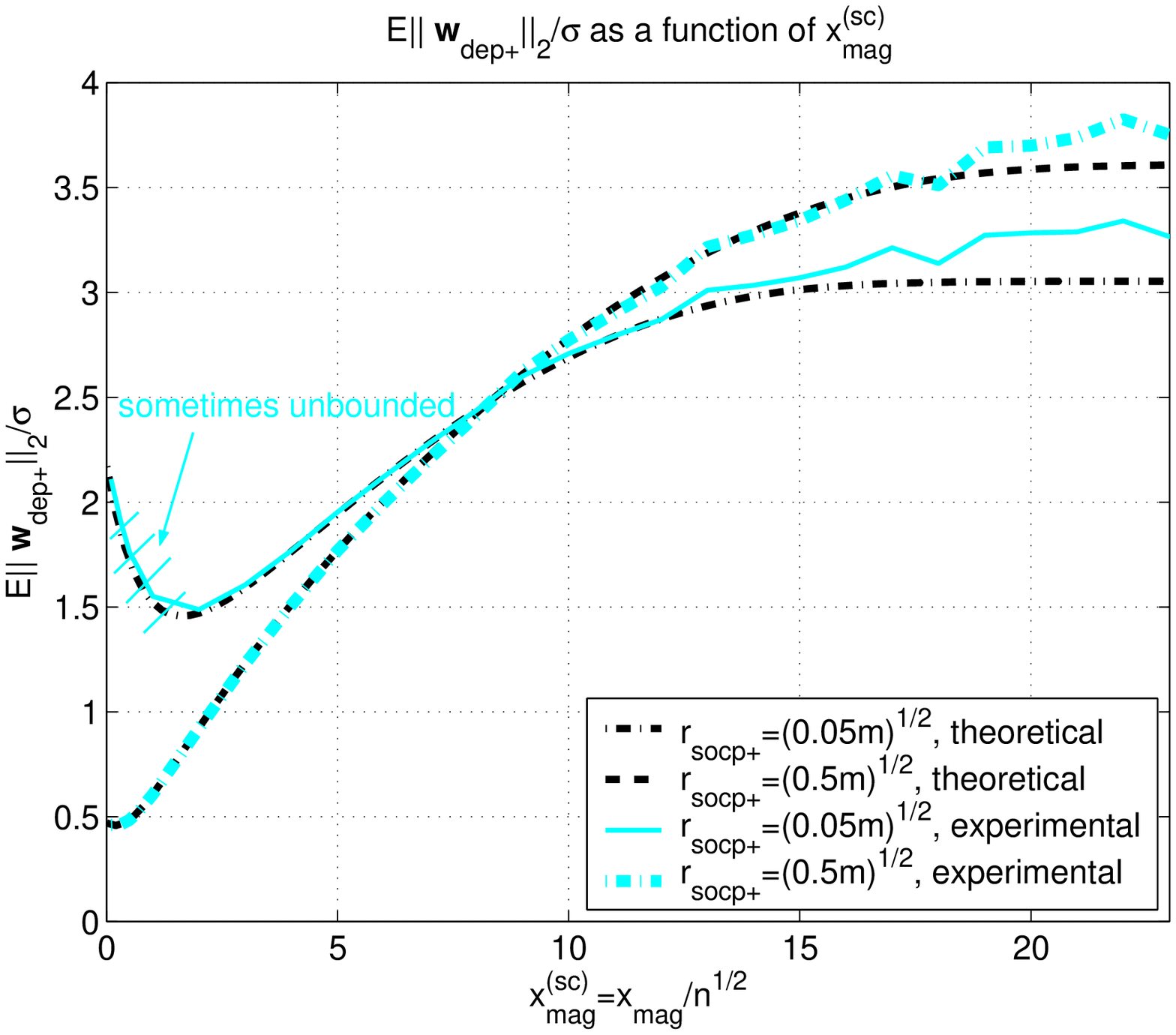,width=8cm,height=7cm}}
\end{minipage}
\caption{Experimental results for $\frac{E\|\w_{socp+}\|_2}{\sigma}$ and $\frac{E\|\w_{dep+}\|_2}{\sigma}$ as a function of $x_{mag}^{(sc)}$; $\rho=3$; $r_{socp+}\in\{\sqrt{0.05 \alpha n},\sqrt{0.5 \alpha n}\}$; left --- SOCP from (\ref{eq:socpnon}), right --- (\ref{eq:mainlasso3vernon})}
\label{fig:errorvarxvarrhosim1non}
\end{figure}
We again observe a solid agreement between the theoretical predictions and the results obtained through numerical experiments. As in Section \ref{sec:unsignednumexp} small glitches that happen in large $x_{mag}^{(sc)}$ regime could have been fixed by choosing a larger $n$. We again purposely chose a smaller $n$ to show that results are fairly good even when $n$ is not very large. In fact, even a smaller $n$ than the one that we have chosen would work quite fine.

Another observation related to Figures \ref{fig:errorvarxvarrhosimnon} and \ref{fig:errorvarxvarrhosim1non} (and several figures that will follow) is in place. For $r_{socp+}=\sqrt{0.05 \alpha n}$ and roughly speaking $x_{mag}^{(sc)}\leq 2$ it may happen that (\ref{eq:socpnon}) is on occasion infeasible and that (\ref{eq:mainlasso3vernon}) is unbounded. From a theoretical point of view this should not happen for any $x_{mag}^{(sc)}$. However, as discussed in Section \ref{sec:feasibilty}, $\alpha=0.5$ is in a sense a border line choice for universal feasibility. On the other hand, since all these claims are of ``with overwhelming probability" type it may sometimes happen that even when $\alpha=0.5$ (\ref{eq:socpnon}) is infeasible and (\ref{eq:mainlasso3vernon}) is unbounded. In Table \ref{tab:feasr005rho2} we show the number of our experiments for which everything worked fined, i.e. for which (\ref{eq:socpnon}) turned out to be feasible and (\ref{eq:mainlasso3vernon}) turned out to be bounded (we restrict only to what we call interesting region, which for this example we found to be roughly $x_{mag}^{(sc)}\leq 2$). We refer to such a number as the number of successes. The results in $r_{socp+}=\sqrt{0.05\alpha n}$ regime shown in Figures \ref{fig:errorvarxvarrhosimnon} and \ref{fig:errorvarxvarrhosim1non} are averaged over the feasible instances of (\ref{eq:socpnon}) and the bounded instances of (\ref{eq:mainlasso3vernon}).

\begin{table}
\caption{Experimental results for the noisy recovery through SOCP; $r_{socp+}=\sqrt{0.05m}$, $\sigma=1$; (\ref{eq:socpnon}) was run $300$ times with $n=600$; (\ref{eq:mainlasso3vernon}) was run $300$ times with $n=2000$}\vspace{.1in}
\hspace{-0in}\centering
\begin{tabular}{||c|c|c|c|c|c|c|c|c||}\hline\hline
& &  $x_{mag}^{(sc)}$ & $0.1$ & $0.3$ & $0.5$ & $0.7$ & $1$ & $2$ \\ \hline\hline
$\rho=2$ & $r_{socp+}=\sqrt{0.05 n}$ &  $\# $ of successes (\ref{eq:socpnon})  &  $269$ & $274$ & $286$ & $294$ & $296$ & $300$  \\ \hline
$\rho=2$ & $r_{socp+}=\sqrt{0.05 n}$ &  $\# $ of successes (\ref{eq:mainlasso3vernon}) & $268$ & $278$ & $281$ & $287$ &  $299$ & $299$  \\ \hline
$\rho=3$ & $r_{socp+}=\sqrt{0.05 n}$ &  $\# $ of successes (\ref{eq:socpnon})  &  $277$ & $287$ & $292$ & $296$ & $299$ & $300$  \\ \hline
$\rho=3$ & $r_{socp+}=\sqrt{0.05 n}$ &  $\# $ of successes (\ref{eq:mainlasso3vernon}) & $270$ & $274$ & $288$ & $294$ &  $299$ & $299$  \\ \hline\hline
\end{tabular}
\label{tab:feasr005rho2}
\end{table}

\textbf{\underline{\emph{4) $\frac{Ef_{obj+}}{\sqrt{n}}$ and $\frac{E\xi_{prim+}^{(dep)}(\sigma,\g,\h,x_{mag},r_{socp+})}{\sqrt{n}}$ as functions of $x_{mag}^{(sc)}$; varying $r_{socp+}$}}}

In this part we will show the numerical results that correspond to the theoretical ones given in part 4) in the previous subsection. These results relate to behavior of $\frac{Ef_{obj+}}{\sqrt{n}}$ and $\frac{E\xi_{prim+}^{(dep)}(\sigma,\g,\h,x_{mag},r_{socp+})}{\sqrt{n}}$ when one varies $r_{socp+}$ in (\ref{eq:socpnon}). We again consider $\rho=2$ or $\rho=3$. The observations made above that relate to the occasional feasibilities do apply to the results presented in this part as well.

\underline{\emph{a) Low $(\alpha,\beta_w^+)$ regime, $\rho=2$}}

The setup that we consider is exactly the same as the one considered in part 3a) of this subsection. We set $\alpha=0.5$, chose $\beta_w^+$ as in part 1), and considered two different possibilities for $r_{socp+}$, namely $r_{socp+}=\sqrt{0.05\alpha n}$ and $r_{socp+}=\sqrt{0.6 \alpha n}$.  The obtained results for $\frac{Ef_{obj+}}{\sqrt{n}}$ and $\frac{E\xi_{prim+}^{(dep)}(\sigma,\g,\h,x_{mag},r_{socp+})}{\sqrt{n}}$ are shown on the left-hand and right-hand side of Figure \ref{fig:objvarxvarrhosimnon}, respectively. The corresponding theoretical predictions obtained in the previous subsection are also shown in Figure \ref{fig:objvarxvarrhosimnon}.
\begin{figure}[htb]
\begin{minipage}[b]{.5\linewidth}
\centering
\centerline{\epsfig{figure=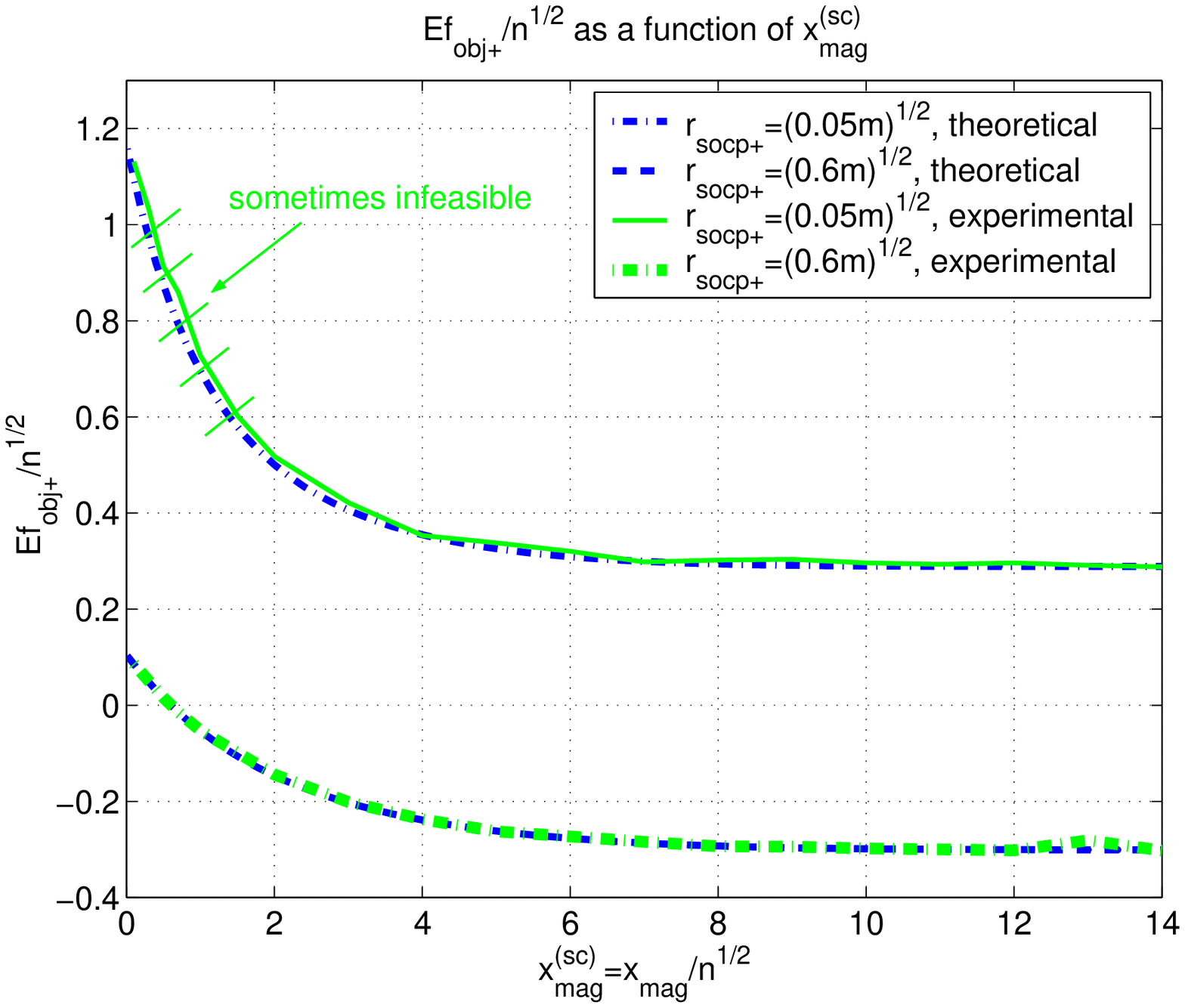,width=8cm,height=7cm}}
\end{minipage}
\begin{minipage}[b]{.5\linewidth}
\centering
\centerline{\epsfig{figure=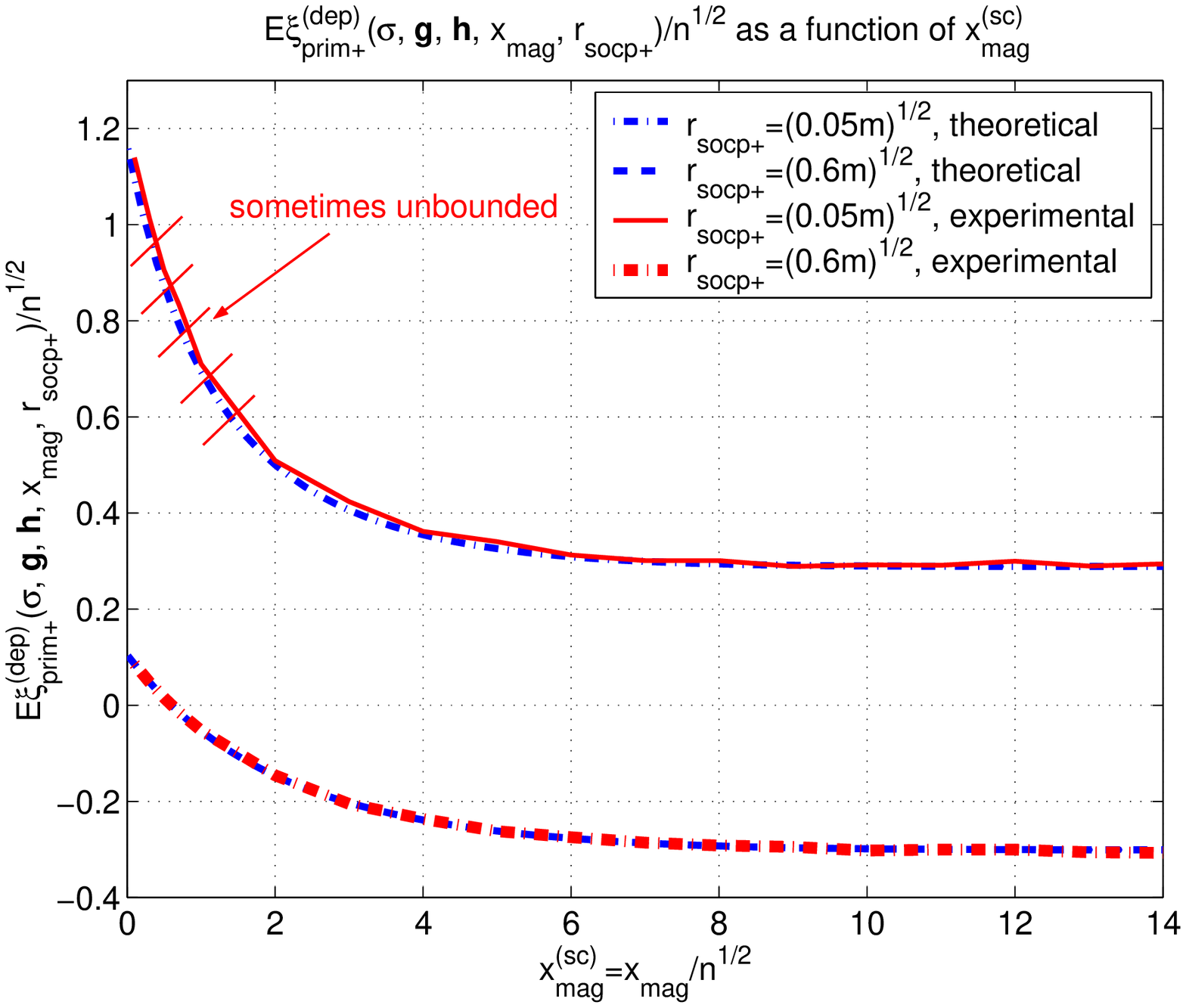,width=8cm,height=7cm}}
\end{minipage}
\caption{Experimental results for $\frac{Ef_{obj+}}{\sqrt{n}}$ and $\frac{E\xi_{prim+}^{(dep)}(\sigma,\g,\h,x_{mag},r_{socp+})}{\sqrt{n}}$ as a function of $x_{mag}^{(sc)}$; $\rho=2$; $r_{socp+}\in\{\sqrt{0.05 \alpha n},\sqrt{0.6 \alpha n}\}$; left --- SOCP from (\ref{eq:socpnon}), right --- (\ref{eq:mainlasso3vernon})}
\label{fig:objvarxvarrhosimnon}
\end{figure}

\underline{\emph{b) High $(\alpha,\beta_w^+)$ regime, $\rho=3$}}

The setup that we consider is exactly the same as the one considered in part 3b) of this subsection. We set $\alpha=0.5$, chose $\beta_w^+$ as in part 1), and considered two different possibilities for $r_{socp+}$, namely $r_{socp+}=\sqrt{0.05\alpha n}$ and $r_{socp+}=\sqrt{0.5 \alpha n}$.  The obtained results for $\frac{E\|\w_{socp+}\|_2}{\sigma}$ and $\frac{E\|\w_{dep+}\|_2}{\sigma}$ are shown on the left-hand and right-hand side of Figure \ref{fig:objvarxvarrhosim1non}, respectively. The corresponding theoretical predictions obtained in the previous subsection are also shown in Figure \ref{fig:objvarxvarrhosim1non}.
\begin{figure}[htb]
\begin{minipage}[b]{.5\linewidth}
\centering
\centerline{\epsfig{figure=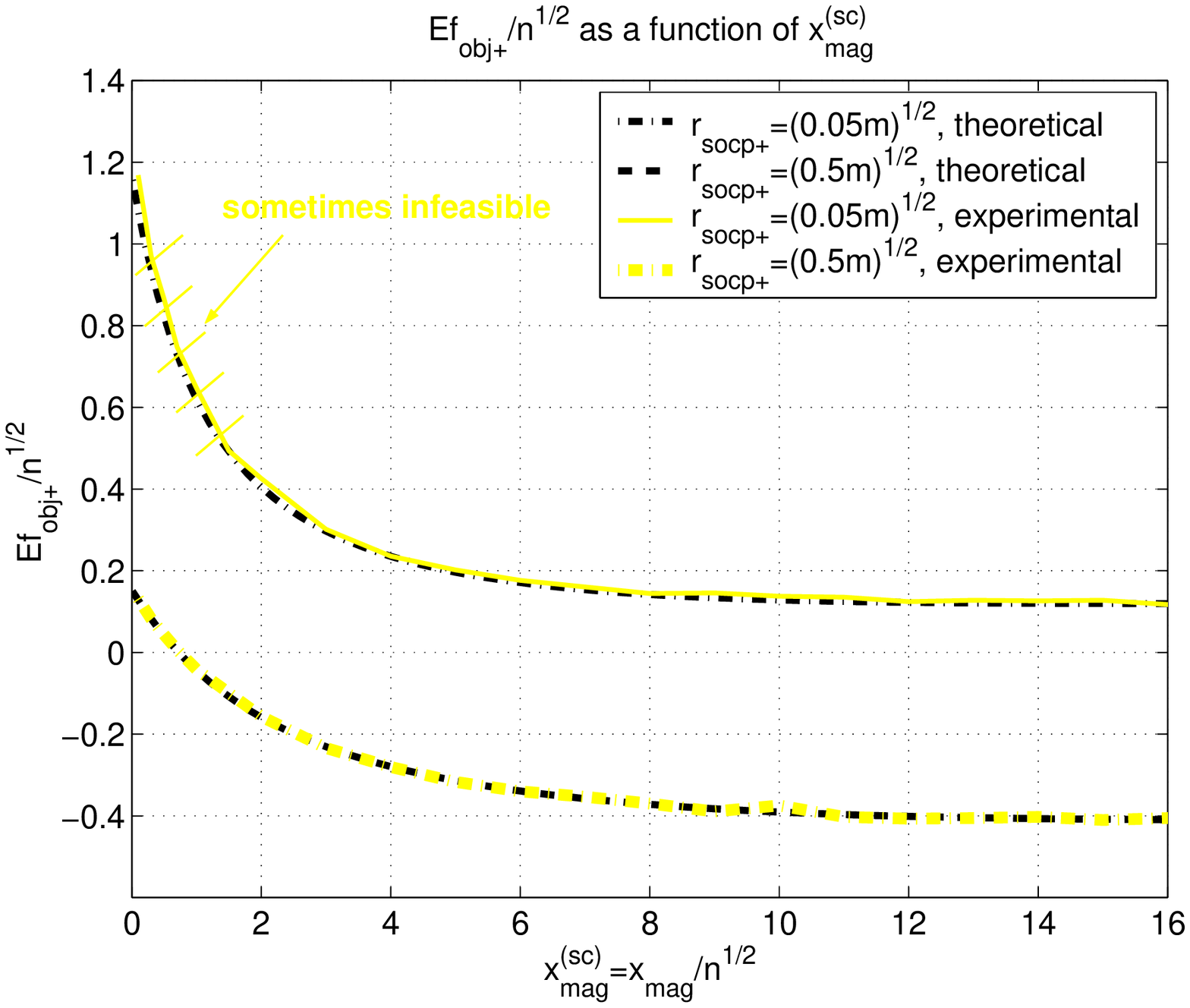,width=8cm,height=7cm}}
\end{minipage}
\begin{minipage}[b]{.5\linewidth}
\centering
\centerline{\epsfig{figure=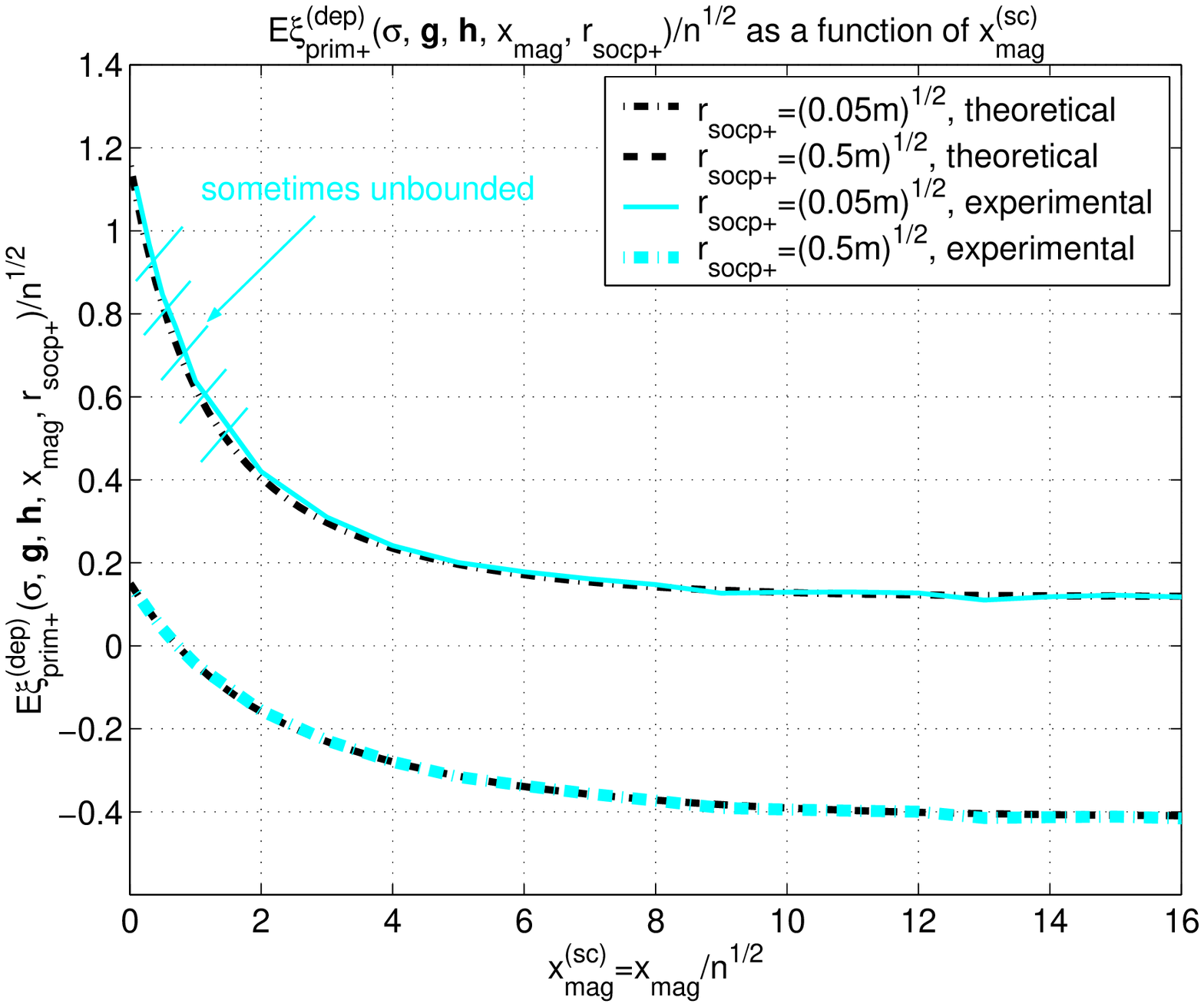,width=8cm,height=7cm}}
\end{minipage}
\caption{Experimental results for $\frac{Ef_{obj+}}{\sqrt{n}}$ and $\frac{E\xi_{prim+}^{(dep)}(\sigma,\g,\h,x_{mag},r_{socp+})}{\sqrt{n}}$ as a function of $x_{mag}^{(sc)}$; $\rho=3$; $r_{socp+}\in\{\sqrt{0.05 \alpha n},\sqrt{0.5 \alpha n}\}$; left --- SOCP from (\ref{eq:socpnon}), right --- (\ref{eq:mainlasso3vernon})}
\label{fig:objvarxvarrhosim1non}
\end{figure}
We again observe a solid agreement between the theoretical predictions and the results obtained through numerical experiments.

\subsubsection{Numerical experiments --- feasibility}
\label{sec:numexpfeas}

In this section we will present a couple of numerical results that relate to the feasibility of (\ref{eq:socpnon}) or unboundedness of (\ref{eq:mainlasso3vernon}).

\textbf{\underline{\emph{1) $\frac{E\|\w_{dep+}\|_2}{\sigma}$ and $\frac{E\|\w_{socp+}\|_2}{\sigma}$ as functions of $x_{mag}^{(sc)}$}}}

In this part we will show the numerical results that correspond to the theoretical ones given in part 1) in the previous subsection. We will restrict our attention to $\alpha=0.7$ regime (we recall that earlier in Section \ref{sec:unsignednumexpnon} we showed the corresponding results one can get for $\alpha=0.5$ regime). We then set all other parameters as in the plot on the right hand side of Figure \ref{fig:errorvarxnon} (these parameters are of course different depending if we are considering $\rho=2$ or $\rho=3$; below we will consider both of them).

\underline{\emph{a) Low $(\alpha,\beta_w^+)$ regime, $\rho=2$}}

We first consider the $\rho=2$ scenario. We set $\alpha=0.7$, $r_{socp+}=\sqrt{\frac{\alpha n}{1+\rho^2}}=\sqrt{0.2\alpha n}$, and (as shown in \cite{StojnicGenSocp10}) $\beta_w^+$ such that $(\alpha_w^+,\beta_w^+)$ satisfy (\ref{eq:fundl1non}) and $\alpha_w^+=\frac{\rho^2}{1+\rho^2}\alpha$. We then ran (\ref{eq:socpnon}) $100$ times with $n=800$ for various $x_{mag}^{(sc)}$. In parallel we ran (\ref{eq:mainlasso3vernon}) for the exact same parameters with only two differences; namely we ran (\ref{eq:mainlasso3vernon}) $200$ times with $n=4000$. The obtained results for $\frac{E\|\w_{socp+}\|_2}{\sigma}$ and $\frac{E\|\w_{dep+}\|_2}{\sigma}$ are shown on the left-hand and right-hand side of Figure \ref{fig:errorvarxsimnon}. We also show in Figure \ref{fig:errorvarxsimnon} the corresponding theoretical predictions obtained earlier.
\begin{figure}[htb]
\begin{minipage}[b]{.5\linewidth}
\centering
\centerline{\epsfig{figure=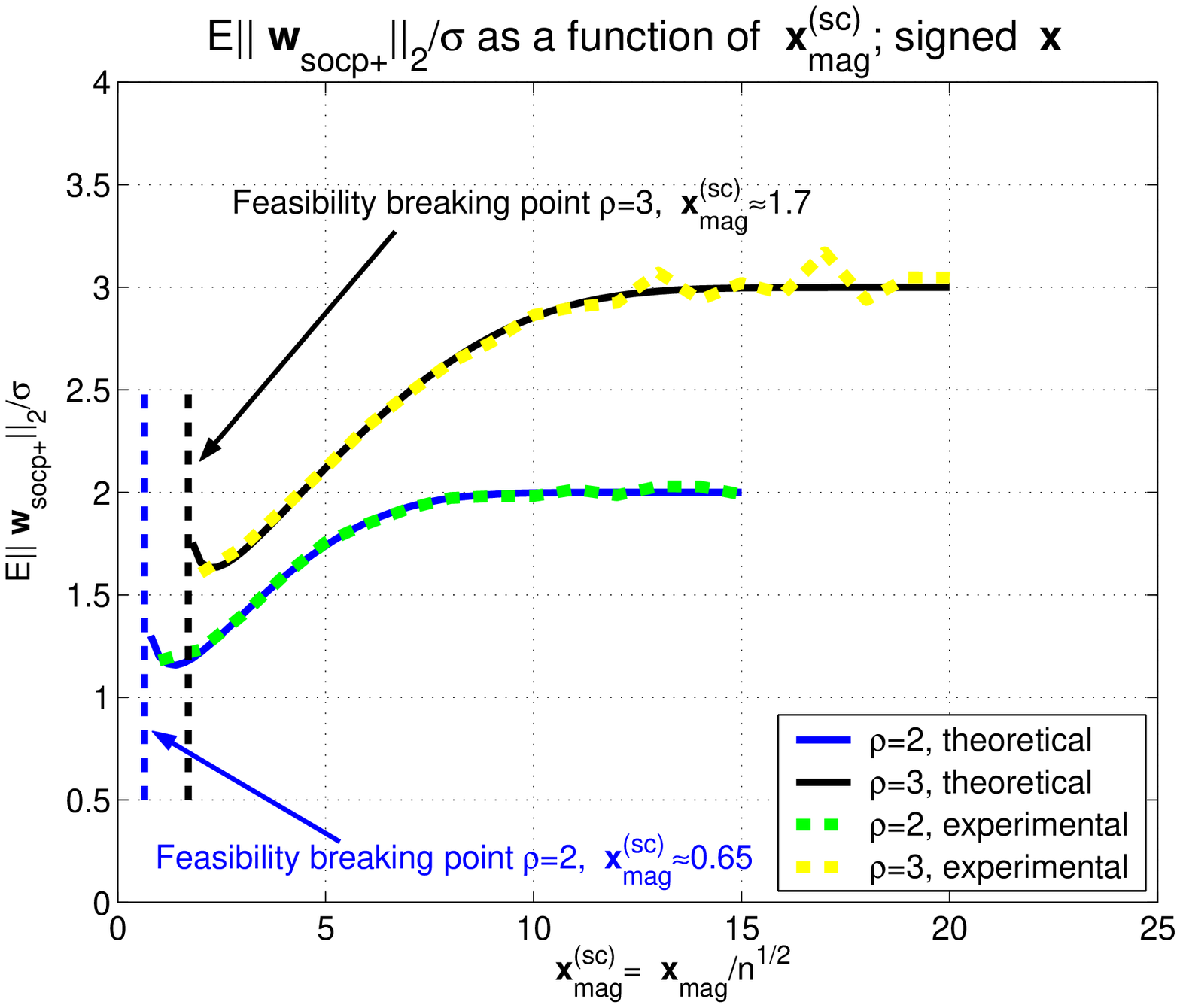,width=8cm,height=7cm}}
\end{minipage}
\begin{minipage}[b]{.5\linewidth}
\centering
\centerline{\epsfig{figure=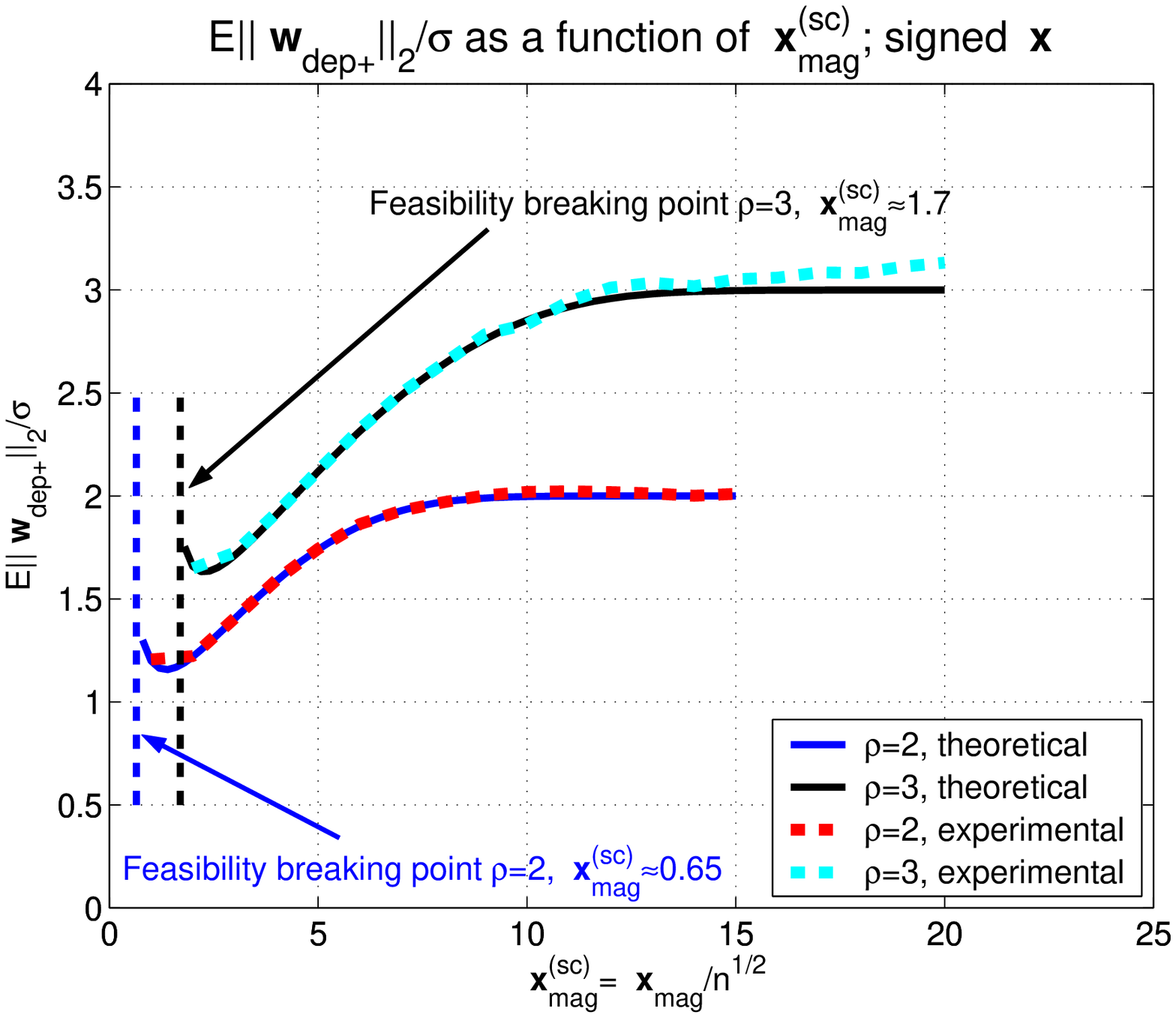,width=8cm,height=7cm}}
\end{minipage}
\caption{Experimental results for $\frac{E\|\w_{socp+}\|_2}{\sigma}$ and $\frac{E\|\w_{dep+}\|_2}{\sigma}$ as a function of $x_{mag}^{(sc)}$; $\alpha=0.7$; $\rho=2$, $r_{socp+}=\sqrt{0.2 \alpha n}$; $\rho=3$, $r_{socp+}=\sqrt{0.1 \alpha n}$; left --- SOCP from (\ref{eq:socpnon}), right --- (\ref{eq:mainlasso3vernon})}
\label{fig:errorvarxsimnon}
\end{figure}

\underline{\emph{b) High $(\alpha,\beta_w^+)$ regime, $\rho=3$}}

We also conducted a set of experiments in the so-called ``high" $(\alpha,\beta_w^+)$ regime. We used exactly the same parameters as in low $(\alpha,\beta_w^+)$ except that we changed $\rho$ from $2$ to $3$. Consequently we chose $r_{socp+}=\sqrt{0.1 \alpha n}$ and $\beta_w^+$ such that $(\alpha_w^+,\beta_w^+)$ satisfy (\ref{eq:fundl1non}) and $\alpha_w^+=\frac{\rho^2}{1+\rho^2}\alpha$. As above we ran $100$ times (\ref{eq:socpnon}) and $200$ times (\ref{eq:mainlasso3vernon}). Also as above, we ran (\ref{eq:socpnon}) with $n=800$ and (\ref{eq:mainlasso3vernon}) with $n=4000$. The numerical results obtained for $\rho=3$ together with the theoretical predictions are again shown in Figure \ref{fig:errorvarxsimnon}. From Figure \ref{fig:errorvarxsimnon} we observe a solid agreement between the theoretical predictions and the results obtained through numerical experiments.

\noindent \textbf{Remark:} We do mention that in a range of $x_{mag}^{(sc)}$ close to the theoretical ``breaking feasibility point" not all instances of (\ref{eq:socpnon}) were feasible and not all instances of (\ref{eq:mainlasso3vernon}) were bounded (instead an overwhelming majority of them was). The results that are shown in Figure \ref{fig:errorvarxsimnon} are obtained by averaging over the feasible instances of (\ref{eq:socpnon}) and the bounded instances of (\ref{eq:mainlasso3vernon}).

\textbf{\underline{\emph{2) Feasibility/boundedness probability}}}

In this part we show how the number of feasible instances of (\ref{eq:socpnon}) and the number of bounded instances of (\ref{eq:mainlasso3vernon}) change as $x_{mag}^{(sc)}$ changes. We set $\alpha=0.7$ and considered the $\rho=3$ regime or in other words ``high" $(\alpha,\beta_w^+)$ regime. As above we chose $r_{socp+}=\sqrt{0.1 \alpha n}$ and $\beta_w^+$ such that $(\alpha_w^+,\beta_w^+)$ satisfy (\ref{eq:fundl1non}) and $\alpha_w^+=\frac{\rho^2}{1+\rho^2}\alpha$. We then ran each of (\ref{eq:socpnon}) and (\ref{eq:mainlasso3vernon}) $200$ times. We ran (\ref{eq:socpnon}) with $n=1000$ and (\ref{eq:mainlasso3vernon}) with $n=10000$. For a range of $x_{mag}^{(sc)}$ we recorded the fraction of instances where (\ref{eq:socpnon}) was feasible. We refer to such a fraction as $p_{socp+}$. Simultaneously, for the same range of $x_{mag}^{(sc)}$ we recorded the fraction of instances where (\ref{eq:mainlasso3vernon}) was bounded. We refer to such a fraction as $p_{prim+}$. The numerical results along with the theoretical prediction for the feasibility breaking point are shown in Figure \ref{fig:probfeasbound}. From Figure \ref{fig:probfeasbound} we observe a solid agreement between the theoretical predictions and the results obtained through numerical experiments.
\begin{figure}[htb]
\begin{minipage}[b]{1\linewidth}
\centering
\centerline{\epsfig{figure=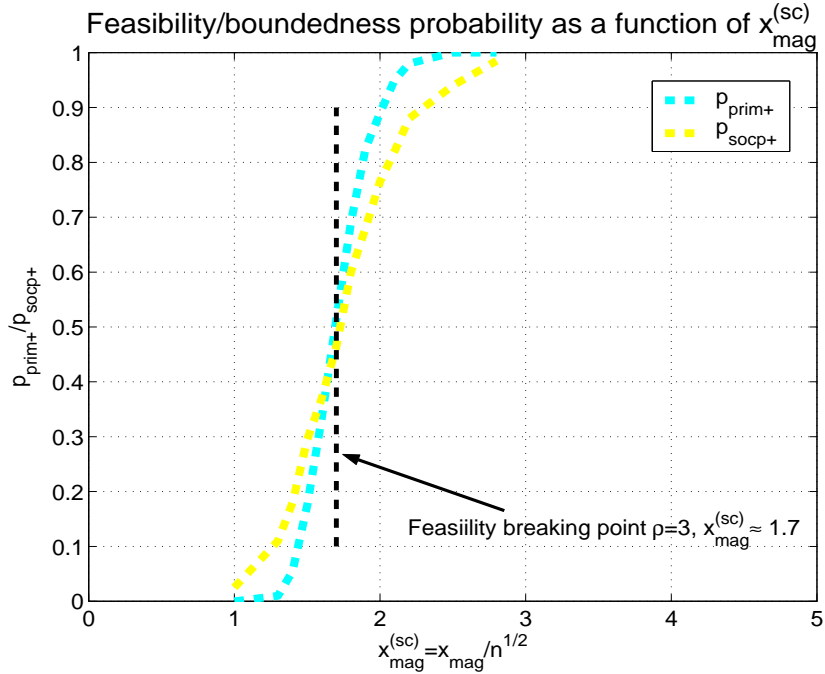,width=11cm,height=9cm}}
\end{minipage}
\caption{$p_{socp+}$ and $p_{prim+}$ as functions of $x_{mag}^{(sc)}$; $\rho=3$; $\alpha=0.7$; $r_{socp+}=\sqrt{\frac{\alpha n}{1+\rho^2}}$}
\label{fig:probfeasbound}
\end{figure}

\section{Discussion}
\label{sec:discuss}

In this paper we considered ``noisy" under-determined systems of linear equations with sparse solutions.
We looked from a theoretical point of view at polynomial-time second-order cone programming (SOCP) algorithms.
More precisely, we looked at a general framework developed for characterization of such algorithms in \cite{StojnicGenSocp10}. Within the framework we then considered what we referred to as the SOCP's \emph{problem dependent} performance. We established the precise values of the norm-2 of the error vector for a wide class of unknown sparse vectors. We also provided a characterization of several important parameters that appear in solving of an SOCP.

Many further developments are possible (one can make essentially the same conclusion for the general framework developed in \cite{StojnicGenSocp10} as well as for the analysis of the LASSO algorithms presented in \cite{StojnicGenLasso10}). Any problem dependent scenario that can be solved in the so-called noiseless case through the mechanisms developed in \cite{StojnicCSetam09} and \cite{StojnicUpper10} can now be handled in the noisy case as well. For example, quantifying performance of SOCP or LASSO optimization problems in solving ``noisy" systems with special structure of the solution vector (block-sparse, box- and binary-sparse, partially known locations of nonzero components, low-rank matrix, just to name a few), ``noisy" systems with noisy (or approximately sparse)) solution vectors can then easily be handled to an ultimate precision. In several forthcoming papers we will present some of these applications.

\begin{singlespace}
\bibliographystyle{plain}
\bibliography{PrDepSocp}
\end{singlespace}

\end{document}